\documentclass[journal]{IEEEtran} 

\usepackage{cite} 

\usepackage{graphicx} 

\usepackage{amsmath,amsfonts,amssymb}
\usepackage{algorithmic}      

\usepackage{array}      
 
\usepackage[caption=false,font=footnotesize]{subfig} 

\usepackage{stfloats}  

\usepackage{url} 
\usepackage{balance}   
\newtheorem{theorem}{Theorem}
\newtheorem{lemma}{Lemma}

\newtheorem{proposition}{Proposition}
\newtheorem{definition}{Definition}
\newtheorem{remark}{Remark} 
 
\usepackage{dsfont}

\usepackage{makecell}
\usepackage{threeparttable} 
\usepackage{multirow} 
\usepackage[table]{xcolor}  
\usepackage{hhline}     

\usepackage{verbatim} 
\allowdisplaybreaks

\definecolor{mycolor}{RGB}{111, 78, 55} 
\definecolor{mycolor2}{RGB}{0,0,139}

\begin{document}
\title{Improved Characterization of the Memory-Rate Tradeoff
	for Demand-Private Coded Caching With Multiple Demands}

\author{Qinyi~Lu, Nan~Liu, Wei~Kang, and  Chunguo~Li%
\thanks{  
Part of this work has been accepted for presentation at the 2026 IEEE Global Communications Conference (GLOBECOM).  
Q. Lu, W. Kang and C. Li are with the School of Information Science
and Engineering, Southeast University, Nanjing 211189, China.
(e-mail: \{qylu, wkang, chunguoli\}@seu.edu.cn).
N. Liu is with the National Mobile Communications Research Laboratory,
Southeast University, Nanjing 211189, China.
(e-mail: nanliu@seu.edu.cn). }    
}    
	\maketitle 
	
\begin{abstract}
	We consider a coded caching problem with multiple demands under a 
privacy constraint. In this problem, a server with access to 
	$N$ files serves $K$ users over a shared link, and each user requests $L$ distinct files. 
	The privacy constraint requires that each user obtain no information about the demands 
	of the other users.   
	 We first propose a new achievable scheme for arbitrary $N$, $K$, and $L$ by applying a novel transformation to a suitable non-private coded caching scheme.  	
	We then derive a  new converse bound, and show that the proposed scheme is order optimal within a  multiplicative  factor of $6$ of this bound.  
 Finally, for the case of two users, we completely characterize the optimal memory-rate tradeoff   for arbitrary $N$ and $L$ through three new achievable schemes and matching converse bounds. 
\end{abstract}
  
\section{Introduction}  
Coded caching, formulated by Maddah-Ali and Niesen in~\cite{MaddahAli2014}, has received significant attention.         
In this model, a server accessing a library of $N$ files serves $K$ users over an error-free shared link,  with each user equipped with a cache of normalized size $M$.  
The system consists of two phases: a placement phase, where each user fills its cache as a function of the $N$ files, and a delivery phase. 
In the delivery phase, each user requests one file, and the server broadcasts a message to ensure that each user can  decode its requested file based on its own cache content and the broadcast message.  
A fundamental problem is to characterize the tradeoff between the cache size and the worst-case rate of the broadcast message, i.e., the memory-rate tradeoff.  Under uncoded placement, the optimal memory-rate tradeoff is characterized and is achieved by the scheme proposed by Yu \emph{et al.}, referred to as the \emph{YMA scheme}~\cite{Qian2018}.  
For  arbitrary placement strategies, the converse bound in~\cite{Qian2019}, together with the YMA scheme, establishes order optimality within a factor of $2.00884$.

In practical scenarios such as a \emph{FemtoCaching} network\cite{Golrezaei2013}, a cache-enabled small-cell base station may serve multiple users, so that multiple demands are aggregated at the same cache.  
 This motivates the coded caching problem with multiple demands, where each user requests $L\geq 1$ files during the delivery phase~\cite{Ji2014,Ji2015}.  
This problem was further studied in~\cite{Sengupta2017}, where an improved converse bound was derived and used to show that the scheme that treats $L$ sets of demands independently, each consisting of one demand from each user, is order optimal within a factor of $11$.     
The optimal memory-rate tradeoff under a fixed symmetric uncoded placement was studied in~\cite{Wei2017}, while the scenario with heterogeneous numbers of demands across users was considered in~\cite{Huang2020}.  
As pointed out in~\cite{Parrinello2020}, coded caching with multiple demands is closely related to the shared-cache scenario, since a user's $L$ demands can be viewed, at a high level, as demands from multiple users associated with the same helper cache.    
The coded caching problem with shared caches has also been investigated in several works, including~\cite{Xu2019,Ibrahim2019,Peter2022,Das2023}.  
  
 However, schemes designed for the above coded caching problems may leak information about users' demands to other users, thereby raising a fundamental privacy issue.   
Demand privacy in coded caching was first considered in \cite{Engelmann2017}.
Since then, the problem has been extensively studied from an information-theoretic perspective in \cite{Chinmay2022,Aravind2020,Yan2021,Namboodiri2021,Wan2021,Gholami2023,Lu2026}. 
For the case where each user requests one file, three achievable schemes were proposed in \cite{Chinmay2022}, and the best performance among them is order optimal within a constant factor.  
One of these schemes is obtained by using the idea of virtual users to transform a scheme originally designed for the coded caching problem without demand privacy into a demand-private scheme. Moreover, the optimal memory-rate tradeoff for the special case $K=2$ was completely characterized in \cite{Chinmay2022}, and further characterizations for several parameter regimes were obtained in \cite{Lu2026}.    
Compared with the privacy constraint considered in \cite{Chinmay2022}, Aravind \emph{et al.} studied a weaker privacy constraint and focused on reducing the subpacketization level~\cite{Aravind2020}. 
In \cite{Yan2021}, demand privacy against colluding users was investigated  for both Single-File Retrieval (SFR) and Linear Function Retrieval (LFR), where each user requests a single file and a linear combination of the files, respectively. 
The authors proposed a scheme based on privacy key and derived a converse bound for the SFR scenario, thereby establishing order optimality for both scenarios.    
In \cite{Namboodiri2021}, the optimal memory-rate tradeoff for the SFR scenario was characterized for small cache sizes under the same privacy constraint.
In~\cite{Gholami2023}, Gholami \emph{et al.} considered the coded caching problem with simultaneously private demands and caches, and proposed a new construction that generates demand-private coded caching schemes by leveraging private information retrieval schemes~\cite{Sun2017PIR}.

	The study of demand-private coded caching with multiple demands is motivated by practical applications in which users may request multiple files while not wishing to reveal their demands to other users.   
 In \cite{Wan2021}, Wan \emph{et al.} investigated this problem, in which each of the \(K\) users requests \(L\)  distinct  files out of \(N\) files.        
They proposed two achievable schemes based on virtual users and MDS codes, referred to as the \emph{virtual-user scheme} and the \emph{MDS-based scheme}, respectively. 
For \(L>1\), they further established order optimality within a factor of
\(22\) when \(N \leq LK\), or when \(N > LK\) and  
	\(M \geq \frac{N}{K}\). In addition, the optimal memory-rate tradeoff in the memory regime  $ M \ge \min \big\{ \frac{2K-1}{2K}, \frac{2^{K-1}}{2^{K-1}+1} \big\} N$    was characterized  in \cite[Theorem~6]{Wan2021}.   
Despite these advances, the memory-rate tradeoff of this problem is still not fully understood. 
More specifically, order optimality has not been established in the regime $N>LK$ and $M<\frac{N}{K}$, while the optimal memory-rate tradeoff 
is known only for the range of large cache sizes specified above.  
   
In this paper, we further investigate the demand-private coded caching problem with multiple demands.  
We begin by proposing a new achievable scheme and proving its order optimality within a constant factor for arbitrary values of $N$, $K$, and  $L$.    
In the corresponding problem without demand privacy, the users' $L$ demands can be treated independently while preserving order optimality~\cite{Sengupta2017}.    
Under demand privacy, however, such independent treatment may fail to preserve privacy, since a broadcast message that reveals no information about the other users' demands on its own may reveal such information when considered jointly with the other messages. 
In addition, we establish the
optimal memory-rate tradeoff for the case of two users over all feasible
values of \(N\) and \(L\), providing further insight into the fundamental limits of the problem.

\subsection{Contributions} 
    
	The main contributions of this paper are summarized as follows.          
	 \begin{enumerate} 
	 	\item 
	 We propose a new achievable scheme for arbitrary $N$, $K$, and $L$ by establishing a novel transformation from non-private schemes to demand-private coded caching schemes with multiple demands. More specifically, the non-private schemes used in the transformation are schemes for an $(N,K\bar N,1)$ restricted-demand-subset non-private coded caching problem under uncoded placement, referred to as the $(N,K\bar N,1)$ \emph{RD-NP-UCC problem}, where $\bar N=\min\{N,KL\}$.     
	 	It consists of $N$ files and $K\bar N$ users, each requesting one file, while the demand vectors are restricted to a subset of all possible demand vectors.    	 	
The transformation builds on the virtual user idea in~\cite{Chinmay2022}, but requires new techniques to handle multiple demands per user.
Each user randomly selects $L$ virtual users and caches the union of their cache contents, while suitable randomizations of the file labels and demand vectors in the RD-NP-UCC problem are used to ensure demand privacy.         
     Applying this transformation to the YMA scheme~\cite{Qian2018} yields the proposed scheme.    
	 	\item 
	 	We generalize the converse bound in~\cite[Theorem~2]{Qian2019}, originally derived for the case where each user requests one file, to coded caching with multiple demands. This yields a new converse bound for arbitrary $N$, $K$, and $L$, which applies both with and without demand privacy.

	 	\item
 Using the above converse bound, we show that our achievable scheme for arbitrary $N$, $K$, and $L$ is order optimal within a multiplicative factor of $6$.  
	 This improves upon the result in~\cite{Wan2021} by covering the parameter regimes left open in that work and by reducing the best previously known factor for \(L>1\) from \(22\) to \(6\).

	 	\item  
	 	 In the case of two users, i.e., $K=2$,  we completely characterize the optimal memory-rate tradeoff.  	    
	 	\begin{itemize}   
	 	 \item 
	On the achievability side, we propose three new schemes for the case $K=2$, whose achieved memory-rate pairs serve as corner points across different regimes of $N$ and $L$.    
	While each file is precoded with a randomly chosen Cauchy matrix to generate coded segments, the key innovation lies in the joint design of the cache contents and broadcast messages based on these segments.  
	More specifically, we carefully choose which coded segments are placed in the caches and delivered in the broadcast messages, and whether they appear directly or through linear combinations. These choices create the side information and delivery structure needed for decoding and, together with suitable randomization in placement and delivery, preserve demand privacy.      
	
\item 
On the converse side, beyond the general converse bound described above, we derive four additional bounds tailored to the case $K=2$.   
These bounds are derived from carefully chosen combinations of selected cache contents and broadcast messages corresponding to appropriate demand vectors, together with the equalities following from demand privacy.  
The resulting converse boundary matches the achievable boundary for all parameter regimes of $N$ and $L$, thereby establishing the optimal memory-rate tradeoff.   
	 	\end{itemize}     
	 \end{enumerate}

\subsection{Paper Organization}  
The rest of this paper is organized as follows. Section~\ref{sec_Model} formally defines the demand-private coded caching problem with multiple demands.
Section~\ref{sec_main} presents the main results and sketches the key ideas of the  proofs.  The detailed proofs are provided in  Sections~\ref{sec_theorem_achA}--\ref{sec_optimalK2Con}, followed by the conclusion of this paper in Section~\ref{sec_Conclu}.

\subsection{Notations}   
We define $[a:b] = \{a, a+1, \dots, b\}$ and $[a]= [0:a-1]$. 
For any finite set $\mathcal A$, $X_{\mathcal A}  =  (X_i)_{i\in\mathcal A}$ denotes the ordered tuple indexed by the elements of $\mathcal A$ in lexicographical order; 
	when all $X_i$, $i\in\mathcal A$, are column vectors of the same dimension,
	$(X_i)_{i\in\mathcal A}$ is understood as their horizontal concatenation in this order; 
	$[A_{i,j}]_{i\in[a],\,j\in[b]}$ denotes the matrix in
	$\mathbb F_q^{a\times b}$ whose $(i,j)$-th entry is $A_{i,j}$;    
$\mathrm{Unif}(\mathcal A)$ denotes the discrete uniform distribution over $\mathcal A$;  
$\mathrm{Perm}(\mathcal A)$ denotes the set of all permutations of $\mathcal A$,   
and $\mathbf 1\{x\in\mathcal A\}$ denotes the indicator function that equals $1$ if $x\in\mathcal A$ and $0$ otherwise.      
 For $i\in[a]$,  let $\mathbf e_i^{a}\in\mathbb F_q^{a}$ denote  the $i$-th standard basis vector of $\mathbb F_q^{a}$,  i.e.,   the vector whose $i$-th  component is $1$ and all others are $0$. For a vector $\mathbf v=(v_0,v_1,  \dots,v_{b-1})^{\mathsf T}  \in[a]^b$,  we define   
$ 
\mathbf I_a(\mathbf v,:)  = 	\big[ \mathbf e_{v_0}^{a}, \mathbf e_{v_1}^{a},\dots,   \mathbf e_{v_{b-1}}^{a} 	\big]^{\mathsf T} 
$, which follows the MATLAB-style row-selection syntax.

\section{System Model}   \label{sec_Model}  
Let $N$, $K$, and $L$ be positive integers.
We consider an $(N,K,L)$ demand-private coded caching problem defined as follows.  
The system contains a server with $N$ files and $K$ cache-aided users,  where the server and the users are connected through an error-free shared link.    
Let $q$ be a sufficiently large prime number. 
The $N$ files are denoted by $W_0, W_1, \ldots, W_{N-1}$, where each file $W_n \in \mathbb{F}_q^F$ is a column vector of $F$ symbols.  These files are independently and uniformly distributed over $\mathbb{F}_q^F$. 
We denote the entire file library by $\mathbf{W} = [W_0, W_1, \ldots, W_{N-1}]^{\mathsf T}$. 
As in prior works on coded caching (e.g., \cite{MaddahAli2014}), the file size $F$ is assumed to be sufficiently large to allow arbitrary subpacketization.  
The system operates in two phases: the placement phase and the delivery phase.

\subsubsection*{Placement Phase}  
Each user $k\in[K]$ is equipped with a cache of size $MF$ symbols over $\mathbb{F}_q$, 
where $M$ is the cache size normalized by the file size $F$.   
The server fills the users’ caches prior to knowing their demands.        
To ensure demand privacy, the server provides each user $k$ with a random variable $S_k$ over the alphabet $\mathcal{S}_k$. 
Let $S =  (S_0, S_1, \dots, S_{K-1})$ be the collection of these  random variables.      	    	
The server also holds a random variable $P$ over the alphabet $\mathcal{P}$,      
which is known only to the server.  
The cache content of user \(k\) is given by
\begin{align*}
	Z_k=\bigl(\phi_k (S_k,P,\mathbf{W}),\,S_k\bigr), \quad  k\in[K],  
\end{align*} 
where the \emph{caching function} \(\phi_k\) is defined as 
\begin{align*}  
	\phi_k:\mathcal{S}_k\times\mathcal{P}\times\mathbb{F}_q^{NF}
	\to \mathbb{F}_q^{MF}.
\end{align*}  
Throughout this paper, logarithms are in base $q$. For sufficiently large $F$, the alphabet size of $S_k$  is negligible relative to $MF$ and is ignored. The cache size constraint is       
\begin{align}   \label{cons_C} 
	\frac{H \left(Z_k\right) }{ F } =  \frac{H \left(	\phi_k (S_k, P, \mathbf{W}) \right) }{ F }   \le M, \quad  k \in [K]. 
\end{align}

\subsubsection*{Delivery Phase}    
Each user requests $L$ distinct files, where $L \le N$. For each  $k \in [K]$, let $\mathbf{d}_k = (d_{k,0}, d_{k,1}, \ldots, d_{k,L-1})^{\mathsf T}$ denote the demand vector of user $k$.   
Collecting all users' demand vectors, the demand matrix is $\mathbf{D} = [\mathbf{d}_0, \mathbf{d}_1, \ldots, \mathbf{d}_{K-1}]^{\mathsf T}$. 
Furthermore, let $ \mathbf{D}_{\setminus k} = [ \mathbf{d}_0,   \ldots,  \mathbf{d}_{k-1}, \mathbf{d}_{k+1},   \ldots,  \mathbf{d}_{K-1} ]^{\mathsf{T}} $  denote the demand vectors of all users except user $k$.  The set of all possible demand matrices is given by   
\begin{align*} 
	\mathcal{D} =  \{ \mathbf{D} :  d_{k,l} \in [N], d_{k,l} \neq d_{k,l'}, \forall k \in [K], l, l' \in [L], l \neq l' \}. 
\end{align*}	 

We assume that the demands, the file library, the random variable held by the server, and the random variables provided to the users are mutually independent, i.e.,   
\begin{align*} 
	& H \left( \mathbf{W}, \mathbf{D},  P,  S\right) \nonumber \\
	&   =  NF  + \sum_{k\in[K]} H ( \mathbf{d}_k) +  H \left(P \right) +  \sum_{k\in[K]} H \left( S_k \right).     
\end{align*}   
Upon receiving the demand matrix $\mathbf{D}\in\mathcal{D}$, the server generates a broadcast message  
\begin{align*}
	X_{\mathbf D}=\bigl(\psi(\mathbf D,P,S,\mathbf W),\,\mu(\mathbf D,P,S)\bigr),
\end{align*}
according to the \emph{encoding functions} 
\begin{align*}
	& \psi: \mathcal{D} \times \mathcal{P} \times \mathcal{S}_0 \times \cdots \times \mathcal{S}_{K-1} \times \mathbb{F}_q^{NF}
	\to \mathbb{F}_q^{RF}, \\
	& \mu: \mathcal{D} \times \mathcal{P} \times \mathcal{S}_0 \times \cdots \times \mathcal{S}_{K-1}
	\to \mathcal{J}, 
\end{align*}
where $R$ is the rate of the shared link, and the size of $\mathcal{J}$ is negligible relative to the file size $F$. The rate constraint is  
\begin{align} \label{cons_X}
	\frac{H(X_{\mathbf D})}{F}
	=
	\frac{H(\psi(\mathbf D,P,S,\mathbf W))}{F}
	\le R,\quad \forall\,\mathbf D\in\mathcal D.
\end{align} 	      
For $k \in [K]$, user $k$ should be able to decode its requested files, while learning no information about the demands of the other users beyond its own demand $\mathbf{d}_k$, i.e.,     
\begin{align} 
	\text{[Correctness]}  \quad   & 	H\bigl(\{W_{d_{k,l}}:l\in[L]\} \big| X_{\mathbf D}, \mathbf d_k, Z_k\bigr) = 0, \nonumber \\ 
	& \qquad \qquad \qquad   \forall k\in[K], \label{correctness}\\
	\text{[Privacy]}  \quad    &   I\bigl(\mathbf D_{\setminus k}; X_{\mathbf D}, Z_k, \mathbf d_k\bigr) = 0, \quad \forall k\in[K]. \label{privacy} 
\end{align}   
   
\begin{remark}
	Since $Z_k$ is independent of $\mathbf D$ and $\mathbf d_k$ is independent
	of $\mathbf D_{\setminus k}$, we have 
	$I\big(\mathbf D_{\setminus k};Z_k,\mathbf d_k\big)=0$.
	Thus, the privacy constraint in~\eqref{privacy} is equivalently  to 
	$I\bigl(\mathbf D_{\setminus k};X_{\mathbf D} \big|
	Z_k,\mathbf d_k\bigr)=0$.
\end{remark}

The pair $(M,R)$ is said to be achievable if there exist caching and encoding functions  satisfying \eqref{correctness} and \eqref{privacy} with cache size $M$ and rate $R$ for sufficiently large file size $F$. We define the optimal memory-rate tradeoff as   
\begin{align*}  
	R_{N,K,L}^{*{\mathrm{p}}}(M) =  \inf\left\{R:(M,R)\text{ is achievable}\right\},    
\end{align*}        
where the infimum is taken over all achievable rates $R$.     
 
\section{Main Results}  \label{sec_main}

This section presents the main results of the paper. We begin with the following theorem, which provides a new achievability result for arbitrary numbers of files and users.   
\begin{theorem} \label{theorem_achA} 
	For the $(N,K,L)$ demand-private coded caching problem, the lower convex envelope of the following pairs   
	\begin{align*}
		(M,R) =  \left( \frac{\binom{K \bar{N}}{r} - \binom{K \bar{N} - L}{r}}{\binom{K \bar{N}}{r} } N   ,  \frac{\binom{K \bar{N}}{r+1} -  \binom{(K-1) \bar{N}}{r+1}}{\binom{K \bar{N}}{r}  } \right),  
	\end{align*}  
	where $r \in [0:K\bar{N}] $ and  $\bar{N} =  \min \{N, KL\}$  is achievable.  	   
\end{theorem}    
\begin{IEEEproof}[Sketch of Proof] 
	We propose a new achievable scheme by transforming a scheme for the $(N,K\bar N,1)$ RD-NP-UCC problem into a scheme for the $(N,K,L)$ demand-private coded caching problem.    
	In the RD-NP-UCC problem, there are $K\bar N$ virtual users, and each real user requesting $L$ files is associated with a group of $\bar N$ virtual users. Each real user randomly selects $L$ virtual users from its  associated group and caches the contents corresponding to these virtual users, where a random permutation is applied to the file labels.   
	The demand vector for the $K\bar N$ virtual users is then generated so that it belongs to the restricted demand subset in which only $\bar N$ distinct files are requested.    
	Under the same random permutation of the file labels, this demand vector is constructed  by placing the $L$ requested files of each real user at the positions corresponding to the previously selected virtual users and suitably randomizing the remaining positions.   
	The server then delivers the broadcast message generated by the scheme for the RD-NP-UCC problem under this demand vector.    
	Finally, applying this transformation to the YMA scheme~\cite{Qian2018} establishes the achievability result stated in Theorem~\ref{theorem_achA}.
\end{IEEEproof}    
\begin{remark}
	The needed subpacketization level of the above scheme is 	$\binom{K\bar N}{r}$, with the maximal subpacketization level
	exponential in $K \cdot \min\{N,KL\}$.   
 By comparison, the maximal
	subpacketization levels of the virtual-user scheme  and the MDS-based scheme 
	in~\cite{Wan2021} are exponential in $K\binom{N}{L}$ and $K$,
	respectively. Thus, the scheme in Theorem~\ref{theorem_achA} can significantly reduce the maximal subpacketization level compared with the virtual-user scheme, although its subpacketization remains considerably larger than that of the MDS-based scheme. 
\end{remark}      

We next present a converse bound for the \((N,K,L)\) demand-private coded caching problem.
\begin{theorem}  \label{th_converse}
	For the $(N,K,L)$ demand-private coded caching problem, any achievable memory-rate pair $(M,R)$ must satisfy 
	\begin{align*}
		R \ge  (s-1+\lambda ) L - \frac{L\big(2 \lambda s+s(s-1) - t(t-1)\big)}{2 (N-L(t-1))} M,  
	\end{align*}   
	for any $s \in [1: \min\{ \lfloor N/L \rfloor, K \}  ] $, $\lambda \in [0,1]$, where $t \in [1:s] $ is the minimum value such that\footnote{Such a $t$ always exists, since \eqref{converse_condition} is satisfied when $t=s$.}     
	\begin{align} \label{converse_condition}
		L 	 \left(  s(s-1)-t(t-1)  +   2 \lambda s \right) \le  2 ( N-(t-1)L)t.    
	\end{align} 
\end{theorem}  
\begin{IEEEproof}[Sketch of Proof]
	We follow the proof of \cite[Theorem 2]{Qian2019} and extend it to the scenario where each user requests \(L\) files.  
	The detailed proof is provided in Appendix~\ref{sec_appconverse}.  
\end{IEEEproof}    
\begin{remark}
	Since the proof does not rely on the privacy constraint, the above converse bound also applies to the  \((N,K,L)\) coded caching problem without demand privacy. 
\end{remark}

The following theorem establishes the order optimality of the proposed scheme.     
\begin{theorem}\label{orderoptimal} 
	For the $(N,K,L)$ demand-private coded caching problem, the memory-rate tradeoff achieved by the scheme in Theorem~1, denoted by $R^{\mathrm p}_{N,K,L}(M)$, is order optimal within a factor of $6$  for $M \in [0,N]$. 
\end{theorem} 
\begin{IEEEproof}[Sketch of Proof]  
	The result follows from comparing the achievability result in Theorem~\ref{theorem_achA} with the converse bound in Theorem~\ref{th_converse}. The detailed proof is provided in Appendix~\ref{sec_apporderoptimal}. 
\end{IEEEproof}     
 \begin{remark}
 	For the case \(L>1\), the best previously known result establishes
 	order optimality within a factor of \(22\) only when \(N \leq LK\), or
 	when \(N > LK\) and \(M \geq N/K\) \cite{Wan2021}.  A tighter factor
 	of \(2\) is obtained when \(M \geq N/2\). In contrast, 
 	Theorem~\ref{orderoptimal}   improves the worst-case factor from \(22\)
 	to \(6\) and establishes order optimality in the remaining parameter
 	regimes. 
 \end{remark}

\begin{figure*}[t]
	\centering
	\subfloat[$(N,K,L)=(15,5,2)$\label{fig_rate_memory_a}]{
		\includegraphics[width=0.42\textwidth]{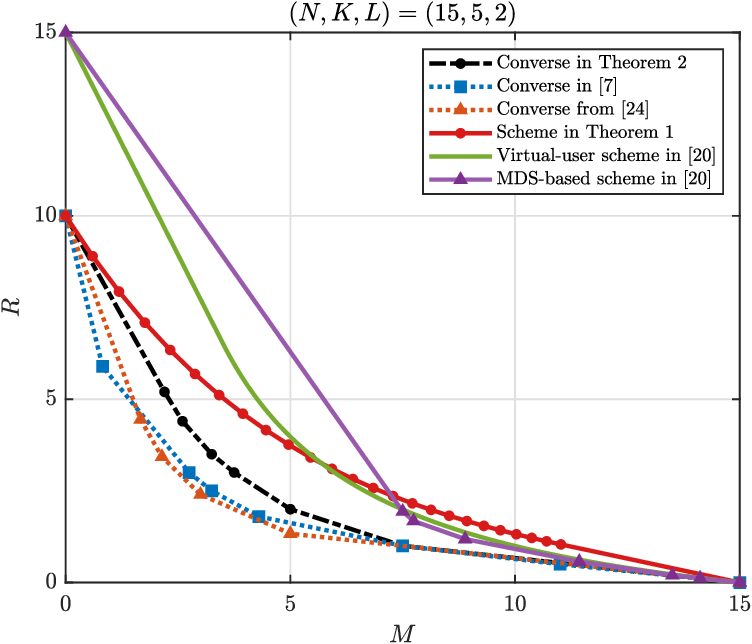}
	}
	\hspace{0.04\textwidth}
	\subfloat[$(N,K,L)=(8,5,2)$\label{fig_rate_memory_b}]{
		\includegraphics[width=0.42\textwidth]{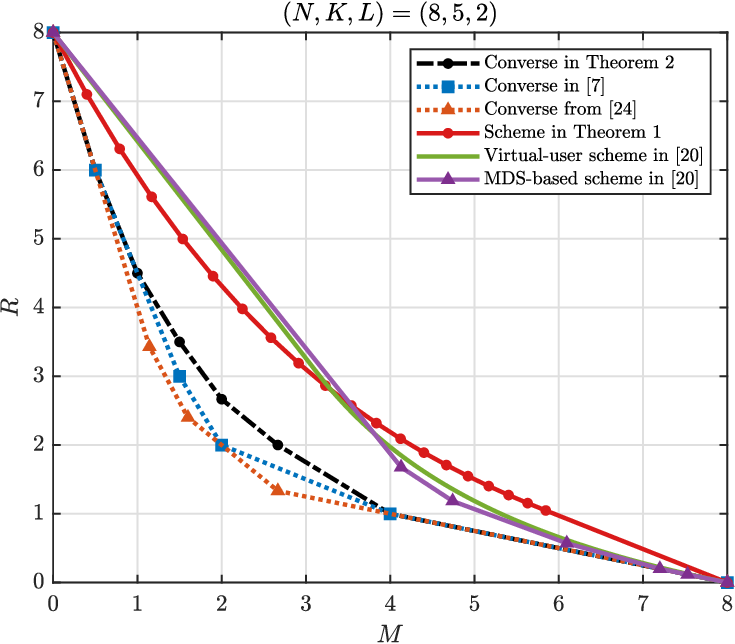}
	}
\caption{Comparison of achievable schemes and converse bounds.} 
	\label{fig_rate_memory}
\end{figure*}

Fig.~\ref{fig_rate_memory} provides numerical comparisons among the proposed and existing schemes and converse bounds for two representative cases.  
More specifically, the plotted converse bounds are those in
Theorem~\ref{th_converse} and \cite[Theorem~1]{Sengupta2017}, together with
the specialization of \cite[Theorem~1]{WanKaiFogRAN2021} to the
$(N,K,L)$ demand-private coded caching problem. 
For numerical evaluation, the bound in Theorem~\ref{th_converse} is
evaluated using the lower convex envelope characterized in
Appendix~\ref{sec_apporderoptimal}. The achievable schemes plotted are the proposed scheme, the virtual-user scheme, and the MDS-based scheme in~\cite{Wan2021}. 
Figs.~\ref{fig_rate_memory}(a) and~\ref{fig_rate_memory}(b) consider
$(N,K,L)=(15,5,2)$ and $(8,5,2)$, respectively, corresponding to
$N\ge KL$ and $N<KL$.   
The numerical results in Fig.~\ref{fig_rate_memory} show that the bound in
Theorem~\ref{th_converse} matches or improves upon the other two converse
bounds in both examples.  
On the achievability side, the proposed scheme outperforms both schemes
in~\cite{Wan2021} at least over $M\in[0,5.6]$ and
$M\in(0,3.3]$ for $(N,K,L)=(15,5,2)$ and $(8,5,2)$,
respectively.

We now consider the case $K=2$, for which the general order-optimal result above can be strengthened to establish the  optimal memory-rate tradeoff.   
\begin{theorem}\label{optimalK2}   
	For the \((N,K,L)\) demand-private coded caching problem with \(K=2\), 
	the optimal  memory-rate tradeoff \(R_{N,2,L}^{*\mathrm{p}}(M)  \) 
	for \(M\in[0,N]\)  is given as follows. 
	\begin{enumerate} 
		\item If \(N\ge 3L\), then 
		\begin{align*}
		& 		R_{N,2,L}^{*\mathrm{p}}(M) = 
			\max \bigg\{  
			2L - \frac{3L}{N}M, 
		  \nonumber \\
		& \quad  \quad  	\frac{(2N+L)L}{N+L} - 	\frac{3L}{N+L} M,	L- \frac{L}{N}M \bigg\},  
		\end{align*}    
		\item If \(2L \le N < 3L\), then   
		\begin{align*}   
		&	R_{N,2,L}^{*\mathrm{p}}(M)
			= 
			\max \bigg\{
			2L-\frac{2L}{N-L}M, 
			\frac{N+3L}{3}-M, \nonumber \\   
			& \quad  \quad \quad  \quad 
			\frac{(2N+L)L}{N+L}-\frac{3L}{N+L}M, 
			L-\frac{L}{N}M
			\bigg\}, 
		\end{align*} 
	 
		\item  {If \(L \le N <  2L\), then} 
	\begin{align*}
		&  R_{N,2,L}^{*\text{p}}(M) =  
		\max \bigg\{  
		N - \frac{N}{L}M,  \frac{2N+L}{3} -   M  
	\nonumber \\   
	& \qquad  \qquad \qquad \qquad \qquad  
			L- \frac{L}{N}M \bigg\}.  
	\end{align*}   
	\end{enumerate}
\end{theorem} 
\begin{IEEEproof}[Sketch of proof] 
The theorem follows from comparing the achievability and converse results in Sections~\ref{sec_optimalK2Ach} and~\ref{sec_optimalK2Con}, respectively, and showing that they meet in all regimes.   
  
\textit{Achievability:}  
For  \(K=2\), we propose three schemes that achieve the memory-rate pairs 
\(\big(\frac{N-\bar N+L}{3},\frac{2\bar N}{3}\big)\),
\(\big(\frac{N}{3},L\big)\), and
\(\big(\frac{N^2}{2N-L},\frac{L(N-L)}{2N-L}\big)\), where
\(\bar N=\min\{N,2L\}\).     
All three schemes begin by precoding each file with a random  Cauchy matrix, thereby generating Cauchy-coded segments.  
In the first scheme, each user caches linear combinations of these segments, with coefficients determined by a random Cauchy matrix, whereas in the second and third schemes, each user directly caches these segments at cache positions determined by a random permutation.  From each user's viewpoint, segments differ only in whether their information is available from the cache content.  
In the delivery phase, the server forms a random vector of length $\bar{N}$ consisting of distinct file indices and containing all requested file indices. It then carefully selects the corresponding Cauchy-coded segments according to this vector and the users' demands.   
The first scheme directly delivers these segments, whereas the second and third schemes deliver linear combinations of them.   
From each user's viewpoint, the resulting broadcast message has the same distribution for all possible demands of the other user, which establishes demand privacy.             
The above three points,  together with the two trivial endpoints \((0,\bar N)\) and \((N,0)\) from Theorem~\ref{theorem_achA}, establish the achievability of Theorem~\ref{optimalK2} by memory sharing.  
    
\textit{Converse:}
		By specializing Theorem~\ref{th_converse} to \(K=2\), we directly obtain the
		bound \(LM+NR\ge LN\), the bound \(3LM+NR\ge 2NL\) for
		\(N\ge 3L\), and the bound \(2LM+(N-L)R\ge 2L(N-L)\) for \(2L \le N<3L\). 
		To complete the converse, it remains to establish four additional bounds, referred to as Bounds~A-D, as follows  
		\begin{align*}
			\text{(A)}\;& 3LM+(N+L)R\ge (2N+L)L, && \text{for } N \ge 2L,\\
			\text{(B)}\;& 3M+3R\ge N+3L, && \text{for } 2L \le N<3L,\\
			\text{(C)}\;& 3M+3R\ge 2N+L, && \text{for } N<2L,\\
			\text{(D)}\;&  NM+LR\ge LN,  && \text{for } N<2L.
		\end{align*}    
 The proofs of Bounds~A-D proceed by deriving suitable entropy inequalities from carefully chosen combinations of cache contents and broadcast messages. 
 In this derivation, the privacy constraint is used to obtain equal entropy terms under different demand vectors, and these equalities allow the relevant terms to be brought into the combinations needed for the submodularity steps. 
 The desired bounds then follow by taking suitable linear combinations of the resulting inequalities or by summing the corresponding one-step inequalities over the relevant indices. 
\end{IEEEproof}  

Fig.~\ref{fig:optimalK2} illustrates the three regimes in Theorem~\ref{optimalK2} and the role of Bounds~A--D. 
\begin{figure*}[t]
	\centering
	\includegraphics[width=\textwidth]{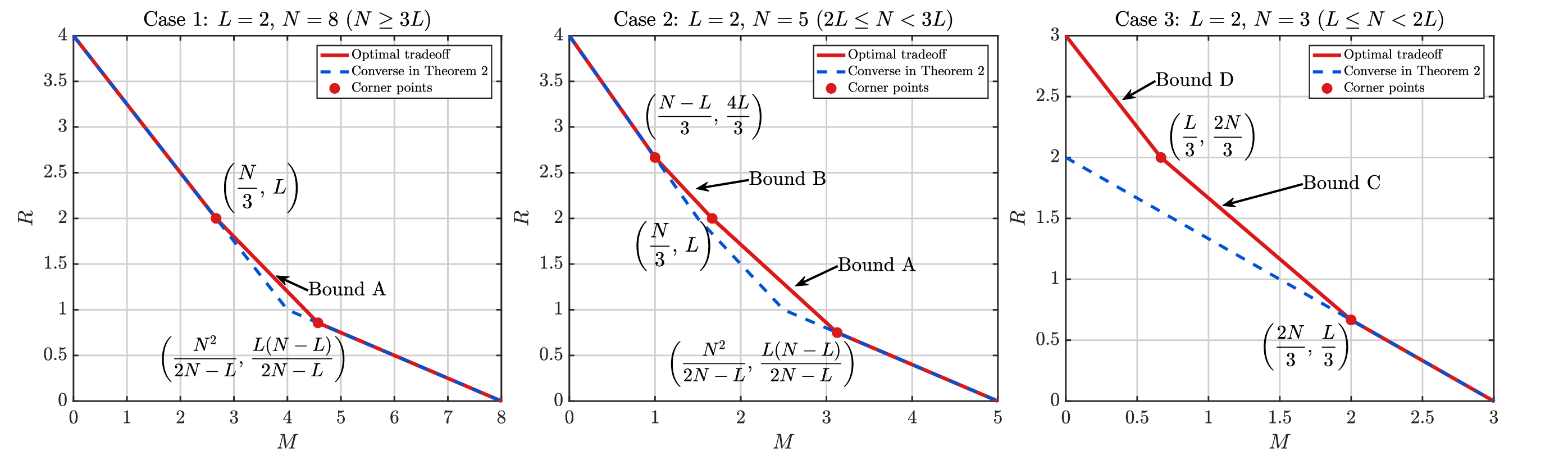}
\caption{Representative memory-rate tradeoffs for the three regimes in
	Theorem~\ref{optimalK2}. The solid and dashed curves represent the
	optimal tradeoffs and the converse bounds in Theorem~\ref{th_converse},
	respectively; the circles denote the nontrivial corner points, and the
	arrows identify Bounds~A--D.} 
	\label{fig:optimalK2}
\end{figure*}
\begin{remark}[Relations to known optimal results] 	  
For $L=1$, Theorem~\ref{optimalK2} coincides with the optimal memory-rate tradeoff characterized in \cite[Theorem~6]{Chinmay2022}.    
Moreover, in the
special case $K=2$, Theorem~\ref{optimalK2} contains the optimal result in
\cite[Theorem~6]{Wan2021}, which gives
$R^*=L\left(1-\frac{M}{N}\right)$ for
$M\ge \frac{2N}{3}$.    
\end{remark}         
 
\section{Proof of Theorem~\ref{theorem_achA}} \label{sec_theorem_achA} 
  Before presenting the proof of Theorem~\ref{theorem_achA}, we introduce the following notation, which will also be used in the achievability proof of Theorem~\ref{optimalK2} in Section~\ref{sec_optimalK2Ach}.  
	\begin{definition} \label{defNe}
 For a demand matrix $\mathbf{D}$, define $\mathcal{N}_k = \{d_{k,l}: l\in[L]\}$  as the set of files requested by user $k$, and define the set of all requested files as $\mathcal{N}_{\text e}(\mathbf{D}) = \bigcup_{k\in[K]} \mathcal{N}_k$.  
	We further define  the family of supersets of $\mathcal{N}_{\text e}(\mathbf{D})  $ of size  $\bar N$  as  
	\begin{align*}  
		& \mathcal{F}_{\bar{N}}(\mathbf{D}) = \left\{ \mathcal{N} \subseteq [N] :  \mathcal{N}_{\text{e}}(\mathbf{D}) \subseteq \mathcal{N},  |\mathcal{N}| = \bar{N} \right\}.    
	\end{align*} 		    
\end{definition}   
For example, Table \ref{tabexample0} lists $\mathcal{F}_{\bar{N}}(\mathbf{D})$  for different demand matrices   with  $N=5,K=2$ and $L=2$.    
\begin{table}[t]
	\centering 
	\caption{ $\mathcal{F}_{\bar{N}}(\mathbf{D})$ for Different $\mathbf{D}$ with  $N=5,K=2,L=2$} 	\label{tabexample0}    
	\begin{tabular}{|c|c|} 
		\hline  
		$\mathbf{D}$ 
		& $\mathcal{F}_{\bar{N}}(\mathbf{D})$  \\
		\hline   
		$[0,1;0,1]$ 
		& \begin{tabular}[t]{l}
			$\{ \{0,1,2,3\},$  
			$ \{0,1,2,4\},$  
			$\{0,1,3,4\}\}$ 
		\end{tabular} \\ 
		\hline 
		$[0,1;0,2]$  
		& \begin{tabular}[t]{l}
			$\{ \{0,1,2,3\},$  
			$\{0,1,2,4\} \}$
		\end{tabular}  \\
		\hline 		 
		$[0,1;2,3]$ 
		& $\{ \{0,1,2,3\} \}$  \\
		\hline  
	\end{tabular}
\end{table}        

In the remainder of this section, we first formally define an   \((N,K\bar N,1)\) RD-NP-UCC problem and  illustrate, through a motivating example, how a scheme for this problem can be transformed into a demand-private coded caching scheme.      
We then establish this  transformation in general and apply it to the YMA scheme to prove Theorem~\ref{theorem_achA}.   
 
\subsection{The \((N,K\bar N,1)\) RD-NP-UCC Problem}
\label{def_NonPrivate}  
We now define the \((N,K\bar N,1)\) RD-NP-UCC problem.   In this coded caching problem, there are \(N\) files and \(K\bar N\) users, each requesting one file.   
The file library is denoted by \(\mathbf{W}=[W_0,W_1,\ldots,W_{N-1}]^{\mathsf T}\) and follows the same statistical assumptions as in the system model. 
The demand vector $\widetilde{\mathbf D}$ has length $K\bar N$, and $\widetilde d_u$ denotes its $u$-th component. 
This problem is defined under uncoded placement and restricts the set of valid demand vectors to a subset \(\mathcal D_{\mathrm{RS}} \subseteq [N]^{K\bar N}\), both defined below.        
\begin{definition}[Uncoded Placement]
	\label{def_uncoded}
	The placement phase is said to be \emph{uncoded} if, for each user $u\in[K\bar N]$ and file $n\in[N]$, user $u$ selects a subset of the symbols of $W_n\in\mathbb{F}_q^F$ and stores them in its cache, without coding.  	
	Formally, let \(W_{n,i}\), \(i\in[F]\), denote the symbols of file \(W_n\). Then there exists an index set \(\mathcal{M}^{\mathrm{(np)}}_{u,n}\subseteq[F]\) such that user \(u\) stores \(\{W_{n,i}: i\in\mathcal{M}^{\mathrm{(np)}}_{u,n}  \}\).     
\end{definition}   
\begin{definition}[Restricted Demand Subset]  
	\label{def_RS} 
	A demand vector $\widetilde{\mathbf{D}}$ belongs to the restricted demand subset $\mathcal{D}_{\mathrm{RS}}$ if it can be partitioned into $K$ subvectors of length $\bar{N}$,
	\begin{align*}  
		\widetilde{\mathbf{D}} =  \big[ \widetilde{\mathbf{d}}_0^{\mathsf{T}}, \widetilde{\mathbf{d}}_1^{\mathsf{T}}, \dots, \widetilde{\mathbf{d}}_{K-1}^{\mathsf{T}} \big]^{\mathsf{T}}, 
	\end{align*}        
	and there exists a file set $\mathcal{N} \subseteq [N]$ with $|\mathcal{N}|=\bar{N}$ such that each subvector $\widetilde{\mathbf{d}}_k$ is a permutation of $\mathcal{N}$. 
\end{definition}

To illustrate Definition~\ref{def_RS},  consider $N=5$, $K=2$, and $L=2$, giving $\bar N = 4$.    The demand vector $	\widetilde{\mathbf{D}} = (0,1,2,3,\,1,2,3,0)^{\mathsf{T}} $   
belongs to $\mathcal{D}_{\mathrm{RS}}$, since it consists of two blocks
$\widetilde{\mathbf{d}}_0 = (0,1,2,3)^{\mathsf{T}}$ and $\widetilde{\mathbf{d}}_1 = (1,2,3,0) ^{\mathsf{T}}$,
both of which are permutations of the file set $\{0,1,2,3\}$.

An \((N,K\bar N,1)\) RD-NP-UCC scheme consists of an uncoded placement as in Definition~\ref{def_uncoded}, an encoding function \(\psi^{\mathrm{(np)}}\), and decoding functions \(\gamma_u^{\mathrm{(np)}}\), \(u\in[K\bar N]\).

Under uncoded placement and without the privacy constraint, it suffices to consider broadcast messages generated using only the requested files.   For any demand in $\mathcal{D}_{\mathrm{RS}}$, the requested files are $ \mathbf I_N(\widetilde{\mathbf d}_0,:)\mathbf W$. 
It follows that each user’s decoding depends only on the broadcast message and its cache contents corresponding to these requested files.    
This observation explains the form of the encoding and decoding functions given below.    

Thus,  for each demand vector \(\widetilde{\mathbf D}\in\mathcal D_{\mathrm{RS}}\), the server delivers 
\begin{align*} 
	X^{\mathrm{(np)}}_{\widetilde{\mathbf D}}
	= 
	\Big(
	\psi^{\mathrm{(np)}}\big(
	\widetilde{\mathbf D},\,
	\mathbf I_N(\widetilde{\mathbf d}_0,:)\mathbf W
	\big), 
	\widetilde{\mathbf D}
	\Big).    
\end{align*}    
and user \(u\in[K\bar N]\) recovers its requested file as  
\begin{align}	\label{achA_np_decoder}  
	W_{\widetilde {d}_u} 
	= 
	\gamma_u^{\mathrm{(np)}} \left( 
	X^{\mathrm{(np)}}_{\widetilde{\mathbf D}}, 
	\left( 	\left\{ W_{\widetilde {d}_i,j} : j \in \mathcal{M}_{u,\widetilde {d}_i}^{\mathrm{(np)}}  \right\}     \right)_{ i \in [\bar{N}]}
	\right).     
\end{align}       
The rate of such a scheme is defined as
\begin{align*}
	R
	=
	\frac{1}{F}
	\max_{\widetilde{\mathbf D}\in\mathcal D_{\mathrm{RS}}}
	H\!\left(X^{\mathrm{(np)}}_{\widetilde{\mathbf D}}\right).
\end{align*}    
\begin{remark} \label{remarkYMA}
	The YMA scheme in \cite{Qian2018} for \(N\) files and \(K \bar{N} \) users is an \((N,K \bar{N},1)\) RD-NP-UCC scheme, since its placement is uncoded and its broadcast message is generated using only the requested files. 
\end{remark}  

\subsection{Motivating Example: $N = 5,K = 2,L = 2, r = 1  $}	   
We consider the $(N,K,L)=(5,2,2)$ demand-private coded caching problem.   
The proposed scheme is based on the YMA scheme for \(N=5\) files and \(K\bar N=8\) users with parameter \(r=1\), which is an \((N,K\bar N,1)=(5,8,1)\) RD-NP-UCC scheme. 
A random permutation $\mathbf P=(P_0,P_1,\ldots,P_4)$ is chosen uniformly from the set of all permutations of $[0:4]$, and the YMA scheme is applied to the resulting relabeled file library.

\subsubsection*{Placement Phase} 
Each file $W_n$, $n\in[0:4]$, is partitioned into  $\binom{K \bar{N}}{r} = 8$ subfiles, denoted by  
$  
W_{n,\{i\}},  i\in[0:7]
$. 
For $k \in \{0,1\} $, the server independently and uniformly selects $ S_k=(S_{k,0},S_{k,1}) $,   
where $S_{k,0}$ and $S_{k,1}$ are distinct elements  in $[0:3]$.       
Then, user $k$ caches the aggregation of the cache contents indexed by $4k+S_{k,0}$ and $4k+S_{k,1}$ in the YMA scheme. 
More specifically, after applying the relabeling $\mathbf P$, the cache content of user $k$ is  given by  
\begin{align*}
	Z_k & =\bigl((Z_{k,i})_{i\in[0:4]},\,S_k\bigr),  \quad k = 0,1, \text{ where} \nonumber \\
	&	Z_{0,P_n}=\{W_{n,\{S_{0,0}\}},\,W_{n,\{S_{0,1}\}}\},\nonumber\\
	&	Z_{1,P_n}=\{W_{n,\{S_{1,0}+4\}},\,W_{n,\{S_{1,1}+4\}}\}, \quad n\in[0:4].  
\end{align*}   
Since the size of $S_k$ is negligible for sufficiently large $F$ and $Z_k$ contains $10$ subfiles of length $\frac{F}{8}$, it follows that $M=\frac{5}{4}$.

\subsubsection*{Delivery Phase}  
Given a demand matrix $\mathbf D$, the server delivers the broadcast message of the YMA scheme, where the key step is to construct the demand vector for the YMA scheme. 
First, the server draws a random set
$\mathcal{T} \sim \mathrm{Unif}(\mathcal{F}_4(\mathbf{D}))$,
where $\mathcal{F}_4(\mathbf{D})$ is defined in \eqref{defNe}.  
It then generates an expanded demand vector 
\begin{align*}
	\widetilde{\mathbf Q}
	=
	(\widetilde{q}_0,\widetilde{q}_1,\ldots,\widetilde{q}_7)^{\mathsf T}
	=
	[\widetilde{\mathbf q}_0^{\mathsf T},\widetilde{\mathbf q}_1^{\mathsf T}]^{\mathsf T},
\end{align*} 
where the blocks $\widetilde{\mathbf q}_0$ and $\widetilde{\mathbf q}_1$ are drawn independently. 
For each $k = 0,1$, given $(\mathcal{T},\mathbf d_k,S_k)$, the block $\widetilde{\mathbf q}_k$ is sampled uniformly from all  possible permutations of $\mathcal{T}$ such that the two files requested by user $k$ are in the positions determined by $S_{k,0}$ and $S_{k,1}$.  %

Since the YMA scheme is applied to the file library relabeled by $\mathbf P$, the demand vector for the YMA scheme is obtained by applying $\mathbf P$ to $\widetilde{\mathbf Q}$, i.e., 
\begin{align*}  
	\widetilde p_u=P_{\widetilde q_u},\quad u\in[0:7],
	\qquad\text{and}\quad
	\widetilde{\mathbf P}=(\widetilde p_0,\widetilde p_1, \ldots,\widetilde p_7)^{\mathsf T}.  
\end{align*}     
Thus, the   broadcast  message, i.e., $X_{\mathbf{D}}$,   is  given by    
\begin{align}  \label{ex1_X} 
	Y_{ \mathcal{R}^+} = &  \sum_{u \in \mathcal{R}^+ } W_{\widetilde{q}_u,  \mathcal{R}^+ \setminus \{u\}   },  \quad  \mathcal{R}^+ \subseteq [0:7], |\mathcal{R}^+|= 2,  \text{ and} \nonumber  \\   
	X_{\mathbf{D}} =  & \big( \{  Y_{ \mathcal{R}^+}:   \mathcal{R}^+ \cap [0:3] \neq \emptyset \},   \widetilde{\mathbf{P}}    \big),    
\end{align}     
where $[0:3]$ corresponds to the set of leaders in  the YMA scheme, since $\widetilde{\mathbf{P} }\in\mathcal D_{\mathrm{RS}}$.  

Following from \eqref{ex1_X}, we note that  $ X_{\mathbf{D}}$ contains $ \binom{8}{2}-\binom{4}{2} = 22$ segments, each of length $\frac{F}{8}$.  Moreover, the size of $\widetilde{\mathbf P}$ is negligible for sufficiently large $F$, thus, $R=\frac{11}{4}$.

To illustrate the above design, consider the demand matrix $ \mathbf {D}= [0,1;0,2]$, where $\mathcal{F}_4(\mathbf{D}) = \{   \{0,1,2,3\},\{0,1,2,4\}  \}$.     
Suppose the server draws $S_0=(0,2)$, $S_1=(1,3)$ and $\mathcal T=\{0,1,2,3\}$. 
Then one possible realization of $\widetilde{\mathbf Q}$ is $(0,2,1,3,1,0,3,2)^{\mathsf T}$, which yields $\widetilde{\mathbf P}=(P_0,P_2,P_1,P_3,P_1,P_0,P_3,P_2)^{\mathsf T}$.  

\subsubsection*{Sketch of Correctness}  
For $k\in \{0,1\}$ and $l\in\{0,1\}$, user $k$ decodes its requested file $W_{d_{k,l}}$ by following the decoding procedure of the virtual user indexed by $4k+S_{k,l}$ in the YMA scheme for demand $\widetilde{\mathbf P}$. Thus, correctness follows directly from that of the YMA scheme. 
For example,     
under the above realization of $\widetilde{\mathbf Q}$ and $\widetilde{\mathbf P}$, consider user $k=0$ and $l=1$, for which the requested file is $W_{d_{0,1}}=W_1$.  
Since $\widetilde p_2=P_1$, user $0$ directly obtains $W_{1,\{2\}}$ from  $Z_{0,P_1}$.
Next, for each $i\in[0:7]\setminus\{2\}$, we have 
$
W_{1,\{i\}} = W_{\widetilde q_2,\{i\}}  = Y_{\{2,i\}} - W_{\widetilde q_i,\{2\}},
$  
where $W_{\widetilde q_i,\{2\}}$ is available  from $Z_{0,\widetilde p_i}$.     

\subsubsection*{Sketch of Privacy}
Since the broadcast message is the same as that of the YMA scheme under demand $\widetilde{\mathbf P}$, any information a user can infer about the other users' demands from it is contained in $\widetilde{\mathbf P}$.   
Moreover, for each user $k$, the random variable $S_k$, known only to user $k$, randomizes the positions in $\widetilde{\mathbf P}$ corresponding to that user's demand. 
The random set $\mathcal T$ hides which files appear in the other entries of $\widetilde{\mathbf P}$.   
Within each block, these files are placed uniformly at random among the remaining positions.  
Finally, the random permutation $\mathbf {P}$ relabels the file indices and hides the original file labels. It follows that, conditioned on user $k$'s own demand and $S_k$, the distribution of $\widetilde{\mathbf P}$ is the same for all realizations of the other users' demands.   Thus, user $k$ learns nothing about the other users' demands from observing  $\widetilde{\mathbf P}$, and hence the privacy constraint is satisfied.

\subsection{From Non-Private to Demand-Private Coded Caching}      
We now formalize the transformation described above in the following lemma. 
\begin{lemma}\label{T1_Lemma}   
	Consider an \((N,K\bar N,1)\) RD-NP-UCC scheme with rate \(R\), and define  
	\begin{align*}
		\mathcal S^{(\bar N,L)}=
		\left\{
		(s_0,s_1,\dots,s_{L-1})\in[\bar N]^L: 
		s_i\neq s_j,  \forall   i\neq j
		\right\}. 
	\end{align*}
	Then, for the $(N,K,L)$ demand-private coded caching problem, there exists a scheme achieving rate $R$ and cache size
	\begin{align} \label{T1_Lemmasize}
		M=\frac{1}{F}\max_{k\in[K]}\max_{S_k\in\mathcal S^{(\bar N,L)}}
		\Bigg(\sum_{n\in[N]}\bigg|\bigcup_{l\in[L]}\mathcal M^{\mathrm{(np)}}_{k\bar N+S_{k,l},n}\bigg|\Bigg).
	\end{align} 
\end{lemma}   

We now prove Lemma~\ref{T1_Lemma} by constructing a scheme for the $(N,K,L)$ demand-private coded caching problem from an arbitrary $(N,K\bar N,1)$ RD-NP-UCC scheme.  

The server first generates a permutation 
\begin{align*}   
	\mathbf{P} = & (P_0,P_1,\ldots,P_{N-1}) \sim \mathrm{Unif}\big(\mathrm{Perm}([N])\big). 
\end{align*}    
Using the permutation, we define an auxiliary file library   
$\mathbf W^{(\mathrm{np})} = \left[W_0^{(\mathrm{np})},W_1^{(\mathrm{np})},\ldots,W^{(\mathrm{np})}_{N-1}  \right]^{\mathsf{T}}$ via the mapping  
\begin{align} \label{achA_Permu}
	W^{(\mathrm{np})}_{P_n}  	= W_n ,  \quad n\in[N]. 
\end{align}   
The following construction is based on an RD-NP-UCC scheme operating on \(\mathbf W^{(\mathrm{np})}\). 

\subsubsection*{Placement Phase}      
For  $k\in[K]$, the server randomly  selects  $L$ distinct elements from $[\bar N]$ as follows  
\begin{align*}      
	S_k = &(S_{k,0},S_{k,1},\ldots,S_{k,L-1}) \sim \mathrm{Unif}\big(\mathcal{S}^{(\bar N,L)}\big),  
\end{align*}  
where the random variables $S_0, S_1, \dots, S_{K-1}$, and $\mathbf{P}$ are mutually independent.    
 For  \(k\in[K]\), user \(k\) caches the contents corresponding to the \(L\) users indexed by \(k\bar N+S_{k,l}\), \(l\in[L]\), in the RD-NP-UCC scheme. Formally,  we define  
\begin{align}\label{achA_Zkl}      
	Z_{k,P_n,l} = & \left\{ W^{(\mathrm{np})}_{P_n,j}:  j \in \mathcal{M}^{(\mathrm{np})}_{k\bar N + S_{k,l},P_n}  \right\} \nonumber \\ 	 \overset{(a)}{=} 	   &  \left\{ W_{n,j}:  j \in \mathcal{M}^{(\mathrm{np})}_{k\bar N + S_{k,l},P_n}  \right\},   
\end{align}    
where $(a)$ follows from \eqref{achA_Permu}.  
The cache content of user $k$ is thus given by     
\begin{align} \label{achA_Z}    
	Z_k =  
	\big((Z_{k,i})_{i\in[N]},\, S_k\big),  \text{where }
	Z_{k,P_n} =  \bigcup_{ l \in [L]}   	Z_{k,P_n,l}.  
\end{align}  
 Since the size of $S_k$ is negligible and placement is uncoded, for any realization of $S_k$, the cache size equals the total number of file symbols in $Z_k$ normalized by $F$, thus \eqref{T1_Lemmasize} follows.     

\subsubsection*{Delivery Phase}  
For a given demand matrix $\mathbf{D}$, the server first draws a random set $\mathcal T \subseteq [N]$ of size $\bar N$ that contains $\mathcal N_{\mathrm e}(\mathbf D)$, i.e.,     
\begin{align}  \label{achA_setT}      
	\mathcal{T} \sim \mathrm{Unif}\big(\mathcal{F}_{\bar N}(\mathbf{D})\big),  
\end{align}    where $\mathcal F_{\bar N}(\mathbf D)$ is introduced in Definition~\ref{defNe}. 
For $k\in[K]$,  given $(\mathcal T,\mathbf d_k,S_k)$, the server then  independently generates $\widetilde{\mathbf q}_k$  for different $k$ as follows     
\begin{align}\label{achA_Qk} 
	& \widetilde{\mathbf{q}}_k |  \mathcal{T},\mathbf d_k ,S_k
	\sim \mathrm{Unif}\big(\mathcal Q_k \big), \quad  \text{where}  \nonumber \\
	& \mathcal{Q}_k  
	=  
	\big\{  (v_0,v_1, \dots, v_{\bar{N}-1} )^{\mathsf T}  \in \mathrm{Perm}(\mathcal T):  \nonumber \\
	& \qquad \qquad 
	v_{S_{k,l}} =d_{k,l},  \forall l\in[L]	\big\}.
\end{align}    
The expanded demand vector $\widetilde{\mathbf Q}  = \left( \widetilde q_0, \widetilde q_1, \dots, \widetilde q_{K\bar{N}-1}  \right)^{\mathsf {T}} $ is thus given by
\begin{align} \label{achA_Q}  
	\widetilde{\mathbf Q} 
	=
	\big[\widetilde{\mathbf q}_0^{\mathsf T},
	\widetilde{\mathbf q}_1^{\mathsf T},
	\dots,
	\widetilde{\mathbf q}_{K-1}^{\mathsf T}\big]^{\mathsf T}.  
\end{align}   
Based on $\widetilde{\mathbf Q}$,  the demand vector
of the RD-NP-UCC  scheme is given by 
\begin{align}\label{achA_P}  
	\widetilde{\mathbf{P}}  = \left( \widetilde p_0, \widetilde p_1, \dots, \widetilde p_{K\bar{N}-1}  \right)^{\mathsf {T}},  \quad \text{where } \widetilde p_u = P_{\widetilde q_u}.  
\end{align}  
Partition \(\widetilde{\mathbf P}\) into \(K\) subvectors of length \(\bar N\) as
\(\widetilde{\mathbf P}=[\widetilde{\mathbf p}_0^{\mathsf T},\widetilde{\mathbf p}_1^{\mathsf T}, \dots,\widetilde{\mathbf p}_{K-1}^{\mathsf T}]^{\mathsf T}\). 
Since each \(\widetilde{\mathbf p}_k\) is a permutation of \(\{P_i:\, i\in\mathcal T\}\),
it follows from Definition~\ref{def_RS} that
\(\widetilde{\mathbf P}\in\mathcal D_{\mathrm{RS}}\).   

The broadcast message is the same as that of the RD-NP-UCC scheme under demand $\widetilde{ \mathbf{P}} $, i.e.,    
\begin{align} \label{achA_X} 	  
	X_{\mathbf{D}} =  & 	X^{\mathrm{(np)}}_{\widetilde{\mathbf P}}
	= \left( \psi^{\mathrm{(np)}} \left(  \widetilde{ \mathbf{P}}, \mathbf{I}_N  \left( \widetilde{\mathbf{p}}_0,: \right) \mathbf{W}^{ \mathrm{(np)}}  \right),    \widetilde{ \mathbf{P}}  \right)     \nonumber \\  
	\overset{(a)}{=} &   \left(   \psi^{\mathrm{(np)}} \left(  \widetilde{ \mathbf{P}}, \mathbf{I}_N (\widetilde{\mathbf{q}}_0,: ) \mathbf{W}    \right),  \widetilde{ \mathbf{P}}   \right),  
\end{align}        
where $(a)$ follows from \eqref{achA_Permu}.  The resulting rate is thus the same as that of the RD-NP-UCC scheme.  

\subsubsection*{Proof of Correctness}   
For any $k \in [K]$ and $l \in [L]$, user $k$ can recover $W_{d_{k,l}}$ from $X_{\mathbf D}$ and $Z_k$ as follows.

Since $\widetilde{\mathbf P}$ is included in $X_{\mathbf D}$, user $k$ can determine $(\widetilde p_i)_{i\in[\bar N]}$ from $X_{\mathbf D}$. 
Then, by \eqref{achA_Z}, user $k$ can obtain $\left( Z_{k,\widetilde p_i,l} \right)_{i\in[\bar N]}$ from $Z_k$. Thus, user $k$ can recover $W_{d_{k,l}}$ through    
\begin{align*}  
	& \gamma_{k \bar{N} + S_{k,l}}^{\mathrm{(np)}}  \left(    X_{\mathbf{D}},  	\left(  	Z_{k,\widetilde {p}_i,l}     \right)_{ i \in [\bar{N}]}    \right) \nonumber \\
	\overset{(a)}{=}  &  \gamma_{k \bar{N} + S_{k,l}}^{\mathrm{(np)}}   \left(    X^{\mathrm{(np)}}_{	\widetilde{\mathbf P}},  	\left( 	\left\{ W^{\mathrm{(np)}}_{\widetilde {p}_i,j} : j \in \mathcal{M}_{ k \bar{N} + S_{k,l},\widetilde {p}_i}^{\mathrm{(np)}}  \right\}     \right)_{ i \in [\bar{N}]}    \right) \nonumber \\
	\overset{(b)}{=}  &    W^{\mathrm{(np)}}_{\widetilde{p}_{ k \bar{N} + S_{k,l}}}  
	\overset{(c)}{=}   W_{\widetilde{q}_{ k \bar{N} + S_{k,l}}}  
	\overset{(d)}{=}	W_{d_{k,l}},  
\end{align*}    
where $(a)$ follows from \eqref{achA_Zkl} and \eqref{achA_X}, 
$(b)$ follows from the correctness of the RD-NP-UCC scheme, i.e.,   \eqref{achA_np_decoder} with $\widetilde{\mathbf D} = \widetilde{\mathbf P}$,  
$(c)$ follows from \eqref{achA_Permu} and \eqref{achA_P},   
and
$(d)$ follows from \eqref{achA_Qk} and \eqref{achA_Q}, which imply that
$\widetilde{q}_{k\bar N + S_{k,l}} = d_{k,l}$. 
The proof of correctness is thus complete.

\subsubsection*{Proof of Privacy}
For any \(k\in[K]\), we have    
\begin{align*} 
	& I \left( \mathbf{D}_{ \setminus k}; Z_k ,  X_{\mathbf{D}}, \mathbf{d}_k   \right) 
	\le    I \big( \mathbf{D}_{ \setminus k}; Z_k ,  X_{\mathbf{D}},   \mathbf{d}_k ,  \mathbf{W}^{\mathrm{(np)}}   \big)   \nonumber \\   
	\overset{(a)}{=}  &    I\big( \mathbf{D}_{ \setminus k}; Z_k ,  X_{\mathbf{D}},   \mathbf{d}_k \big|   \mathbf{W}^{\mathrm{(np)}}    \big)  \nonumber \\ 
	= &   I\Big( \mathbf{D}_{ \setminus k};  (Z_{k,i})_{i \in [N]},    \psi^{\mathrm{(np)}} \left(  \widetilde{ \mathbf{P}}, \mathbf{I}_N (\widetilde{\mathbf{p}}_0,: ) \mathbf{W}^{\mathrm{(np)}} \right),  \nonumber \\      
	& \qquad     S_{k},  \widetilde{\mathbf{P}},  \mathbf{d}_k     \Big|   \mathbf{W}^{\mathrm{(np)}}   \Big)   	\overset{(b)}{=}     I \left( \mathbf{D}_{ \setminus k};    S_{k},  \widetilde{\mathbf{P}},  \mathbf{d}_k  \Big|   \mathbf{W}^{\mathrm{(np)}}     \right)    \nonumber \\   
	\overset{(c)}{=}  &  I \left( \mathbf{D}_{ \setminus k};    S_{k},  \widetilde{\mathbf{P}},  \mathbf{d}_k     \right)    \overset{(d)}{=}    I \left( \mathbf{D}_{ \setminus k};     \widetilde{\mathbf{P}}   \Big |  S_{k},   \mathbf{d}_k    \right),    
\end{align*}   
where $(a)$ follows from the fact that demands, files, and the permutation \(\mathbf P\) are mutually independent, which implies that \(I \big( \mathbf{D}_{ \setminus k};  \mathbf{W}^{\mathrm{(np)}} \big)=0\), 
$(b)$ follows from the fact that each \(Z_{k,i}\) is determined by \(W_i^{\mathrm{(np)}}\) and \(S_k\), and \(\psi^{\mathrm{(np)}} \big(\widetilde{\mathbf P}, \mathbf I_N(\widetilde{\mathbf p}_0,:)\mathbf W^{\mathrm{(np)}}\big)\) is determined by \(\widetilde{\mathbf P}\) and \(\mathbf W^{\mathrm{(np)}}\),     
$(c)$ follows from the fact that the files are i.i.d. and independent of $\mathbf P$, 
and $(d)$   follows from the fact that \(S_k, \mathbf d_0, \mathbf d_1, \ldots, \mathbf d_{K-1}\) are mutually independent.

 It remains to show that
	$I\big(\mathbf{D}_{\setminus k}; \widetilde{\mathbf{P}}  \big|  S_k,\mathbf{d}_k\big)=0$.
	We begin by characterizing the conditional distribution of \(\widetilde{\mathbf Q}\).   
For any fixed realization of \(\mathbf D\) and \(\mathcal T\), and any
$
\mathbf q=\big[\mathbf q_0^{\mathsf T},\mathbf q_1^{\mathsf T},\dots,\mathbf q_{K-1}^{\mathsf T}\big]^{\mathsf T}
$
such that \(\mathbf q_i\in \mathrm{Perm}(\mathcal T)\),  \(i\in[K]\setminus\{k\}\), and \(\mathbf q_k\in\mathcal Q_k\),  we have      
\begin{align}   \label{achA_PrivacyQ}  
	& \Pr \left( \widetilde{\mathbf{Q}} =  \mathbf{q} \Big |  \mathbf{D},\mathcal{T},S_k  \right)    
	\overset{(a)}{=}  \prod_{i\in[K]} \Pr \left( \widetilde{\mathbf q}_i = \mathbf{q}_i  |  \mathbf{D},\mathcal{T},S_k \right)       \nonumber \\    
	\overset{(b)}{=}      &  \Pr  \left( \widetilde{\mathbf{q}}_k = \mathbf{q}_k     | \mathbf{d}_k,\mathcal{T},S_k \right)   \times   \prod_{i\in[K] \setminus \{k\}} \Pr(\widetilde{\mathbf{q}}_i =   \mathbf{q}_i  | \mathbf{d}_i,\mathcal{T} )        \nonumber \\
	\overset{(c)}{=}      & \frac{1}{|\mathcal{Q}_k|}   \times \prod_{i\in[K] \setminus \{k\}}    \frac{1}{|\mathcal{Q}_i||\mathcal S^{(\bar N,L)}| }        \nonumber \\   
	\overset{(d)}{=}      & \frac{1}{(\bar{N} - L )! }   \times  \left( \frac{1}{ \bar{N} ! }     \right)^{(K-1)},  
\end{align}   
where $(a)$ and $(b)$ follow from the fact that \(\widetilde{\mathbf q}_i\), \(i\in[K]\), are generated independently across users from \((\mathcal T,\mathbf d_i,S_i)\), and that \((\mathbf d_i,S_i)\), \(i\in[K]\), are mutually independent,     
$(c)$ follows from \eqref{achA_Qk}, the uniformity of \(S_i\) over \(\mathcal S^{(\bar N,L)}\) for \(i\in[K]\setminus\{k\}\),  and the fact that, given \((\mathbf d_i,\mathcal T)\), for each \(\mathbf q_i\in\mathrm{Perm}(\mathcal T)\), there exists a unique \(S_i\in\mathcal S^{(\bar N,L)}\) such that \(\mathbf q_i\in\mathcal Q_i\),  
and $(d)$ follows from the definitions of \(\mathcal Q_i\), $i \in [K]$,  and \(\mathcal S^{(\bar N,L)}\).

Next, for $k \in [K]$, given $\mathcal Q_k$ and $\mathbf P$,  we define  
\begin{align*}  
	\mathcal{P}_k 
	=
	\left\{ 
	( 	P_{v_0},P_{v_1},\dots,P_{v_{\bar N-1}}   )^{\mathsf T}:
	(  
	v_0,v_1,\dots,v_{\bar N-1}
	)^{\mathsf T}  
	\in \mathcal Q_k
	\right\},  
\end{align*}  
where $ \mathcal{P}_k $ is determined by $  (\mathbf{P}, \mathbf{D},\mathcal{T},S_k)  $. 
With this definition, for any realization
$ \mathbf p=
\big[
\mathbf p_0^{\mathsf T},
\mathbf p_1^{\mathsf T},
\dots, 
\mathbf p_{K-1}^{\mathsf T}
\big]^{\mathsf T} $ 
of $\widetilde{\mathbf P}$ such that $\mathbf p\in\mathcal D_{\mathrm{RS}}$
and $\mathbf p_k\in\mathcal P_k$,    
since \(\mathbf P\) is a permutation of \([N]\), the event \(\{\widetilde{\mathbf P}=\mathbf p\}\) uniquely determines a realization of \(\widetilde{\mathbf Q}\). This realization satisfies \(\widetilde{\mathbf q}_k\in\mathcal Q_k\) and \(\widetilde{\mathbf q}_i\in\mathrm{Perm}(\mathcal T)\) for all \(i\in[K]\setminus\{k\}\).  
Otherwise, the event $\{\widetilde{\mathbf P}=\mathbf p\}$ is impossible. 
Thus, for any fixed $\mathbf P$, $\mathbf D$, $\mathcal T$, and $S_k$, and any
$\mathbf p\in\mathcal D_{\mathrm{RS}}$,   \eqref{achA_PrivacyQ} implies that    
\begin{align}  \label{achA_PrivacyP}
	\Pr \left( \widetilde{\mathbf{P}} =   \mathbf{p}    \Big | \mathbf{P}, \mathbf{D},\mathcal{T},S_k \right)   =  \frac{	\mathbf{1}\{ \mathbf{p}_k \in \mathcal{P}_k  \} }{(\bar{N} - L )! }   \times  \left( \frac{1}{ \bar{N} ! }     \right)^{(K-1)}.  
\end{align}         
Finally, for any fixed $\mathbf D$ and $S_k$, and any
$\mathbf p\in\mathcal D_{\mathrm{RS}}$, we have
\begin{align} \label{achA_Privacylast}  
	& \Pr \left( \widetilde{\mathbf{P}} = \mathbf{p} \Big| \mathbf{D},S_k \right) \nonumber \\  
	\overset{(a)}{=}    & \sum_{\mathcal{T} \in \mathcal F_{\bar N}(\mathbf D)}    
	\sum_{\mathbf{P} \in \mathrm{Perm}([N]) }\frac{ 	\Pr \left( \widetilde{\mathbf{P}} = \mathbf{p} \Big|  \mathbf{P}, \mathbf{D},\mathcal{T},S_k \right)  }{ N ! |\mathcal F_{\bar N}(\mathbf D)|  }      \nonumber \\   
	\overset{(b)}{=}     &	\sum_{\mathcal{T} \in \mathcal F_{\bar N}(\mathbf D)}   	\sum_{\mathbf{P}  \in \mathrm{Perm}([N]) }\frac{  \mathbf{1}\{ \mathbf{p}_k \in \mathcal{P}_k  \}  }{ N ! (\bar{N} - L )! (\bar{N}! )^{K-1}  |\mathcal F_{\bar N}(\mathbf D)|  }    \nonumber \\
	\overset{(c)}{=}     &	\sum_{\mathcal{T} \in \mathcal F_{\bar N}(\mathbf D)}   	   \frac{ |\mathcal Q_k|(N-\bar N)!  }{ N ! (\bar{N} - L )! (\bar{N}! )^{K-1}  |\mathcal F_{\bar N}(\mathbf D)|  } \nonumber \\
	\overset{(d)}{=}     &	 	   \frac{  (N-\bar N)!  }{ N !  (\bar{N}! )^{K-1}     } ,        
\end{align}   
where $(a)$ follows from the conditional independence and uniformity of $\mathbf{P}$ and $\mathcal{T}$ given $(\mathbf{D}, S_k)$,   
$(b)$  follows from \eqref{achA_PrivacyP},   
$(c)$ follows from the fact that, for any fixed $(\mathbf D,\mathcal T,S_k)$, the set $\mathcal Q_k$ is determined, and for any fixed $\mathbf p_k$, each  $ (v_0,v_1, \dots, v_{\bar{N}-1} )^{\mathsf T}  \in\mathcal Q_k$   determines a disjoint set of $(N-\bar N)!$ permutations $\mathbf  P\in\mathrm{Perm}([N])$ such that $(P_{v_0},\dots,P_{v_{\bar N-1}})^{\mathsf T}=\mathbf p_k$, 
and $(d)$ follows from $|\mathcal Q_k|=(\bar N-L)!$ and the summation over $\mathcal T$. 

It follows from \eqref{achA_Privacylast} that, for any $\mathbf p\in\mathcal D_{\mathrm{RS}}$, 
$\Pr \big( \widetilde{\mathbf{P}} = \mathbf{p} \big| \mathbf{D},S_k \big)$ is independent of $\mathbf D$.
Thus,  we have 
$
I \big( \mathbf D_{\setminus k};\widetilde{\mathbf P} \big| S_k,\mathbf d_k \big)=0 
$.  This completes the proof of privacy, and so is the proof of Lemma~\ref{T1_Lemma}.

\subsection{Completion of the Proof of Theorem~\ref{theorem_achA}}

As noted in Remark~\ref{remarkYMA}, the YMA scheme for \(N\) files and \(K\bar N\) users is an \((N,K\bar N,1)\) RD-NP-UCC scheme.  
In this scheme, for \(r\in [0:K\bar N]\), each file is partitioned into \(\binom{K\bar N}{r}\) subfiles \(W_{n,\mathcal R}\), where \(\mathcal R\subseteq[K\bar N]\) with \(|\mathcal R|=r\), and user \(u\) caches all subfiles with \(u\in\mathcal R\).
Hence, the union of the caches of any \(L\) users contains exactly those subfiles \(W_{n,\mathcal R}\) with \(\mathcal R\) intersecting this set, whose number is \(\binom{K\bar N}{r}-\binom{K\bar N-L}{r}\). 
Applying Lemma~\ref{T1_Lemma}, we obtain
\begin{align*}
	M=\frac{\binom{K\bar N}{r}-\binom{K\bar N-L}{r}}{\binom{K\bar N}{r}} N.
\end{align*}  
Moreover, since \(\widetilde{\mathbf P}\in\mathcal D_{\mathrm{RS}}\), the number of distinct file indices in \(\widetilde{\mathbf P}\) is \(\bar N\), and thus   
\begin{align*}
	R=\frac{\binom{K\bar N}{r+1}-\binom{(K-1)\bar N}{r+1}}{\binom{K\bar N}{r}}.
\end{align*}   
By using memory sharing \cite{MaddahAli2014}, the lower convex envelope of the above memory-rate pairs is also achievable.  

The proof of Theorem~\ref{theorem_achA} is thus complete.   
   
  \section{Proof of Theorem~\ref{optimalK2}: Achievability}  \label{sec_optimalK2Ach}  
In this section, we prove the achievability part of
Theorem~\ref{optimalK2} by presenting three schemes. We first introduce two definitions and one lemma that will be used throughout the achievability proof. A running example is provided at the end of this section to illustrate the proposed schemes.  
  
The following encoding procedure is applied to each file in all three schemes.   
	\begin{definition}[$(a,b)$ Random Cauchy Encoding] \label{def_Cauchy} 
		Each file $W_n$, $n\in[N]$, is split into $b$ disjoint subfiles of equal size,
		denoted by $W_{n,0}, W_{n,1}, \ldots, W_{n,b-1}$, where each $W_{n,i}$ is a column vector of length $ \frac{F}{b}$.
		Let $\mathbf{W}_n=\big[W_{n,0},W_{n,1},\ldots,W_{n,b-1}\big]^{\mathsf{T}}$.  
		For each $n\in[N]$, the server generates an $a\times b$ random Cauchy matrix
		$\mathbf{C}_n=[C_{n,i,j}]_{i\in[a],\,j\in[b]}\in\mathbb{F}_q^{a\times b}$, with the matrices $\{\mathbf{C}_n:n\in[N]\}$ being mutually independent.  
		The entries of	$\mathbf{C}_n$ are given by
		\begin{align} \label{def_Cauchy_1}
			C_{n,i,j} = \frac{1}{\alpha_{n,i} - (a+j)}, \quad i\in[a],\ j\in[b], 
		\end{align} 
		where
		$\alpha_n=(\alpha_{n,0}, \alpha_{n,1}, \ldots,\alpha_{n,a-1})^{\mathsf{T}}
		\sim\mathrm{Unif}(\mathrm{Perm}([a]))$.
		The  Cauchy-coded segments are then defined by  
		\begin{align} \label{def_Cauchy_2}
			\mathbf{F}_n  
			=
			\big[F_{n,0}, F_{n,1}, \ldots, F_{n,a-1}\big]^{\mathsf{T}}
			=
			\mathbf{C}_n\mathbf{W}_n.
		\end{align} 
	\end{definition} 
	 
We also define an auxiliary demand vector $\mathbf{Q}$, which, together with the demand matrix $\mathbf{D}$, will be used to construct the broadcast messages in all three schemes.     
\begin{definition} \label{def_Q}
Recall that $\bar N = \min\{N, KL\}$ and that 	  $\mathcal{F}_{\bar{N}}(\mathbf{D})$ is defined  in Definition~\ref{defNe}. We define
	$\mathcal{D}_{\bar{N}}(\mathbf{D})$ as the set of all permutations of all sets in
	$\mathcal{F}_{\bar{N}}(\mathbf{D})$, i.e.,
	$
		\mathcal{D}_{\bar{N}}(\mathbf{D})
		=
		\bigcup_{\mathcal{N} \in \mathcal{F}_{\bar{N}}(\mathbf{D})}
		\mathrm{Perm}(\mathcal{N})$.  
	The auxiliary demand vector $\mathbf{Q}$ is drawn uniformly at random from	$\mathcal{D}_{\bar{N}}(\mathbf{D})$, i.e.,
	\begin{align*}
		\mathbf{Q}  
		=
		\left(Q_{0}, Q_{1}, \ldots, Q_{\bar{N}-1}\right)^{\mathsf{T}}
		\sim
		\mathrm{Unif}\left(\mathcal{D}_{\bar{N}}(\mathbf{D})\right).
	\end{align*} 
	For any file index $n$ that appears in $\mathbf{Q}$, define its inverse $Q_n^{-1}$ as the unique position $i \in [\bar{N}]$ such that $Q_i = n$. 
\end{definition}
 
The following lemma characterizes the distribution of the positions of the files
requested by each user in the auxiliary demand vector \(\mathbf Q\). 
\begin{lemma} \label{lemma_Q}
	Let \(T_{k,l}=Q^{-1}_{d_{k,l}}\), \(l\in[L]\). For each user $k$, conditioned on $\mathbf D$, the position vector
	$(T_{k,l})_{l\in[L]}$ is uniformly distributed over all vectors of length
	$L$ with distinct entries in $[\bar N]$. Equivalently, for any distinct 
	\(t_{k,l}\in[\bar N]\), \(l\in[L]\), we have 
	\begin{align*}
		\Pr\left(
		(T_{k,l})_{l\in[L]}=(t_{k,l})_{l\in[L]}
		\,\middle|\, \mathbf D
		\right)
		=
		\frac{(\bar N-L)!}{\bar N!}.
	\end{align*}
\end{lemma} 
\begin{IEEEproof}
	For any fixed $\mathbf D$, any fixed
	$\mathcal N\in\mathcal F_{\bar N}(\mathbf D)$, and any distinct
	$t_{k,l}\in[\bar N]$, $l\in[L]$, we have
    \begin{align*}
		\Pr\left(
		(T_{k,l})_{l\in[L]}=(t_{k,l})_{l\in[L]}
		\big|  \mathbf{D},\mathcal{N}
		\right)
		=
		\frac{(\bar{N}-L)!}{\bar{N}!}, 
	\end{align*}   
	where the equality follows from Definitions~\ref{defNe} and~\ref{def_Q}. More specifically, $\mathcal N$ contains $d_{k,l}$ for all $l\in[L]$,
	and $\mathbf Q$ is uniformly distributed over its $\bar N!$ permutations,
	of which exactly $(\bar N-L)!$ satisfy
	$Q_{t_{k,l}}=d_{k,l}$ for all $l\in[L]$. 
	  
Since the conditional probability above is independent of
$\mathcal N\in\mathcal F_{\bar N}(\mathbf D)$, averaging over $\mathcal N$
yields the desired result. 
\end{IEEEproof}   
 
\subsection{Achievability of $\big(\frac{N-\bar{N} + L}{3},  \frac{2\bar{N}}{3}\big)$ }    \label{subsec_achB}   
We now present a scheme that achieves the memory-rate pair 
$\big(\frac{N-\bar{N} + L}{3},  \frac{2\bar{N}}{3}\big)$, where $ \bar{N} = \min\{N,2L\}$.  
\subsubsection*{Placement Phase}    
The server first applies the $(4,3)$ Random Cauchy Encoding in Definition~\ref{def_Cauchy} to each file $W_n$, $n\in[N]$, thereby obtaining the Cauchy-coded segments
$F_{n,0},F_{n,1},F_{n,2},F_{n,3}$.  
Next, the main load of the cache content for user $k \in \{0,1\}$, denoted by $ Z_k' $,  is defined as follows   
 \begin{align} \label{achB_C1} 
 Z_k' = \left[ \left( Z_{k,i}\right)_{i \in [N-\bar N + L]} 	 \right]^{\mathsf{T}}    = \mathbf{H} 	\left[ \left( F_{n,k}\right)_{n \in [N]} 	 \right] ^{\mathsf{T}},   
	\end{align}     
where $\mathbf{H} = [H_{i,j}]_{ i\in[N-\bar N + L], \, j\in[N]}  \in\mathbb{F}_q^{(N-\bar N + L)\times N}$ is a random Cauchy matrix with entries  
\begin{align}   \label{achB_C2} 
	& H_{i,j} = \frac{1}{N + i  - P_j}, \; i\in[N-\bar N + L], \; j\in[N],  
	\qquad \nonumber \\
	&  \text{and} \;  \mathbf{P} = (P_0,P_1,\dots,P_{N-1}) ^{\mathsf{T}}
	\sim \mathrm{Unif}\bigl(\mathrm{Perm}([N])\bigr). 
\end{align}      
Finally,  the cache content for  user $k$ is given by
\begin{align}   \label{achB_Zk}
	Z_k = (Z_k', S_k ),  \; \text{where} \;  S_k =  ( S_{k,l,i})_{l\in[L], i \in [0:3]},     
\end{align}    
and the random variables $S_{k,l,i}$,
$k\in\{0,1\}$, $l\in[L]$, $i\in[0:3]$, are i.i.d.  and uniformly
distributed over $\mathbb{F}_q$. 

Since $Z_k'$ consists of   $N - \bar{N}+L$ linear combinations of the Cauchy-coded segments, each of length  $\frac{F}{3}$, and  the size of $S_k$ is negligible for sufficiently large $F$, it follows that
$ M =   \frac{N-\bar{N} + L}{3}$.

 \subsubsection*{Delivery Phase}    
The server generates a broadcast message $X_{\mathbf D}=(X_{\mathbf D}',J)$, consisting of the main load $X_{\mathbf D}'$ and an auxiliary variable $J$.  
 
Recall that \(\mathbf{Q} = (Q_0, Q_1, \ldots, Q_{\bar N-1})^{\mathsf{T}}\) is defined in Definition~\ref{def_Q}. The main load \(X_{\mathbf D}'\) is given by 
\begin{align} \label{achB_DX} 
	&	X_{\mathbf D}' = 	  \left[ \left( X_{i}\right)_{i \in [2\bar{N}]} 	 \right]^{\mathsf{T}}  =  \left[ \left(Y_{i,0},   Y_{i,1} \right)_{ i \in  [\bar{N}] }\right]^{\mathsf{T}},     \\  
	& 	\text{where } Y_{i,k} = 
	\begin{cases} 
		F_{Q_{i},k+2}, & \text{if } Q_{i} \in \mathcal{N}_{k}  \\
		F_{Q_{i},k},   & \text{if } Q_{i} \notin \mathcal{N}_{k}  
	\end{cases}.   \label{achB_DY}  
\end{align}    
The auxiliary variable  $J$ is given by
\begin{align}  \label{achB_DJ1}   
	J  =  
	\Bigl(  (J_{k,l})_{k\in\{0,1\},  l\in[L]}, 
	(J_{2,i})_{i\in[\bar N]} \Bigr),  
\end{align}  
with the components specified as follows.  For each   $k \in \{0,1\}$ and $ l \in [L] $, the component \(J_{k,l}\) provides user \(k\) with the coding coefficients and index information required to decode its \(l\)-th requested file, and is further divided into four parts, i.e.,  $  J_{k,l} = ( J_{k,l,0}, J_{k,l,1}, J_{k,l,2}, J_{k,l,3})$. 
To specify its first three components,  we define the selection matrix  $ \mathbf{U}_{k,l}$  as follows 
	\begin{align}\label{achB_DU} 
	& 	\mathbf U_{0,l}  
	= 
	\begin{cases}
		\mathbf I_4 \big((2,3,0)^{\mathsf T},:\big), & \text{if } d_{0,l}\in\mathcal{N}_1 \\
		\mathbf I_4 \big((2,1,0)^{\mathsf T},:\big), & \text{if } d_{0,l}\notin\mathcal{N}_1 
	\end{cases}, \; \text{and}
	\nonumber \\ 
	& 	\mathbf U_{1,l}
	= 
	\begin{cases}
		\mathbf I_4 \big((2,3,1)^{\mathsf T},:\big), & \text{if } d_{1,l}\in\mathcal{N}_0 \\
		\mathbf I_4 \big((0,3,1)^{\mathsf T}, :\big), & \text{if } d_{1,l}\notin\mathcal{N}_0 
	\end{cases}. 
\end{align}    
Thus, the first three components of $J_{k,l}$ are given by      
\begin{align} 
\label{achB_DJ2} 
& J_{k,l,i} =  A_{k,l,i}+S_{k,l,i}, \qquad i\in [0:2], \text{ where} \nonumber \\
&	[A_{k,l,0},A_{k,l,1},A_{k,l,2}]^{\mathsf{T}} =  \mathbf U_{k,l}  \left[ (\alpha_{d_{k,l}, i })_{i\in[0:3]}    \right]^{\mathsf{T}},    
\end{align}     
and the last component is given by    
\begin{align}\label{achB_DJ3}
J_{k,l,3}=T_{k,l}+S_{k,l,3},  \text{ where }  T_{k,l}=Q^{-1}_{d_{k,l}} .
\end{align}     
In addition,  $(J_{2,i})_{i\in[\bar N]} $  contains the publicly known information used in the decoding and is given by 
\begin{align} \label{achB_DJ4}   
J_{2,i} = P_{Q_i},  \quad i \in [\bar{N}]. 
\end{align}
From \eqref{achB_DX}, we see that $X_{\mathbf D}'$ contains $2\bar{N}$ segments, each of length $\frac{F}{3}$.  Since the size of  $J$ is  negligible when $F$ is sufficiently large, the achievable rate is  given by    $	R   =  \frac{2\bar{N}}{3}$.

\subsubsection*{Proof of Correctness}
For each \(k \in \{0,1\}\) and \(l \in [L]\), we describe how user \(k\)
decodes the requested file \(W_{d_{k,l}}\).  

We first show how two Cauchy-coded segments of \(W_{d_{k,l}}\) are obtained from
\(X_{\mathbf D}'\).  
User \(k\) recovers the index \(T_{k,l}=J_{k,l,3}-S_{k,l,3}\), and then, for
\(i\in\{0,1\}\), we have 
\begin{align}   \label{achB_pc2}  
	X_{2T_{k,l}+i} 	\overset{(a)}{=} &   Y_{T_{k,l} ,i}   
	\overset{(b)}{=}  
	\left( \mathbf{e}_i^{3} \right)^{\mathsf{T}}    \mathbf{U}_{k,l}  \mathbf{F}_{d_{k,l}},    
\end{align}   
where $(a)$ follows from \eqref{achB_DX}, and $(b)$ follows from
\eqref{achB_DY}, \eqref{achB_DU}, and $Q_{T_{k,l}}=d_{k,l}$.

Next, user $k$ decodes $ F_{d_{k,l},k}$,  i.e., $(\mathbf e_2^3)^{\mathsf T}\mathbf U_{k,l}\mathbf F_{d_{k,l}}$.  
Note that user \(k\) can recover the indices \((T_{k,l})_{l\in[L]}\) from 
\(J\) and \(S_k\).    
For \(i\in[N-\bar N+L]\),  user \(k\) computes   
\begin{align}  \label{achB_pc3}
	\widetilde Z_{k,i}
	=
	Z_{k,i}
	-
	\sum_{j\in[\bar N]\setminus\{T_{k,r}:r\in[L]\}}
	\frac{1}{N+i-J_{2,j}}X_{2j+k}.
\end{align} 
User \(k\) then forms a matrix \(\mathbf G_k\) as follows 
\begin{align}
	\label{achB_pcg1}
	\mathbf G_k
	=
	[g_{k,i,j}]_{i,j\in[N-\bar N+L]},
	\qquad
	g_{k,i,j}
	=
	\frac{1}{N+i-\theta_{k,j}},
\end{align}
where the values \(\theta_{k,j}\) are given by 
\begin{align}
	\label{achB_pcg2}
	&\theta_{k,j}=J_{2,T_{k,j}}, \qquad j\in[L], \nonumber\\
	&\{\theta_{k,L},\theta_{k,L+1},\dots,\theta_{k,N-\bar N+L-1}\}
	\nonumber\\
	&\qquad =
	[N]\setminus\{J_{2,0},J_{2,1},\dots,J_{2,\bar N-1}\}.
\end{align}  
Then, the segment $F_{d_{k,l},k}$ is obtained by 
\begin{align} 
	\label{achB_pc4}
	F_{d_{k,l},k}
	=
	\left(
	\mathbf G_k^{-1}
	\left[
	\big(\widetilde Z_{k,i}\big)_{i\in[N-\bar N+L]}
	\right]^{\mathsf T}
	\right)^{\mathsf T}
	\mathbf e_l^{\,N-\bar N+L}.
\end{align} 
The correctness of \eqref{achB_pc4}, together with the invertibility of 
\(\mathbf G_k\), is provided in Appendix~\ref{appendix_pc3correctness}.

Combining the result of \eqref{achB_pc2}  and   \eqref{achB_pc4},  user~$k$ obtains $\mathbf U_{k,l}\mathbf F_{d_{k,l}}$. 
Moreover, user \(k\) recovers
\(A_{k,l,i}=J_{k,l,i}-S_{k,l,i}\), \(i\in[0:2]\).
Using these values,  user $k$ can determine 
\begin{align}  \label{achB_pc5}  
	\mathbf{U}_{k,l}\mathbf{C}_{d_{k,l}}  =  \left[  \frac{1}{ A_{k,l,i}-(4+j)}\right]_{i \in [0:2], j \in [0:2]},     
\end{align}   
which follows from  \eqref{def_Cauchy_1}  and \eqref{achB_DJ2}.     
Since $\mathbf{C}_{d_{k,l}}$ is a Cauchy matrix, $\mathbf{U}_{k,l}\mathbf{C}_{d_{k,l}}$ is invertible.   Thus, user $k$  can  compute  
\begin{align}  \label{achB_pc6}   
	(\mathbf{U}_{k,l}\mathbf{C}_{d_{k,l}})^{-1}\mathbf{U}_{k,l}\mathbf{F}_{d_{k,l}}
	\overset{(a)}{=} &  (\mathbf{U}_{k,l}\mathbf{C}_{d_{k,l}})^{-1} \mathbf{U}_{k,l}\mathbf{C}_{d_{k,l}} \mathbf{W}_{d_{k,l}} \nonumber \\
	=  &  [{W}_{d_{k,l},0}, {W}_{d_{k,l},1}, {W}_{d_{k,l},2} ]^{\mathsf T} , 
\end{align}   
where $(a)$  follows from \eqref{def_Cauchy_2}.      
Finally, user $k$  concatenates ${W}_{d_{k,l},0}, {W}_{d_{k,l},1}, {W}_{d_{k,l},2} $ in order to obtain $W_{d_{k,l}}$.

The proof of correctness is thus complete.
    
\subsubsection*{Proof of Privacy}   
First,  for
$k \in \{0,1\}$, we have 
 \begin{align}  \label{achB_pp1} 
& H \left(   Z_k' ,  X_{ \mathbf{D}}'  |  \mathbf{D}, J, S_k \right)  \nonumber \\ 
\overset{(a)} \ge  &   H  \left( \left(F_{n,k} \right)_{n \in [N] }, \left(F_{n,k+2} \right)_{n \in \mathcal{N}_k},   \left(Y_{i,1-k} \right)_{i \in [\bar{N}]}  \Big|  \mathbf{D}, J, S_k    \right)   \nonumber \\     
\overset{(b)}{=}   &   
  \frac{L+N + \bar{N}}{3}F,     
 \end{align}   
 where $(a)$ follows from \eqref{achB_C1}, \eqref{achB_DX}, and \eqref{achB_DY},
 	and  the fact that user $k$ can recover $(F_{n,k})_{n\in[N]}$, 
 	and $(b)$ follows  from the fact that the entropy term in $(a)$ involves
 	$L+N+\bar N$ mutually independent Cauchy-coded segments, each of length
 	$\frac{F}{3}$.   
 	 Combining \eqref{achB_pp1}  with the fact that 
 $ 
 H \left( Z_k', X_{\mathbf{D}}' \big| \mathbf{D}, J, S_k \right)
 \le (M+R)F = \frac{L+N+\bar{N}}{3}F, 
 $ 
 it follows that for $k \in \{0,1\}$ and any fixed realization of $(\mathbf{D}, J, S_k)$,    
\begin{align}
	\label{achB_pp2}
	\left( X_{\mathbf{D}}', Z_k' \big| \mathbf{D}, J, S_k \right)
	\sim \mathrm{Unif} \left( \mathbb{F}_q^{(\bar{N}+N+L)F/3} \right).
\end{align}       
 
Next, we consider the auxiliary variables. 
For any fixed demand matrix $\mathbf D\in\mathcal D$, any distinct
$t_{k,l}\in[\bar N]$, $l\in[L]$, any distinct $p_i\in[N]$,
$i\in[\bar N]$, and any $a_{k,l,i}\in[0:3]$, $i\in[0:2]$, $l\in[L]$,
such that $a_{k,l,0},a_{k,l,1},a_{k,l,2}$ are distinct for each $l\in[L]$,
we have   
\begin{align} \label{achB_pp3}
	& \Pr  \big( 
	(T_{k,l})_{l\in[L]}=(t_{k,l})_{l\in[L]},  
	(J_{2,i})_{i\in[\bar N]}=(p_i)_{i\in[\bar N]},  \nonumber \\
	& \qquad \quad  
	(A_{k,l,i})_{l\in[L],i\in[0:2] }
	=(a_{k,l,i})_{l\in[L], i\in[0:2] }
	\big| \mathbf D \big)  \nonumber \\
	\overset{(a)}{=} 	& 
	\prod_{l\in[L]}   
	\Pr  \big(   
	(A_{k,l,i})_{i\in[0:2]}=(a_{k,l,i})_{i\in[0:2]}  
	\big| \mathbf D \big)  \nonumber \\ 
	&  \times   
	\Pr  \big( 
	(T_{k,l})_{l\in[L]}=(t_{k,l})_{l\in[L]} 
	\big| \mathbf D \big)	 \nonumber \\ 
	& \times  
	\Pr  \big(  
	(J_{2,i})_{i\in[\bar N]}=(p_i)_{i\in[\bar N]} 
	\big|
	(T_{k,l})_{l\in[L]}=(t_{k,l})_{l\in[L]}, \mathbf D 
	\big) 	\nonumber \\    
	\overset{(b)}{=}  &  
	\left( \frac{1}{4!} \right)^{L}   
	\frac{(\bar N-L)!}{\bar N!}
	\frac{(N-\bar N)!}{N!}, 
\end{align} 
where $(a)$ follows from the definitions in \eqref{achB_DJ2}--\eqref{achB_DJ4}
and from the mutual independence of
$\alpha_{d_{0,l}}$, $l\in[L]$, $\mathbf Q$, and $\mathbf P$,   
the first two terms in $(b)$ follow from \eqref{achB_DU}, \eqref{achB_DJ2}, the uniformity of $\alpha_{d_{k,l}}$ over $\mathrm{Perm}([0:3])$, and
Lemma~\ref{lemma_Q},  
and the last term in $(b)$ follows from the
uniformity of \(\mathbf P\) over \(\mathrm{Perm}([N])\) and the fact that
\((J_{2,i})_{i\in[\bar N]}\) consists of \(\bar N\) distinct entries of
\(\mathbf P\). 

Finally, for $k\in \{0,1\}$, we have 
\begin{align} \label{achB_ppall}  
	& I \left( 	\mathbf{D}_{\setminus k}  ; X_{\mathbf{D}}    \big|  Z_k, \mathbf{d}_{k}  \right) 
=  I \left(\mathbf{D}_{\setminus k};  X_\mathbf{D}',  J    |Z'_k,  S_k , \mathbf{d}_{k}  \right)   \nonumber \\ 
\overset{(a)}{=}   & I \left(\mathbf{D}_{\setminus k} ;    (J_{u,l})_{u\in\{0,1\}, l\in[L]}, 
(J_{2,i})_{i\in[\bar N]}   \big |   (S_{k,l,i} )_{l\in[L], i \in [4]} , \mathbf{d}_{k}  \right)   \nonumber \\  
\overset{(b)}{=} 	  & I \big(\mathbf{D}_{\setminus k};   (A_{k,l,i})_{l \in [L], i \in [0:2]},    (T_{k,l})_{l \in [L]}, (J_{2,i})_{i\in[\bar N]} 	 \big |   \mathbf{d}_{k}   \big) \nonumber \\  
\overset{(c)}{=} &   0,    
\end{align}         
where $(a)$ follows from  \eqref{achB_pp2}, and the definitions in   \eqref{achB_Zk} and \eqref{achB_DJ1},  
$(b)$  follows from the fact that $ (A_{k,l,i})_{l \in [L], i \in [0:2]}$ and $(T_{k,l})_{l \in [L]}$ are encrypted using one-time pads  $ (S_{k,l,i} )_{l\in[L], i \in [0:3]}$   
$(c)$ follows from \eqref{achB_pp3}, which implies that 
$\left(  (A_{k,l,i})_{l \in [L], i \in [0:2]},    (T_{k,l})_{l \in [L]}, (J_{2,i})_{i\in[\bar N]}  \right)$ is independent of $\mathbf{D}$.     

The proof of privacy is thus complete, and the achievability of
$\big(\frac{N-\bar{N}+L}{3}, \frac{2\bar{N}}{3}\big)$ is established.

\subsection{Achievability of $\big(\frac{N}{3},  L\big)$ for $N \ge 2L$ }
\label{subsec_achC} 
We now present a scheme that achieves the memory-rate pair
$\big(\frac{N}{3},  L\big)$  
for the case \(N\ge 2L\). 
Throughout this subsection, since \(K=2\) and \(N\ge 2L\), we have \(\bar N=\min\{N,KL\}=2L\).
Several steps of the following scheme are the same as those in Subsection~\ref{subsec_achB}, to which we refer whenever such steps are reused.

\subsubsection*{Placement Phase} 
As in the scheme in Subsection~\ref{subsec_achB}, the server uses the 
$(4,3)$ Random Cauchy Encoding to obtain \(F_{n,0},F_{n,1},F_{n,2},F_{n,3}\),
\(n\in[N]\), draws
\(\mathbf P=(P_0,P_1,\dots,P_{N-1})^{\mathsf T}
\sim \mathrm{Unif}(\mathrm{Perm}([N]))\), and generates
\(S_k=(S_{k,l,i})_{l\in[L],\,i\in[0:3]}\in\mathbb F_q^{4L}\), \(k\in\{0,1\}\),
with i.i.d. entries uniformly distributed over \(\mathbb F_q\). 
The cache content of user
	\(k\) is \(Z_k=(Z_k',S_k)\), with the main load \(Z_k'\) defined by
	\begin{align} \label{achC_c1}
		Z_{k}'=  \left[ \left( Z_{k,i}\right)_{i \in [N]} 	 \right]^{\mathsf{T}}, \; \text{where}  \;    Z_{k,P_n} = F_{n, k}, \;   n\in[N].    
	\end{align}  
     
We note that \(Z_k'\) contains \(N\) segments, each of length  \(\frac{F}{3}\), and that the size of \(S_k\) is negligible when \(F\) is sufficiently large. Thus, the cache size is \(M=\frac{N}{3}\). 

\subsubsection*{Delivery Phase}
Recalling the auxiliary demand vector
\(\mathbf Q=(Q_0,Q_1,\ldots,Q_{2L-1})^{\mathsf T}\)   in Definition~\ref{def_Q}, we define 
\begin{align} \label{achC_refDY} 
	Y_{i,k}=   
	\begin{cases}
		F_{Q_i,k+2}, & \text{if } Q_i\in\mathcal{N}_k\\
		F_{Q_i,k},   & \text{if } Q_i\notin\mathcal{N}_k 
	\end{cases}, \;  k \in \{0,1\},   i\in [2L],  
\end{align}   
which coincides with \eqref{achB_DY}. 
The main load of the broadcast message is given by     
\begin{align}\label{achC_DX}
	X_{\mathbf{D}}'=\Big[\big(X_i\big)_{i\in[3L]}\Big]^{\mathsf{T}}
	=\mathbf{V}  \Big[\big(Y_{i,0},Y_{i,1}\big)_{i\in[2L]}\Big]^{\mathsf{T}}, 
\end{align}  
where
\(\mathbf V=[v_{i,j}]_{i\in[3L],\,j\in[4L]}\in\mathbb F_q^{3L\times 4L}\)
is a fixed Cauchy matrix known to all users.  
Thus, the broadcast message is  $X_{\mathbf{D}} = \left( X_{\mathbf{D}}', J \right)$, where  the auxiliary variable $J$ is defined as in Subsection~\ref{subsec_achB}, i.e., by
\eqref{achB_DJ1}-\eqref{achB_DJ4}, with \(\bar N=2L\). 
  
Since $X'_{\mathbf{D}} $ consists of $3L$ segments of length $\frac{F}{3}$ and the size of $J$ is negligible when $F$ is sufficiently large, the achievable rate is  given by    $R = L$.

\subsubsection*{Proof of Correctness} 
For each \(k \in \{0,1\}\) and \(l \in [L]\), user \(k\) decodes \(W_{d_{k,l}}\) similarly to the scheme in
Subsection~\ref{subsec_achB}.
The only difference lies in how the three
segments in \(\mathbf U_{k,l}\mathbf F_{d_{k,l}}\) are obtained.

More specifically,  the user $k$ recovers  the first two Cauchy-coded  segments,  i.e., $(\mathbf e_i^3)^{\mathsf T}\mathbf{U}_{k,l}\mathbf F_{d_{k,l}}$, $i \in \{0,1\}$,   as follows.
From \eqref{achB_DJ3}, user \(k\) recovers
\(T_{k,r}=J_{k,r,3}-S_{k,r,3}\), \(r\in[L]\), and computes
\begin{align}  \label{achC_pc1}
	\widetilde X_{k,i} 
	=
	X_i
	-
	\sum_{j\in[2L]\setminus\{T_{k,r}:r\in[L]\}}
	v_{i,2j+k}Z_{k,J_{2,j}}, \quad i\in[3L].
\end{align}  
User \(k\) then forms the matrix
\(\mathbf G_k=[g_{k,i,j}]_{i,j\in[3L]}\), where
\begin{align}  \label{achC_gk}
	g_{k,i,j}
	=
	\begin{cases}
		v_{i,2T_{k,j}+k}, & \text{if } j\in[L] \\
		v_{i,2(j-L)+1-k}, & \text{if } j\in[L:3L-1] 
	\end{cases}. 
\end{align}  
Since the columns of \(\mathbf V\) selected in \(\mathbf G_k\) are distinct,
\(\mathbf G_k\) is a square Cauchy matrix and is thus invertible.    
For \(m\in\{0,1\}\), user \(k\) obtains 
\begin{align} \label{achC_pc_e12}
	&   \left( \mathbf{e}_m^{3} \right)^{\mathsf{T}}    \mathbf{U}_{k,l}  \mathbf{F}_{d_{k,l}}
	\nonumber \\
	= 
	&  \begin{cases}
		\left( \mathbf G_k^{-1}
		\left[ \big( \widetilde{X}_{k,i} \big)_{i\in[3L]} \right]^{\mathsf T}
		\right)^{\mathsf T}
		\mathbf e_{l}^{3L},   & \text{if } m = k \\
		\left( \mathbf G_k^{-1}
		\left[ \big( \widetilde{X}_{k,i} \big)_{i\in[3L]} \right]^{\mathsf T}
		\right)^{\mathsf T}
		\mathbf e_{L+T_{k,l}}^{3L},   & \text{if } m \neq k 
	\end{cases}. 
\end{align} 
The correctness of \eqref{achC_pc_e12} is proved in Appendix~\ref{achC_pc_e12}.

The third Cauchy-coded  segment is obtained directly from the cache.  By \eqref{achB_DJ3},  \eqref{achB_DJ4} and     Definition~\ref{def_Q}, we have 
$J_{2,T_{k,l}} = P_{d_{k,l}} $. Using this index,  user $k$ obtains  
\begin{align} \label{achC_pc}
	Z_{k,J_{2,T_{k,l}}} = 
	Z_{k,P_{d_{k,l}}}
	\overset{(a)}{=}
	F_{d_{k,l},k} 
	\overset{(b)}{=}
	\left(\mathbf e_2^3\right)^{\mathsf T}
	\mathbf U_{k,l}\mathbf F_{d_{k,l}}, 
\end{align}  
where \((a)\) follows from \eqref{achC_c1}, and \((b)\) follows from the
definition of \(\mathbf U_{k,l}\) in  \eqref{achB_DU}.

Combining \eqref{achC_pc_e12} and \eqref{achC_pc}, user \(k\) obtains
\(\mathbf U_{k,l}\mathbf F_{d_{k,l}}\).  
The remaining steps are exactly the same as those in Subsection~\ref{subsec_achB}, where user  \(k\) decodes \(W_{d_{k,l}}\) using \eqref{achB_pc5} and \eqref{achB_pc6}. 
 
The proof of correctness is thus complete.

\subsubsection*{Proof of Privacy}

Similar to the proof of privacy for the scheme in Subsection~\ref{subsec_achB}, 
we first consider the main load of both the cache content and the broadcast message. 
For  $k \in \{0,1\}$, we have
\begin{align}  \label{achC_pp1}
& H \left(   Z_k' ,  X_{ \mathbf{D}}'  |  \mathbf{D}, J, S_k \right)  \nonumber \\ 
\overset{(a)} \ge
&   H  \left( Z_k',      (Y_{i,j}  )_{i \in [2L], j \in \{0,1\}}   \big|  \mathbf{D}, J, S_k    \right)   \nonumber \\
\overset{(b)} \ge 
&   H  \left(  (F_{n,k}  )_{n \in [N] },  (F_{n,k+2} )_{n \in \mathcal{N}_k},    (Y_{i,1-k}  )_{i \in [2L]}  \big|  \mathbf{D}, J, S_k    \right)   \nonumber \\
\overset{(c)}{=}
&     \frac{N + 3L}{3}F. 
\end{align} 
where $(a)$ follows from  the fact that $\left(Y_{i,j} \right)_{i \in [2L], j \in \{0,1\}}$ can be decoded from  $(  Z_k' ,  X_{ \mathbf{D}}')$ given $(\mathbf{d}_k, J, S_k)$,  $(b)$ follows  from  \eqref{achC_c1} and  \eqref{achC_refDY},
and $(c)$ follows from the fact that the segments in $(b)$ consist of $N + 3L$
distinct and independent Cauchy-coded segments, each with entropy $\frac{F}{3}$. 
 
Since \((Z_k',X_{\mathbf D}')\) consists of \(\frac{N+3L}{3}F \) symbols, 
\eqref{achC_pp1} implies that, for \(k\in\{0,1\}\) and any fixed realization of
\((\mathbf D,J,S_k)\), 
\begin{align}  \label{achC_pp2}
	\left( Z_k', X_{\mathbf D}' \big| \mathbf D, J, S_k \right)
	\sim \mathrm{Unif}\left(\mathbb F_q^{(N+3L)F/3}\right).
\end{align} 
  
With \eqref{achC_pp2} established, the remainder of the proof follows \eqref{achB_ppall}, with the only change being that step \((a)\) now follows from \eqref{achC_pp2} rather than \eqref{achB_pp2}.  
The remaining steps are unchanged since the auxiliary variable \(J\) is the same as in the scheme in Subsection~\ref{subsec_achB}, as shown in \eqref{achB_DJ1}-\eqref{achB_DJ4}.  
 
The proof of privacy is thus complete, and the achievability of
\(\big(\frac{N}{3},L\big)\) is established.

\subsection{Achievability of $\left(\frac{N^2}{2N-L}, \frac{L(N-L) }{2N-L} \right)$ for $N \ge 2L $}  \label{subsec_achD}
For the case \(N\ge 2L\), we now present a scheme that achieves the memory-rate pair $\left(\frac{N^2}{2N-L}, \frac{L(N-L) }{2N-L} \right)$.

\subsubsection*{Placement Phase}
The server first encodes each file \(W_n\), \(n\in[N]\), using the $(3N - 2L,2N-L)$ Random Cauchy Encoding in Definition~\ref{def_Cauchy}, thereby obtaining  $F_{n,0},F_{n,1}, \dots, F_{n,3N-2L-1}$.   
Next,  draw a permutation vector  $\mathbf P=(P_0,P_1,\dots,P_{N^2-1})^{\mathsf T}\sim \mathrm{Unif}(\mathrm{Perm}([N^2]))$, and  for  $k\in\{0,1\}$,   we set 
\begin{align} \label{achD_C} 
	Z'_{k} = & \left[  \left( Z_{k,i} \right)_{i \in [N^2]}   \right]^{\mathsf{T}}  ,  \; \text{where} \nonumber \\
	Z_{k,P_{Nn+j}} = 
	& \begin{cases}
		F_{n,j}  & \text{if } j \in [L]  \\ 
		F_{n,k(N-L)+j}, & \text{if }  j \in [L:N-1] 
	\end{cases}, \;  n\in[N].  
\end{align}      
The cache content of user $k$ is given by  
\begin{align*}
	Z_k = (Z_k', S_k), \; S_k=(S_{k,l,j})_{l\in[L],  j\in[2N+1]}, 
\end{align*} 
where the entries $S_{k,l,j}$ of $S_k$ are i.i.d. and uniformly distributed over $\mathbb F_q$.  
  
Since $Z_k'$ contains of  $N^2$ segments, each of length $\frac{F}{2N-L}$, and the size of $S_k$ is negligible for sufficiently large $F$, we have $M = \frac{N^2}{2N-L}$. 
\subsubsection*{Delivery Phase}  
The broadcast message is given by 
\begin{align*}
X_\mathbf{D} = (X'_\mathbf{D}, J), \quad \text{where } 	X_{\mathbf{D}}' = \left( X_i \right)_{i \in [N-L]},
\end{align*}    
and \(J\) is an auxiliary variable of negligible size.

To describe the components \(X_i\) of \(X'_{\mathbf D}\), recall that
\(\mathcal N_{\mathrm e}(\mathbf D)=\mathcal N_0\cup\mathcal N_1\) is the set
of requested files and that, since \(N\ge 2L\), the auxiliary demand vector 
\(\mathbf Q\) has length \(\bar N=2L\). We then define the set \(\Gamma\) and an injective map \(\sigma\) on \(\Gamma\) as follows   
	\begin{align*} 
		\Gamma
		&=  \{(i,j) \in [N-L]\times[2L]:
		Q_j \notin  \mathcal{N}_{\mathrm e}(\mathbf{D})  \}, \\
		\sigma
		&:\Gamma\to
		\big([N]\setminus\mathcal N_{\mathrm e}(\mathbf D)\big)\times[L],
		\;
		\sigma(i,j)=\big(\sigma_0(i,j),\sigma_1(i,j)\big).
	\end{align*}   
	with $\sigma_0(i,j)\in[N]\setminus\mathcal N_{\mathrm e}(\mathbf D)$ and $\sigma_1(i,j)\in[L]$. 
	The existence of such a map follows from the fact that $N \ge 2L \ge |\mathcal N_{\mathrm e}(\mathbf D)| \ge L$, which implies $|\Gamma| = (N-L)\bigl(2L-|\mathcal N_{\mathrm e}(\mathbf D)|\bigr) \le (N-|\mathcal N_{\mathrm e}(\mathbf D)|)L$.   For each \(i \in [N-L]\), we set  
\begin{align} 
	X_{i}   =  & \left[ \left(  X_{i,r} \right)_{r \in [L]}  \right]^{\mathsf{T}}   
	=   \mathbf{V}   \left[  \left(Y_{i,j} \right)_{j \in [2L]}  \right]^{\mathsf{T}},  \;  \text{where} \nonumber \\ 
	Y_{i,j} =  & 
	\begin{cases} 
		F_{Q_j, 2N-L+i},     &  \text{if}  \;   Q_j \in  \mathcal{N}_0 \cap \mathcal{N}_1  \\ 
		F_{Q_j, N+i},  &  \text{if}  \;    Q_j  \in  \mathcal{N}_0 \setminus  (  \mathcal{N}_0 \cap \mathcal{N}_1)    \\
		F_{Q_j, L+i},  & \text{if}  \;    Q_j  \in  \mathcal{N}_1 \setminus  (  \mathcal{N}_0 \cap \mathcal{N}_1)  \\  
		F_{\sigma_0(i,j), \sigma_1(i,j)},   & \text{if}  \;    Q_j  \notin    \mathcal{N}_{\mathrm e}(\mathbf{D}) 
	\end{cases},    \label{achD_XY} 
\end{align}   
and $  \mathbf{V}  = [v_{r,j}]_{ r \in [L],  j \in [2L]}$ is a fixed $ L \times 2L$  Cauchy matrix known to all users.

We now turn to describing the auxiliary variable $J$, which is given by  
\begin{align*}
	J = \left(  (J_{k,l})_{k \in \{0,1\}, l \in [L]},  (J_{2,i,j})_{i \in [N-L], j \in [2L]} \right).  
\end{align*} 
For $k \in \{0,1\}$ and $l \in [L]$,  $J_{k,l}$ contains the encrypted
information needed by user $k$ to decode its $l$-th requested file and consists
of $2N+1$ entries, i.e., $J_{k,l} =  (J_{k,l,j})_{j \in [2N+1]}$. 
The entries $J_{k,l,j}$ are defined over the following three index ranges:    
\begin{itemize}
	\item 
	For $j \in [2N-L]$, the definition of $J_{k,l,j}$ is given as follows.  
	We first define an auxiliary index vector $\mathbf u_{k,l} =(u_{k,l,0},u_{k,l,1},\ldots,u_{k,l,2N-L-1})^{\mathsf T}
	$, where $u_{k,l,j}$, $j \in [2N-L]$,   are given by 
	\begin{align} \label{achD_u1}
		u_{k,l,j} =
		\begin{cases}
			j,   \qquad  \qquad  \qquad  \qquad    \; \, \text{if } j \in [L] \\
			j+k(N-L),   \qquad  \quad  \text{if } j \in [L:N-1] \\ 
			j + N-L,   \\
			\quad  \text{if } j \in [N:2N-L-1], d_{k,l} \in \mathcal{N}_{1-k} \\ 
			j -k (N-L), \\ 
			\quad   \text{if } j \in [N:2N-L-1], d_{k,l} \notin \mathcal{N}_{1-k}    
		\end{cases}. 
	\end{align}  
	Based on $\mathbf u_{k,l}$, we define the selection matrix  
	\begin{align*}
		\mathbf{U}_{k,l}   = \mathbf{I}_{3N-2L} \left(\mathbf u_{k,l},:\right).   
	\end{align*}   	  
	Thus, $J_{k,l,j}$ is  given by   
	\begin{align} \label{achD_A} 
		&	J_{k,l,j} =   A_{k,l,j} + S_{k,l,j}, \quad j \in [2N-L],  \text{ where} \nonumber \\ 
		&   	\left(  A_{k,l,0},  A_{k,l,1}, \dots,  A_{k,l,2N-L-1} \right)^{\mathsf T}   =   	\mathbf{U}_{k,l}    \alpha_{d_{k,l}},    
	\end{align}     
	where   $\alpha_{d_{k,l}}$ is defined  in  Definition \ref{def_Cauchy} with $a=3N-2L$.

	\item For $j \in [2N-L:2N-1]$, we set
	\begin{align} \label{achD_JB}   
		J_{k,l,j}&
		= B_{k,l,j-(2N-L)} + S_{k,l,j},\text{ where} \nonumber \\   
		B_{k,l,j-(2N-L)}& 
		= P_{N d_{k,l} + j-(2N-L)} .
	\end{align}   
	\item For $j = 2N$,  we set    
	\begin{align*}  
		&	J_{k,l,2N} =   T_{k,l} + S_{k,l,2N},   \; \text{where} \; 
		T_{k,l}  = Q_{d_{k,l}}^{-1}.    
	\end{align*}       
\end{itemize}  
The entries  $J_{2,i,j}$  for $i \in [N-L]$ and $j \in [2L]$ are defined as follows   
\begin{align}    \label{achD_J2}
	J_{2,i,j} =   
	\begin{cases} 
		P_{N Q_j + L + i},     &  \text{if}  \;   Q_j \in  \mathcal{N}_{\mathrm e}(\mathbf{D})         \\  
		P_{\sigma_0(i,j) N  + \sigma_1(i,j)},   & \text{if}  \;    Q_j  \notin   \mathcal{N}_{\mathrm e}(\mathbf{D})  
	\end{cases}.   
\end{align}    
 
\begin{remark}[Observations] \label{re_achDJ} 
The following observations will be used in the proofs of correctness and privacy.
By \eqref{achD_u1}, the entries of $\mathbf u_{k,l}$ are distinct and lie in $[3N-2L]$, and hence $(A_{k,l,i})_{i\in[2N-L]}$ consists of $2N-L$ distinct entries of $\alpha_{d_{k,l}}$.
Moreover, it follows from $\eqref{achD_JB}$, $\eqref{achD_J2}$, and the definition of the map $\sigma$ that, for each $k\in{0,1}$, the collection  
	$ 
	\left( 
	(B_{k,l,i})_{l\in[L], i\in[L]},
	(J_{2,i,j})_{i\in[N-L], j\in[2L]}
	\right)
	$  consists of  $L(2N-L)$ distinct entries of $\mathbf P$. 
\end{remark}

For each $i\in[N-L]$, $X_i$ consists of $L$ segments $X_{i,r}$, $r\in[L]$, each of length $\frac{F}{2N-L}$. 
Thus, the main load $X'_{\mathbf D}$ consists of $L(N-L)$ such segments. Since the size of $J$ is negligible for sufficiently large $F$, the rate is $R=\frac{L(N-L)}{2N-L}$.

\subsubsection*{Proof of Correctness}
For each $k\in\{0,1\}$, $l\in[L]$,   we now describe the steps for  user $k$ to decode file $W_{d_{k,l}}$.

First, user \(k\) subtracts the corresponding one-time pads from the
entries of \(J_{k,l}\) to recover
\begin{align*}
	A_{k,l,j} = & J_{k,l,j}-S_{k,l,j}, \quad j\in[2N-L],  \\
	B_{k,l,i} = & J_{k,l,i+(2N-L)}-S_{k,l,i+(2N-L)},\quad i\in[L],  \text{ and} \\
	T_{k,l} =  & J_{k,l,2N}-S_{k,l,2N}. 
\end{align*}    

Second, user \(k\) directly obtains \(N\) different Cauchy-coded segments of
file \(W_{d_{k,l}}\) from its cache. Among them, \(L\) segments are obtained
using \(B_{k,l,i}\), \(i\in[L]\), and the other \(N-L\) segments are obtained
using \(J_{2,i,T_{k,l}}\), \(i\in[N-L]\). 
More specifically, for $i \in [L]$, we have  
\begin{align}  \label{ach_correct1}      
	Z_{k, B_{k,l,i}}	&  	\overset{(a)}{=}  
	F_{d_{k,l},i}  \overset{(b)}{=}    	\left( \mathbf{e}_{u_{k,l,i}}^{3N-2L} \right)^{\mathsf{T}}     \mathbf{F}_{d_{k,l}},     
\end{align}      
where $(a)$ follows from substituting \eqref{achD_JB} and then using \eqref{achD_C},      
and $(b)$ follows from  the definition of $ \mathbf{u}_{k,l}$ in \eqref{achD_u1}.   
For \(i \in [N-L]\), since \(Q_{T_{k,l}}=d_{k,l}\in
\mathcal N_{\mathrm e}(\mathbf D)\), it follows from \eqref{achD_J2} that
$J_{2,i,T_{k,l}}=P_{N d_{k,l}+L+i}$. 
Thus, user \(k\) obtains 
\begin{align}  \label{ach_correct2}       
	Z_{k,J_{2,i,T_{k,l}}}
	\overset{(a)}{=}
	F_{d_{k,l}, k(N-L)+L+i}
	\overset{(b)}{=}
	\left( \mathbf{e}_{u_{k,l,i+L}}^{3N-2L} \right)^{\mathsf{T}}
	\mathbf{F}_{d_{k,l}},  
\end{align} 
where $(a)$ follows from \eqref{achD_C}, and $(b)$ follows from the definition
of \(\mathbf u_{k,l}\) in \eqref{achD_u1}.   
 
Third, the user decodes another \(N-L\) Cauchy-coded segments of file \(W_{d_{k,l}}\) from the broadcast message using its cache content.  
	More specifically, using \(T_{k,m}\), \(m\in[L]\), and the known matrix
	\(\mathbf V\), user \(k\) forms 
	\begin{align*} 
		\mathbf{G}_{k}
		=
		[g_{k,r,m}]_{r \in [L], m \in [L]},
		\qquad
		g_{k,r,m}=v_{r,T_{k,m}}.  
	\end{align*} 
		For \(i\in[N-L]\) and \(r\in[L]\), user \(k\) computes
		\begin{align} \label{achD_Xwide}
			\widetilde X_{k,i,r}
			=&
			X_{i,r}
			-
			\sum_{j\in[2L]\setminus\{T_{k,m}:m\in[L]\}}
			v_{r,j}Z_{k,J_{2,i,j}}
			\nonumber\\
			\overset{(a)}{=}
			& 
			X_{i,r}
			-
			\sum_{j\in[2L]\setminus\{T_{k,m}:m\in[L]\}}
			v_{r,j}Y_{i,j}
			\nonumber\\ = 
			& 
			\sum_{j\in\{T_{k,m}:m\in[L]\}}
			v_{r,j}Y_{i,j} = 
			\sum_{m\in[L]}g_{k,r,m}Y_{i,T_{k,m}},
		\end{align} 
		where \((a)\)follows from \eqref{achD_C}, \eqref{achD_XY}, \eqref{achD_J2}, and the fact that \(Q_j\in\mathcal N_k\) if and only if \(j\in\{T_{k,m}:m\in[L]\}\).   
	Since \(\mathbf G_k\) is formed by \(L\) distinct columns of the Cauchy matrix
	\(\mathbf V\), it is invertible.  Thus, for \(i\in[N-L]\), user \(k\)
computes 
	\begin{align} \label{ach_correct3}         
		& \left(
		\mathbf G_k^{-1}
		\left[
		\big(\widetilde X_{k,i,r}\big)_{r\in[L]}
		\right]^{\mathsf T}
		\right)^{\mathsf T}
		\mathbf e_l^L
		\nonumber\\
		\overset{(a)}{=}\,
		& \left(
		\mathbf G_k^{-1}\mathbf G_k
		\left[
		\big(Y_{i,T_{k,m}}\big)_{m\in[L]}
		\right]^{\mathsf T}
		\right)^{\mathsf T}
		\mathbf e_l^L
		=
		Y_{i,T_{k,l}}
		\nonumber\\
		\overset{(b)}{=}\,
		& \left(\mathbf e_{u_{k,l,i+N}}^{3N-2L}\right)^{\mathsf T}
		\mathbf F_{d_{k,l}},
	\end{align}
	where \((a)\) follows from \eqref{achD_Xwide}, and \((b)\) follows from
	\eqref{achD_XY}, 	\eqref{achD_u1}, and the fact that \(Q_{T_{k,l}}=d_{k,l}\). 
 
	Finally, combining \eqref{ach_correct1}, \eqref{ach_correct2}, and
	\eqref{ach_correct3}, user \(k\) obtains
	\(\mathbf U_{k,l}\mathbf F_{d_{k,l}}\). 
	Using \((A_{k,l,i})_{i\in[2N-L]}\), user \(k\)  obtains     
	\begin{align*} 
		\mathbf U_{k,l}\mathbf C_{d_{k,l}}
		=
		\left[
		\frac{1}{A_{k,l,i}-(3N-2L+j)}
		\right]_{i\in[2N-L],\,j\in[2N-L]},
	\end{align*}
	which follows from \eqref{def_Cauchy_1} and \eqref{achD_A}. 
 	Since the entries of \(\mathbf u_{k,l}\) are distinct, as noted in
	Remark~\ref{re_achDJ}, the matrix
	\(\mathbf U_{k,l}\mathbf C_{d_{k,l}}\) is a
	\((2N-L)\times(2N-L)\) Cauchy matrix and is thus invertible.    
	User \(k\) then computes
	\begin{align*}
		& \left(\mathbf U_{k,l}\mathbf C_{d_{k,l}}\right)^{-1}
		\mathbf U_{k,l}\mathbf F_{d_{k,l}}
		\overset{(a)}{=}
		\left(\mathbf U_{k,l}\mathbf C_{d_{k,l}}\right)^{-1}
		\mathbf U_{k,l}\mathbf C_{d_{k,l}}\mathbf W_{d_{k,l}}
		\nonumber\\
		=  &
		\left[
		W_{d_{k,l},0},
		W_{d_{k,l},1},
		\dots, 
		W_{d_{k,l},2N-L-1}
		\right]^{\mathsf T},
	\end{align*}
	where \((a)\) follows from \eqref{def_Cauchy_2}.
	Thus, user $k$ can concatenate  the recovered subfiles to obtain $W_{d_{k,l}}$.  
	
	The proof of correctness is thus complete.

\subsubsection*{Proof of Privacy}    
First, for $k \in \{0,1\}$, we have     
 \begin{align} \label{achD_pp1}   
		& H  \left( Z'_k, X'_{\mathbf{D}} \big|\mathbf{D}, J, S_k \right)  	\overset{(a)}{ \ge}     H  \left( (W_n)_{n \in \mathcal{N}_k}, Z'_k \big|\mathbf{D}, J, S_k \right)  \nonumber \\  
		\overset{(b)}{\ge}  &   LF   + \sum_{n \in [N] \setminus \mathcal{N}_k}  H  \left(  (F_{n,j})_{j \in [L] \cup [k(N-L)+L:k(N-L)+N-1 ]}   \right) \nonumber \\ 
		\overset{(c)}{=} &  LF +  \frac{N(N-L)F}{2N-L}  =   \frac{N^2+LN-L^2 }{2N-L}F, 
\end{align} 
where $(a)$ follows from the fact that $(W_n)_{n \in \mathcal{N}_k}$  can be decoded from \(Z'_k\) and \(X'_{\mathbf D}\) given \((\mathbf d_k,J,S_k)\), 
$(b)$  follows from \eqref{achD_C} and the fact that $W_0, W_1, \dots, W_{N-1}$ are independent and each has $F$ symbols,  and $(c)$ follows from the fact that the Cauchy-coded segments in the second term of $(b)$ are linearly independent and are each of length $\frac{F}{2N-L}$.

Since $H \bigl( Z'_k, X'_{\mathbf{D}} \big| \mathbf{D}, J, S_k \bigr) \le (M+R)F = \frac{N^2+LN-L^2}{2N-L}F$, \eqref{achD_pp1} implies that, for $k\in\{0,1\}$ and any fixed realization of $(\mathbf{D},J,S_k)$,   
\begin{align} \label{achD_pp2}  
	\left( Z'_k, X'_{\mathbf{D}} \big|\mathbf{D}, J, S_k  \right) \sim \text{Unif} \Big(  \mathbb{F}_q^{ \frac{N^2+LN-L^2}{2N-L}F}   \Big).  
\end{align}  
  
Next, for any fixed demand matrix $\mathbf{D} \in \mathcal{D}$, any  distinct $ t_{k,l} \in [\bar{N}]$ with $l \in [L]$,  any distinct $ p_{i} \in [N^2]$ with $i \in [L(2N-L)]$, and any  $ a_{k,l,i} \in [3N-2L]$ such that, for each $l \in [L]$, $(a_{k,l,i})_{i \in [2N-L]}$ are distinct, we have  
\begin{align} \label{achD_pp3} 
	& \Pr \big(  (T_{k,l})_{l \in [L]}  = (t_{k,l})_{l \in [L]},(B_{k,l,i})_{l \in [L], i \in [L]} =  (p_i)_{i \in [L^2]},  \nonumber \\ 
	& \qquad   ( J_{2,i,j})_{i \in [N-L], j \in [2L] }  = (p_i)_{i \in [L^2: L(2N-L)-1]}, 
	\nonumber \\ 	& \qquad \quad   (A_{k,l,i})_{l\in [L], i \in [2N-L]} = (a_{k,l,i})_{l\in [L], i \in [2N-L]}   \big| \mathbf {D} \big)  \nonumber \\    
	\overset{(a)}{=} 	& \prod_{l\in[L]}   \Pr  \big(   (A_{k,l,i})_{i\in[2N-L]}=(a_{k,l,i})_{i\in[2N-L]}  \big| \mathbf {D} \big)  \nonumber \\ 
	&  \times   \Pr  \big(
	(T_{k,l})_{l\in[L]}=(t_{k,l})_{l\in[L]} \big| \mathbf {D} \big)	 \nonumber \\ 
	& \times  \Pr  \Big(  \left(  (B_{k,l,i})_{l \in [L], i \in [L]},  ( J_{2,i,j})_{i \in [N-L], j \in [2L] }  \right)   \nonumber \\ 
	&  \quad =  (p_i)_{i \in [L(2N-L)]}   \Big|(T_{k,l})_{l\in[L]}=(t_{k,l})_{l\in[L]} , \mathbf {D}  \Big) 	\nonumber \\      
	\overset{(b)}{=}  &  \left( \frac{(N-L)!}{(3N-2L)!}   \right)^L   \frac{(\bar{N}-L)!}{\bar{N}!}  \frac{(N^2-2NL+L^2)!}{(N^2)!},      
\end{align}  
where $(a)$ follows from the definitions in \eqref{achD_u1}--\eqref{achD_J2}
and from the mutual independence of
\(\alpha_{d_{k,l}}\), \(l\in[L]\), \(\mathbf Q\), and \(\mathbf P\), 
	 the first term in $(b)$ follows from the uniformity of
	\(\alpha_{d_{k,l}}\) over \(\mathrm{Perm}([3N-2L])\) and the distinctness of
	\((A_{k,l,i})_{i\in[2N-L]}\) stated in Remark~\ref{re_achDJ},  
	the second term in $(b)$ follows from Lemma~\ref{lemma_Q}, 
	and the last term in $(b)$ follows from the uniformity
	of \(\mathbf P\) over \(\mathrm{Perm}([N^2])\) and the fact that
	$
	\left(
	(B_{k,l,i})_{l\in[L],i\in[L]},
	(J_{2,i,j})_{i\in[N-L],j\in[2L]}
	\right)
$
	consists of \(L(2N-L)\) distinct entries of \(\mathbf P\).

Finally, for $k \in \{0,1\}$, we have 	
\begin{align*}   
	&  	I  \left( \mathbf{D}_{ \setminus k} ; X_{\mathbf{D}} \big| Z_k, \mathbf{d}_k \right)    
	=  	I  \left( \mathbf{D}_{ \setminus k} ; X'_{\mathbf{D}}, J \big| Z'_k,  S_k, \mathbf{d}_k \right)    \nonumber \\ 
	\overset{(a)}{=}    &   I  \Big( \mathbf{D}_{ \setminus k} ;   \left(  (J_{0,l},J_{1,l})_{ l \in [L]},  (J_{2,i,j})_{i \in [N-L], j \in [2L]} \right)  \nonumber \\
	& \qquad \qquad  \qquad \qquad  \Big| (S_{k,l,j})_{l\in[L],  j\in[2N+1]},  \mathbf{d}_k  \Big)  \nonumber \\ 
	\overset{(b)}{=}   &   I  \big( \mathbf{D}_{ \setminus k} ;   (A_{k,l,i})_{l\in [L], i \in [2N-L]}, (B_{k,l,i})_{l \in [L], i \in [L]}, \nonumber \\
	& \qquad   (T_{k,l})_{l \in [L]},  ( J_{2,i,j})_{i \in [N-L], j \in [2L] }   \big|   \mathbf{d}_k \big)  \overset{(c)}{=}   0,   
\end{align*}             
where $(a)$ follows from  \eqref{achD_pp2}, and the definitions of $J$ and $S_k$, 
$(b)$  follows from the fact that $ (A_{k,l,i})_{l\in [L], i \in [2N-L]}$, $(B_{k,l,i})_{l \in [L], i \in [L]}$,  and $(T_{k,l})_{l \in [L]}$ are encrypted using one-time pads  $ (S_{k,l,j})_{l\in[L],  j\in[2N+1]}$, 
and   $(c)$ follows from \eqref{achD_pp3}, which implies that  the variables in the mutual information term  is independent of $\mathbf{D}$.

The proof of privacy is thus complete, and the achievability of
 $\left(\frac{N^2}{2N-L}, \frac{L(N-L) }{2N-L} \right)$  is established.

 \subsection{Running Example} 
 We provide a running example to illustrate the proposed schemes.
 Consider a system with $K=2$, $L=2$, and $N=5$, and suppose that
 $\mathbf D= [0,1;0,2]  $.
 Then $\mathcal N_0=\{0,1\}$, $\mathcal N_1=\{0,2\}$, and
 $\bar N=\min\{N,KL\}=4$. 
Choose one possible realization of the auxiliary demand vector as
\[
\mathbf Q=(2,0,3,1)^{\mathsf T}\in\mathcal D_{\bar N}(\mathbf D),
\]
whose inverse positions are
$Q^{-1}_0=1$, $Q^{-1}_1=3$, $Q^{-1}_2=0$, and $Q^{-1}_3=2$.
We next illustrate the three proposed schemes under this common demand and auxiliary vector.
 
\subsubsection{Example for the Memory-Rate Pair $(M,R)=\left(1,\frac{8}{3}\right)$}
In the scheme described in Subsection~\ref{subsec_achB}, applying the $(4,3)$
Random Cauchy Encoding to each file $W_n$ yields the Cauchy-coded  segments
$F_{n,0},F_{n,1},F_{n,2},F_{n,3}$.  
Since \(N-\bar N+L=3\), \(\mathbf H\in\mathbb F_q^{3\times 5}\) is a random 
Cauchy matrix with 
\begin{align*} 
	H_{i,j}
	=
	\frac{1}{5+i-P_j},
	\qquad i\in[0:2],\; j\in[0:4],
\end{align*}
where \(\mathbf P\sim\mathrm{Unif}(\mathrm{Perm}([0:4]))\). User \(k\) caches
\begin{align*} 
	Z_k=(Z_k',S_k),
	\qquad
	Z_k'=\mathbf H\left[(F_{n,k})_{n\in[0:4]}\right]^{\mathsf T},
\end{align*}
where \(S_k=(S_{k,l,i})_{l\in\{0,1\},\,i\in[0:3]}\) denotes the collection of one-time pads.   Since  \(\mathbf H\in\mathbb F_q^{3\times 5}\), the cache size is \(M= 1\).  
  
Under the above realization of \(\mathbf Q\), the server delivers
\(X_{\mathbf D}=(X_{\mathbf D}',J)\), where 
\begin{align*}
	X_{\mathbf D}'= & \left[ (X_i)_{i \in [0:7]}  \right]^{\mathsf T} \nonumber \\
= &	\left[ 
	F_{2,0},F_{2,3},
	F_{0,2},F_{0,3},
	F_{3,0},F_{3,1},
	F_{1,2},F_{1,1}
	\right]^{\mathsf T},
\end{align*} 
and the auxiliary variable \(J\) is given in
Table~\ref{tab:running_achB_J}.  The achieved rate is \(R= \frac{8}{3}\).   
\begin{table}[t]
	\centering
	\caption{Auxiliary variable \(J\) for the two schemes achieving
		\((M,R)=\left(1,\frac{8}{3}\right)\) and
		\((M,R)=\left(\frac{5}{3},2\right)\).}
	\label{tab:running_achB_J}
	\small
	\renewcommand{\arraystretch}{1.05}
	\begin{tabular}{|c|c|c|c|}
		\hline
		\multicolumn{4}{|c|}{
			\(J=\bigl((J_{k,l,i})_{k\in\{0,1\}, l\in\{0,1\}, i \in [0:3]},
			(J_{2,i})_{i\in[0:3]}\bigr)\) 
		} \\ \hline  
		\multicolumn{4}{|c|}{
			\(J_{k,l,i}=A_{k,l,i}+S_{k,l,i},\ i\in[0:2],
			\; J_{k,l,3}=T_{k,l}+S_{k,l,3}\)
		} \\ \hline
		\((k,l)\)
		& $d_{k,l}$
		& \((A_{k,l,i})_{i\in[0:2]}\)
		& \(T_{k,l}\) \\ \hline
		\((0,0)\)  
		& \(0\)
		& \((\alpha_{0,2},\alpha_{0,3},\alpha_{0,0})\)
		& \(1\) \\ \hline
		\((0,1)\)
		& \(1\)
		& \((\alpha_{1,2},\alpha_{1,1},\alpha_{1,0})\)
		& \(3\) \\ \hline
		\((1,0)\)
		& \( 0\)
		& \((\alpha_{0,2},\alpha_{0,3},\alpha_{0,1})\)
		& \(1\) \\ \hline
		\((1,1)\)
		& \( 2\)
		& \((\alpha_{2,0},\alpha_{2,3},\alpha_{2,1})\)
		& \(0\) \\ \hline
		\multicolumn{2}{|c|}{\((J_{2,i})_{i\in[0:3]}\)}
		& \multicolumn{2}{c|}{\((P_2,P_0,P_3,P_1)\)} \\ \hline
	\end{tabular}
\end{table}

To illustrate the decoding, we take \(k=0\) and \(l=1\) as an example, where
user \(0\) decodes \(W_1\).  
Using \(J\) and \(S_0\), user \(0\) obtains \(T_{0,0}\), \(T_{0,1}\), and
\((A_{0,1,i})_{i\in[0:2]}\), as shown in
Table~\ref{tab:running_achB_J}.
Next, using \(T_{0,1}=3\), user \(0\) obtains two components from
\(X_{\mathbf D}'\), i.e., 
\[
X_{2T_{0,1}}=X_6=F_{1,2},
\qquad
X_{2T_{0,1}+1}=X_7=F_{1,1}.
\]   
It remains to recover \(F_{1,0}\). Since \(T_{0,0}=1\) and \(T_{0,1}=3\), 
we have \( [0:3]\setminus\{T_{0,0},T_{0,1}\}=\{0,2\}\). For
$i \in [0:2]$, user \(0\)  computes    
	\begin{align}
		\label{running_dec_Ztilde}
		\widetilde Z_{0,i}
		&=		Z_{0,i} 	- 	\frac{1}{5+i-J_{2,0}}X_0	-
		\frac{1}{5+i-J_{2,2}}X_4
		\nonumber\\
		& = 
	Z_{0,i} -H_{i,2}F_{2,0}-	H_{i,3}F_{3,0} \nonumber\\
		&=	H_{i,0}F_{0,0}	+	H_{i,1}F_{1,0}	+	H_{i,4}F_{4,0}.  
	\end{align}  
Since \(J_{2,T_{0,0}}=J_{2,1}=P_0\), \(J_{2,T_{0,1}}=J_{2,3}=P_1\), and  $P_4$ is the only element of $[0:4]$ not appearing in $(J_{2,i})_{i\in[0:3]}$,   user \(0\) forms
\(\mathbf G_0=[g_{0,i,j}]_{i,j\in[0:2]}\) according to
\eqref{achB_pcg1} and \eqref{achB_pcg2},  where 
\begin{align*}
	g_{0,i,0}
	&=
	\frac{1}{5+i-P_0},   \quad 
	g_{0,i,1} =	\frac{1}{5+i-P_1}, \\
	g_{0,i,2}
	&=
	\frac{1}{5+i-P_4},
	\qquad i\in[0:2].
\end{align*} 
Combining this with   \eqref{running_dec_Ztilde}, we have 
\begin{align*} 
\left[
(\widetilde Z_{0,i})_{i\in[0:2]}
\right]^{\mathsf T}
=
\mathbf G_0
\left[
F_{0,0}, F_{1,0},F_{4,0} 
\right]^{\mathsf T},   
\end{align*}
from which user \(0\) recovers \(F_{1,0}\) since \(\mathbf G_0\) is invertible.  Together with $F_{1,2}$  and $F_{1,1}$ obtained from $X_{\mathbf D}'$, user $0$ now has three Cauchy-coded segments of $W_1$. 
Using the recovered $(A_{0,1,i})_{i\in[0:2]}$, these segments satisfy 
\begin{align*} 
\begin{bmatrix}
F_{1,2}\\
F_{1,1}\\
F_{1,0}
\end{bmatrix}
=
\left[
\frac{1}{A_{0,1,i}-(4+j)}
\right]_{i,j\in[0:2]}
\begin{bmatrix}
W_{1,0}\\
W_{1,1}\\
W_{1,2}
\end{bmatrix}.
\end{align*} 
Thus, user $0$ recovers $W_{1,0}$, $W_{1,1}$, and $W_{1,2}$ and decodes $W_1$.   The remaining requested file of user $0$ and the requested files of user $1$ can be decoded similarly.

\subsubsection{Example for the Memory-Rate Pair $(M,R)=\left(\frac{5}{3},2\right)$}
For the scheme  in Subsection~\ref{subsec_achC}, each file $W_n$ is encoded into the four Cauchy-coded segments $F_{n,0},F_{n,1},F_{n,2},F_{n,3}$. 
In the placement phase, user \(k\in\{0,1\}\) caches
\(Z_k=(Z_k',S_k)\), where
\[
Z_k'=\left[(Z_{k,i})_{i\in[0:4]}\right]^{\mathsf T},
\qquad
Z_{k,P_n}=F_{n,k},\quad n\in[0:4].
\] 
Since \(Z_k'\) consists of \(5\) Cauchy-coded segments, the cache size is \(M= \frac{5}{3}\).

Under the above realization of \(\mathbf Q\), the server delivers
\(X_{\mathbf D}=(X_{\mathbf D}',J)\), where \(J\) is summarized in
Table~\ref{tab:running_achB_J}, and
\begin{align*}
	X_{\mathbf D}'
	&=\left[(X_i)_{i\in[0:5]}\right]^{\mathsf T} \\
	&=
	\mathbf V
	\left[
	F_{2,0},F_{2,3},
	F_{0,2},F_{0,3},
	F_{3,0},F_{3,1},
	F_{1,2},F_{1,1}
	\right]^{\mathsf T}. 
\end{align*} 
The matrix $\mathbf V$ is a fixed $6\times 8$ Cauchy matrix known to all users, and the resulting rate is $R=2$.

To illustrate the decoding, we   take \(k=0\) and \(l=1\) as an example.   
 Using \(J\) and \(S_0\), user \(0\) obtains
\(T_{0,0}\), \(T_{0,1}\), and \((A_{0,1,i})_{i\in[0:2]}\). 
From Table~\ref{tab:running_achB_J}, we have 
\([0:3]\setminus\{T_{0,0},T_{0,1}\}=\{0,2\}\). For  
\(i\in[0:5]\), user \(0\) computes 
\begin{align}
 	\label{running_dec_Xtilde_achC}
 	\widetilde X_{0,i}
 	&=
 	X_i
 	-
 	v_{i,0}Z_{0,J_{2,0}}
 	-
 	v_{i,4}Z_{0,J_{2,2}}
 	\nonumber\\
 	& = 
 	X_i
 	-
 	v_{i,0}F_{2,0}
 	-
 	v_{i,4}F_{3,0}
 	\nonumber\\
	&=
v_{i,1}F_{2,3}
+v_{i,2}F_{0,2}
+v_{i,3}F_{0,3}
+v_{i,5}F_{3,1}
\nonumber\\
&\quad
+v_{i,6}F_{1,2}
+v_{i,7}F_{1,1}.  
 \end{align} 
 Using \eqref{achC_gk}, $T_{0,0}=1$, and $T_{0,1}=3$, user $0$ constructs the matrix $\mathbf G_0=[g_{0,i,j}]_{i,j\in[0:5]}$, where 
\begin{align*}
(g_{0,i,j})_{j\in[0:5]} 
=
(v_{i,2},v_{i,6},v_{i,1},v_{i,3},v_{i,5},v_{i,7}),
\quad i\in[0:5].
 \end{align*}    
Combining this with \eqref{running_dec_Xtilde_achC}, we have
\begin{align*}
 \left[
 (\widetilde X_{0,i})_{i\in[0:5]}
 \right]^{\mathsf T}
 =
 \mathbf G_0
 \left[
 F_{0,2},F_{1,2},F_{2,3},F_{0,3},F_{3,1},F_{1,1}
 \right]^{\mathsf T}. 
 \end{align*}    
 Since $\mathbf G_0$ is invertible, user $0$ recovers the desired segments $F_{1,2}$ and $F_{1,1}$ at positions $l=1$ and $L+T_{0,1}=5$, respectively. 
Moreover, user \(0\) obtains $F_{1,0}$ from its cache, i.e., 
\begin{align*} 
	Z_{0,J_{2,T_{0,1}}}=Z_{0,J_{2,3}}=Z_{0,P_1}=F_{1,0}. \end{align*}   
Thus, user $0$ obtains $F_{1,2}$, $F_{1,1}$, and $F_{1,0}$ and, using
$(A_{0,1,i})_{i\in[0:2]}$, decodes $W_1$ as in the previous example. 
The remaining requested file of user $0$ and the requested files of user $1$ can be decoded similarly.

\subsubsection{Example for the Memory-Rate Pair 
	\(\left(M,R\right)=\left(\frac{25}{8},\frac{3}{4}\right)\)} 
For the scheme  in Subsection~\ref{subsec_achD}, applying the
\((3N-2L,2N-L) = (11,8)\) Random Cauchy Encoding to each file \(W_n\) yields the Cauchy-coded
segments \(F_{n,0},F_{n,1},\ldots,F_{n,10}\).   
The cache content of user \(k\) is
\begin{align*}
Z_k=(Z_k',S_k),
\qquad
Z_k'=\left[(Z_{k,i})_{i\in[0:24]}\right]^{\mathsf T}, \quad k \in\{0,1\},  
\end{align*}  
where \(S_k=(S_{k,l,j})_{l\in\{0,1\},\,j\in[0:10]}\) collects the one-time pads cached by user $k$.  Let  $
\mathbf P=(P_0,P_1,\ldots,P_{24})^{\mathsf T}
\sim \mathrm{Unif}(\mathrm{Perm}([0:24])).
$
Then, for each \(n\in[0:4]\), the  components of \(Z_k'\) are given by      
\begin{align*} 
	\left(Z_{0,P_{5n+j}}\right)_{j\in[0:4]}
	&=
	\left(F_{n,0},F_{n,1},F_{n,2},F_{n,3},F_{n,4}\right),\\
	\left(Z_{1,P_{5n+j}}\right)_{j\in[0:4]}
	&=
	\left(F_{n,0},F_{n,1},F_{n,5},F_{n,6},F_{n,7}\right).
\end{align*}  
Thus, each user caches five  Cauchy-coded segments per file and $25$ segments in total, yielding $ M=\frac{25}{8}$.

We now describe the broadcast message for the chosen
\(\mathbf Q=(2,0,3,1)^{\mathsf T}\) and
\(\mathbf D=[0,1; 0,2]\).   
The only entry of
\(\mathbf Q\) that is not in \(\mathcal N_{\mathrm e}(\mathbf D)\) is 
\(Q_2=3\), and thus \(\Gamma=\{(0,2),(1,2),(2,2)\}\). We then choose one possible  
map \(\sigma\) as    
\begin{align*} 
		\sigma(0,2)=(3,0),\quad
	\sigma(1,2)=(3,1),\quad
	\sigma(2,2)=(4,0).
\end{align*}
 The server delivers $  X_{\mathbf D} =(X_{\mathbf D}',J) 
=\left((X_i)_{i\in[0:2]},J\right) $,  where
\begin{align*}
	X_0
	&=[X_{0,0},X_{0,1}]^{\mathsf T}
	=\mathbf V\left[ F_{2,2},F_{0,8},F_{3,0},F_{1,5}\right]^{\mathsf T},\\
	X_1
	&=[X_{1,0},X_{1,1}]^{\mathsf T}
	=\mathbf V\left[ F_{2,3},F_{0,9},F_{3,1},F_{1,6}\right]^{\mathsf T},\\
	X_2
	&=[ X_{2,0},X_{2,1}]^{\mathsf T}
	=\mathbf V\left[ F_{2,4},F_{0,10},F_{4,0},F_{1,7}\right]^{\mathsf T},
\end{align*} 
and \(\mathbf V=[v_{r,j}]_{r\in\{0,1\},\,j\in[0:3]}\) is a fixed
\(2\times4\) Cauchy matrix known to  all users. 
The auxiliary variable \(J\) is summarized in
Table~\ref{tab:running_achD_J}.   Since \(X_{\mathbf D}'\) consists of \(6\) Cauchy-coded segments, each of size \(\frac{F}{8} \), and the size of \(J\) is negligible, the achieved  rate is \(R =\frac{3}{4}\).    
\begin{table*}[!t]
	\centering
	\caption{Auxiliary variable \(J\) for the scheme achieving
		\((M,R)=\left(\frac{25}{8},\frac{3}{4}\right)\).}
	\label{tab:running_achD_J}
	\footnotesize
	\renewcommand{\arraystretch}{1.15}
	\begin{tabular}{|c|c|c|c|c|}
		\hline
		\multicolumn{5}{|c|}{
			\(\displaystyle
			J=
			\left(
			(J_{k,l,j})_{k\in\{0,1\}, l\in\{ 0,1 \},j\in[0:10]},
			(J_{2,i,j})_{i\in[0:2],\,j\in[0:3]}
			\right)
			\)} 
		\\ \hline
		\multicolumn{5}{|c|}{
			\(\displaystyle
			J_{k,l,j}=A_{k,l,j}+S_{k,l,j},  j\in[0:7],\quad
			J_{k,l,8+i}=B_{k,l,i}+S_{k,l,8+i},  i\in\{0,1\},\quad
			J_{k,l,10}=T_{k,l}+S_{k,l,10}
			\)}
		\\ \hline 
		\((k,l)\)
		& \(d_{k,l}\) 
		& \((A_{k,l,j})_{j\in[0:7]}\)
		& \((B_{k,l,i})_{i\in \{0,1\}}\)
		& \(T_{k,l}\) \\ \hline
		\((0,0)\)
		& \(0\)
		& \((\alpha_{0,0},\alpha_{0,1},\alpha_{0,2},\alpha_{0,3},
		\alpha_{0,4},\alpha_{0,8},\alpha_{0,9},\alpha_{0,10})\)
		& \((P_0,P_1)\)
		& \(1\) \\ \hline
		\((0,1)\)
		& \(1\)
		& \((\alpha_{1,0},\alpha_{1,1},\alpha_{1,2},\alpha_{1,3},
		\alpha_{1,4},\alpha_{1,5},\alpha_{1,6},\alpha_{1,7})\)
		& \((P_5,P_6)\)
		& \(3\) \\ \hline
		\((1,0)\)
		& \(0\)
		& \((\alpha_{0,0},\alpha_{0,1},\alpha_{0,5},\alpha_{0,6},
		\alpha_{0,7},\alpha_{0,8},\alpha_{0,9},\alpha_{0,10})\)
		& \((P_0,P_1)\)
		& \(1\) \\ \hline
		\((1,1)\)
		& \(2\)
		& \((\alpha_{2,0},\alpha_{2,1},\alpha_{2,5},\alpha_{2,6},
		\alpha_{2,7},\alpha_{2,2},\alpha_{2,3},\alpha_{2,4})\)
		& \((P_{10},P_{11})\)
		& \(0\) \\ \hline
		\multicolumn{2}{|c|}{\(\left((J_{2,i,j})_{j\in[0:3]}\right)_{i\in[0:2]}\)}
		& \multicolumn{3}{c|}{
			\(\big((P_{12},P_2,P_{15},P_7),
			(P_{13},P_3,P_{16},P_8),
			(P_{14},P_4,P_{20},P_9)\big)\)}
		\\ \hline
	\end{tabular}
\end{table*}  

To illustrate the decoding, we  take \(k=0\) and \(l=1\) as an example.    
From \(J\) and \(S_0\), user \(0\) obtains
\(T_{0,0}=1\), \(T_{0,1}=3\),
\((B_{0,1,0},B_{0,1,1})=(P_5,P_6)\),
\((A_{0,1,i})_{i\in[0:7]}=(\alpha_{1,0},\alpha_{1,1},\ldots,\alpha_{1,7})\), 
and
\((J_{2,i,T_{0,1}})_{i\in[0:2]}= (J_{2,i,3})_{i\in[0:2]} = (P_7,P_8,P_9)\). 
Using the cache indices $B_{0,1,0}$, $B_{0,1,1}$, and
$(J_{2,i,3})_{i\in[0:2]}$, user $0$ obtains
$(Z_{0,P_5},Z_{0,P_6})=(F_{1,0},F_{1,1})$ and
$(Z_{0,P_7},Z_{0,P_8},Z_{0,P_9})=(F_{1,2},F_{1,3},F_{1,4})$
from its cache. 
Next, for \(i\in[0:2]\) and \(r\in\{0,1\}\), user \(0\) computes  
\begin{align*}
\widetilde X_{0,i,r}
=
X_{i,r}
-
v_{r,0}Z_{0,J_{2,i,0}}
-
v_{r,2}Z_{0,J_{2,i,2}}, 
\end{align*} 
where the subtracted terms correspond to
\(j\in[0:3]\setminus\{T_{0,0},T_{0,1}\}=\{0,2\}\). 
By the definition of \(X_i\), we have     
\begin{align*}
\begin{bmatrix}
	\widetilde X_{0,i,0}\\
	\widetilde X_{0,i,1}
\end{bmatrix}
=
\begin{bmatrix}
	v_{0,1} & v_{0,3}\\
	v_{1,1} & v_{1,3}
\end{bmatrix}
\begin{bmatrix}
	F_{0,8+i}\\
	F_{1,5+i}
\end{bmatrix},
\qquad i\in[0:2].
\end{align*}  
Since \([v_{r,j}]_{r\in\{0,1\},\,j\in\{1,3\}}\) is a square   submatrix of
\(\mathbf V\), it is invertible.  
Thus, user $0$ decodes $F_{1,5+i}$, $i\in[0:2]$, and, together with the
segments previously obtained from its cache, recovers
$(F_{1,j})_{j\in[0:7]}$. 
The recovered segments satisfy 
\begin{align*}  
	& \left[ (F_{1,j})_{j\in[0:7]}\right]^{\mathsf T} \nonumber \\
	= &  
	\left[
	\frac{1}{A_{0,1,i}-(11+j)}
	\right]_{i,j\in[0:7]}
	\left[ (W_{1,j})_{j\in[0:7]}\right] ^{\mathsf T}. 
\end{align*}  
Since the coefficient matrix is a square Cauchy matrix, user $0$ recovers all subfiles of $W_1$ and hence decodes $W_1$.  The remaining requested file of user $0$ and the requested files of user $1$ can be decoded similarly.

 \section{Proof of Theorem~\ref{optimalK2}: Converse}  
\label{sec_optimalK2Con}   
In this section, we prove the converse part of
Theorem~\ref{optimalK2}. 
By specializing Theorem~\ref{th_converse} to \(K=2\),
we directly obtain three converse bounds. More specifically, choosing \(s=1\)
and \(\lambda=1\), we get \(LM+NR\ge LN\). Choosing \(s=2\) and
\(\lambda=1\), we get \(3LM+NR\ge 2NL\) for \(N\ge 3L\) and
\(2LM+(N-L)R\ge 2L(N-L)\) for \(2L \le N<3L\), with \(t=1\) and \(t=2\), respectively.

To complete the converse proof, it remains to establish Bounds~A-D.  After introducing the necessary notation and preliminary results in Subsection~\ref{sec:conv-prelim}, we establish Bounds~A and~B in Subsection~\ref{sec:conv-Nge2L}, and Bounds~C and~D in Subsections~\ref{sec:conv-Nle2L-C} and~\ref{sec:conv-Nle2L-D}, respectively.

\subsection{Preliminaries} \label{sec:conv-prelim} 
Throughout the converse proof for \(K=2\), we consider the following class of users' demands.    
For each \(i \in [N]\), let 
\begin{align*}
	 \mathbf{b}_i =  & \big( 	(iL) \bmod N,  	(iL+1) \bmod N,  \dots,  \nonumber \\ 
	 &
	\quad  (iL+L-1) \bmod N
	 \big)^{\mathsf T},   
\end{align*} 
and define the corresponding set of requested files as 
\begin{align*}
\mathcal B_i
=
\{(iL+l)\bmod N : l \in [L]\}.
\end{align*}  
For any \(i,j \in [N]\), define the demand matrix as   
\begin{align}
	\mathbf D^{(i,j)} = [\mathbf b_i, \mathbf b_j]^{\mathsf T}. 
	\label{def_Dij}
\end{align}      
 
The following lemma states that, for either user, the joint distribution of the information available to that user does not depend on the other user's demand.   
\begin{lemma}\label{Le_conprivacy} 
	For the \((N,K,L)\) demand-private coded caching problem with \(K=2\), the following holds for any \(i,j,j' \in [N]\):  
	\begin{align*}  
	&		\left(
	W_{\mathcal B_i},
	X_{\mathbf D^{(i,j)}},
	Z_0  
	\right)
	\sim
	\left(
	W_{\mathcal B_i},
	X_{\mathbf D^{(i,j')}},
	Z_0
	\right),  \quad  \text{and} \nonumber  \\ 
	&	\left(
	W_{\mathcal B_i},
	X_{\mathbf D^{(j,i)}},
	Z_1
	\right)
	\sim
	\left(
	W_{\mathcal B_i},
	X_{\mathbf D^{(j',i)}},
	Z_1
	\right).  
\end{align*}        
\end{lemma}   
\begin{IEEEproof} 
The lemma follows from the demand privacy constraint. The proof is a straightforward extension of \cite[Lemma~4]{Chinmay2022} from  \(L=1\) to general \(L\), and is thus omitted.
\end{IEEEproof}   
  
Following \cite{Tian2016}, we next state that, without loss of optimality, it suffices to consider schemes with user-index and file-index symmetries.  
For any permutations 
\(\bar{\pi}\in\mathrm{Perm}([K])\) and
\(\hat{\pi}\in\mathrm{Perm}([N])\), any index sets
\(\mathcal K\subseteq[K]\) and \(\mathcal N\subseteq[N]\), and any collection  of broadcast messages
\(\mathcal X\subseteq\{X_{\mathbf D}:\mathbf D\in\mathcal D\}\), define
\begin{align*}
	\bar{\pi}(\mathcal K)
	&=\{\bar{\pi}(k):k\in\mathcal K\},
	&
	\hat{\pi}(\mathcal N)
 =& \{\hat{\pi}(n):n\in\mathcal N\},   \nonumber \\
	 	\bar{\pi}(\mathcal X)
	 & = 
	 \{X_{\bar{\pi}(\mathbf D)}:X_{\mathbf D}\in\mathcal X\}, 
	&
	 \hat{\pi}(\mathcal X)
	   =&
	 \{X_{\hat{\pi}(\mathbf D)}:X_{\mathbf D}\in\mathcal X\}, 
\end{align*}  
where
\begin{align*}
	\bar{\pi}(\mathbf D)
	&=
	[d_{\bar{\pi}^{-1}(k),\,l}]_{k\in[K],\,l\in[L]},
	&
	\hat{\pi}(\mathbf D)
	&=
	[\hat{\pi}(d_{k,l})]_{k\in[K],\,l\in[L]}.
\end{align*}  
A scheme is said to have \emph{user-index symmetry} and
\emph{file-index symmetry} if, for any
\(\mathcal K\subseteq[K]\), \(\mathcal N\subseteq[N]\),
\(\mathcal X\subseteq\{X_{\mathbf D}:\mathbf D\in\mathcal D\}\),
\(\bar{\pi}\in\mathrm{Perm}([K])\), and
\(\hat{\pi}\in\mathrm{Perm}([N])\), the following equalities hold:
\begin{subequations} 
	\begin{align*} 
		H\big(W_{\mathcal N},Z_{\mathcal K},\mathcal X\big)
		&=
		H\big(W_{\mathcal N},Z_{\bar{\pi}(\mathcal K)},\bar{\pi}(\mathcal X)\big),
		\\
		H\big(W_{\mathcal N},Z_{\mathcal K},\mathcal X\big)
		&=
		H\big(W_{\hat{\pi}(\mathcal N)},Z_{\mathcal K},\hat{\pi}(\mathcal X)\big).
	\end{align*}  
\end{subequations}  
The following proposition formalizes the above claim. 
\begin{proposition}\label{symmetry}
	For the \((N,K,L)\) demand-private coded caching problem, there exists a scheme with no larger cache size and rate that has both user-index symmetry and file-index symmetry.
\end{proposition} 
 \begin{IEEEproof}
 	The proof   is a straightforward extension of the  argument used to prove \cite[Proposition~1]{Tian2016}, and is  omitted.
 \end{IEEEproof}

\subsection{Converse for $N\ge 2L$: Bounds A and B}  
 \label{sec:conv-Nge2L} 
In this subsection, we first establish two intermediate inequalities, \eqref{conA1Le_1} and \eqref{conA1Le_2}, and then provide a lemma. 
Bounds~A and~B are subsequently obtained by taking suitable linear combinations of these inequalities and applying this lemma.

We first have     
\begin{align} \label{conA1Le_1}    
	&   H  ( Z_0  ) +  H  ( X_{\mathbf D^{(0,1)}}  )     +  H \left( W_{[L]}, Z_0 ,  Z_1 \right) \nonumber \\  
	\overset{(a)}{\ge}  	&   H \left( X_{\mathbf D^{(0,1)}} , Z_0 \right)  +  H \left(    W_{[L]}, Z_0 ,  Z_1 \right) \nonumber \\     
	\overset{(b)}{=}   	&   H \left(  W_{[L]} , Z_0 ,  X_{\mathbf D^{(0,1)}}\right)  +  H \left(    W_{[L]}, Z_0 ,  Z_1 \right) \nonumber \\   
	\overset{(c)}{\ge}  	&   H \left(  W_{[L]} , Z_0,  Z_1 ,  X_{\mathbf D^{(0,1)}}\right)  +  H \left(    W_{[L]}, Z_0  \right) \nonumber \\    
	\overset{(d)}{\ge}   &  H  \left(  W_{[2L]} ,Z_0 ,  Z_1 \right) +   H\left( W_{[L]},Z_0 \right),    
\end{align}   
where \((a)\) and \((c)\) follow from the submodularity of the entropy
function, and \((b)\) and \((d)\) follow from the correctness
constraint~\eqref{correctness}.  

Next, we have 
\begin{align}   \label{conA1Le_2}    
	&  2 H\left( X_{\mathbf D^{(0,0)}}  \right) +    H\left( Z_0  \right)  + H\left( Z_1  \right)  +   H\left( W_{[L]},Z_0 \right)    \nonumber \\ 
	\ge &   H  \left( X_{	\mathbf D^{(0,0)}},  Z_0  \right) +  H  \left( X_{	\mathbf D^{(0,0)} },  Z_1  \right) +   H\left( W_{[L]},Z_0 \right)   \nonumber \\    
	\overset{(a)}{=}   	&   H  \left(W_{[L]}, X_{\mathbf D^{(0,0)}}   ,  Z_0  \right) +  H  \left( W_{[L]},  X_{\mathbf D^{(0,0)}} ,   Z_1  \right) \nonumber\\
	& \qquad  +   H\left( W_{[L]},Z_0 \right) \nonumber \\    
	\overset{(b)}{\ge}  	& H  \left( W_{[L]}, Z_0 ,  Z_1 \right) +   H\left(   W_{[L]}, X_{\mathbf D^{(0,0)}  } \right)  +   H\left( W_{[L]},Z_0 \right)    
	\nonumber \\   
	\overset{(c)}{=} 	& H  \left(    W_{[L]}, Z_0 ,  Z_1 \right) +   H\left( W_{[L]},   X_{\mathbf D^{(1,0)}}   \right) +   H\left( W_{[L]},Z_0 \right) 	\nonumber \\  
	\overset{(d)}{\ge} 	& H  \left(    W_{[L]}, Z_0 ,  Z_1 \right) +   H\left( W_{[L]},Z_0, X_{\mathbf D^{(1,0)}} \right) +   H\left( W_{[L]}  \right)   	\nonumber \\   
	\overset{(e)}{\ge}  	& H  \left(    W_{[L]}, Z_0 ,  Z_1 \right) +   H\left( W_{[2L]},Z_0 \right) +   H\left( W_{[L]}  \right) ,                  
\end{align}      
where $(a)$ and $(e)$ follow from  the correctness constraint~\eqref{correctness},   
$(b)$ and $(d)$  follow from the submodularity of the entropy function, 
and  $(c)$   follows from   the privacy constraint, more specifically, Lemma \ref{Le_conprivacy}.     

The following lemma provides a useful relation between the joint entropies corresponding to two file sets with the same cache contents.  
\begin{lemma}\label{Le_conA1} 
	For any $\mathcal{K}\subseteq \{0,1\}$ and 
	$\mathcal{V}_0\subseteq \mathcal{V}_1 \subseteq [N]$, we have 
	\begin{align*}
		H \left(W_{\mathcal{V}_1}, Z_{\mathcal{K}}\right)
		\ge
		\frac{N-|\mathcal{V}_1|}{N-|\mathcal{V}_0|}
		H \left(W_{\mathcal{V}_0}, Z_{\mathcal{K}}\right)
		+
		\frac{|\mathcal{V}_1\setminus\mathcal{V}_0|\,NF}{N-|\mathcal{V}_0|}.
	\end{align*} 
\end{lemma}

\begin{IEEEproof}
Using the chain rule of entropy, we have
 	\begin{align*}
 	& 	 H  \left( W_{\mathcal{V}_1}, Z_{\mathcal{K}}  \right)    =  H  \left( W_{\mathcal{V}_1 \setminus \mathcal{V}_0} \big|  W_{\mathcal{V}_0} ,   Z_{\mathcal{K}}  \right)  +   H  \left( W_{\mathcal{V}_0}, Z_{\mathcal{K}}  \right)  \nonumber \\
 	\overset{(a)}{\ge}      	  &  \frac{|\mathcal{V}_1 \setminus \mathcal{V}_0| }{N -  |\mathcal{V}_0 |}    H  \left( W_{[N] \setminus \mathcal{V}_0}  \big|   W_{\mathcal{V}_0} ,   Z_{\mathcal{K}}   \right)  +  H  \left( W_{\mathcal{V}_0}, Z_{\mathcal{K}}  \right)   \nonumber \\
 	= & \frac{|\mathcal{V}_1 \setminus \mathcal{V}_0| }{N -  |\mathcal{V}_0 |}    H  \left( W_{[N]}   , Z_{\mathcal{K}}   \right)   +   \frac{N - |\mathcal{V}_1| }{N -  |\mathcal{V}_0 |}     H  \left( W_{\mathcal{V}_0}, Z_{\mathcal{K}}  \right)    \nonumber \\
 	\overset{(b)}{\ge}    & \frac{|\mathcal{V}_1 \setminus \mathcal{V}_0| }{N -  |\mathcal{V}_0 |}NF   +   \frac{N - |\mathcal{V}_1| }{N -  |\mathcal{V}_0 |}     H  \left( W_{\mathcal{V}_0}, Z_{\mathcal{K}}  \right), 
 \end{align*}    
where $(a)$ follows from file-index symmetry guaranteed by  Proposition~\ref{symmetry}, 
and $(b)$ follows from the assumption that the files are independent and uniformly distributed,
each consisting of $F$ symbols. 
\end{IEEEproof} 
 
\subsubsection{Proof of Bound A}  
With \eqref{conA1Le_1}, \eqref{conA1Le_2}, and Lemma~\ref{Le_conA1} established, we now prove Bound~A. 
For the case $N \ge 2L$, taking the linear combination $(N-L)\times \eqref{conA1Le_1} + L\times \eqref{conA1Le_2}$ and canceling the common terms on both sides, we obtain      
\begin{align*} 
&   NH(Z_0)  +  L  H(Z_1)   + 2 L H( X_{\mathbf D^{(0,0)}})  \nonumber \\  
&  + (N-L) H  ( X_{\mathbf D^{(0,1)}})  +  (N-2L)  H  \left(    W_{[L]}, Z_0, Z_1 \right)   \nonumber \\  
\ge 	&  (N-L) H  \left( W_{[2L]} ,Z_0 ,  Z_1 \right) +    (N-2L) H\left( W_{[L]},Z_0 \right)   \nonumber \\  
&   +  L H\left( W_{[2L]},Z_0 \right) +   L H\left( W_{[L]}  \right)  \nonumber\\ 
\overset{(a)}{\ge} & (N-L) \left(  \frac{N-2L}{N-L}  H  \left( W_{[L]} , Z_0,Z_1  \right)  + \frac{  NLF}{N-L}    \right)  \nonumber \\ 
  &   +     (N-2L)  \left( \frac{N-L}{N} H(Z_{0})  +  LF    \right)   \nonumber \\   
&   +  L  \left(    \frac{N-2L}{N} H(Z_{0})  +  2LF   \right) +   L H\left( W_{[L]}  \right)    \nonumber \\
= &   (N-2L)     H  \left(    W_{[L]}, Z_0, Z_1 \right)    +  (N-2L) H (Z_0) \nonumber \\
& \qquad +   (2N +L)LF,   
\end{align*}      
where the first three terms of $(a)$ follow from Lemma~\ref{Le_conA1} with  
$(\mathcal{V}_1,\mathcal{V}_0,\mathcal{K})=([2L],[L],\{0,1\})$,
$([L],\emptyset,\{0\})$, and $([2L],\emptyset,\{0\})$, respectively.      
Canceling the common terms on both sides, we obtain  
\begin{align*}
	& 2LH(Z_0)  +  L  H(Z_1)    + 2 L H(X_{	\mathbf D^{(0,0)}})
	  \nonumber \\
	& + (N-L) H( X_{\mathbf D^{(0,1)}})   
	\ge  (2N +L)LF.   
	\end{align*} 
Applying the cache size and rate constraints \eqref{cons_C} and \eqref{cons_X} yields
\(3LM+(N+L)R \ge (2N+L)L\). 
The proof of bound A is thus complete.

\subsubsection{Proof of Bound B}   
For $2L \le  N  <  3L$,  taking the linear combination 
$(2N-3L)   \times \eqref{conA1Le_1} +  (3L -N)  \times \eqref{conA1Le_2} $,  and canceling the common terms on both sides, we obtain          
\begin{align*}  
	&  NH( Z_0 ) +  (3L -N)  H ( Z_1 ) + (2N-3L) H  ( X_{\mathbf D^{(0,1)}} )   \nonumber \\  
	&  +   (6L -2N)  H ( X_{\mathbf D^{(0,0)}}   )    +   (3N-6L)   H \left(    W_{[L]}, Z_0 ,  Z_1 \right)   \nonumber \\  
	\ge  & (2N-3L)   H  \left(  W_{[2L]} ,Z_0 ,  Z_1 \right) +  (3N-6L)   H\left( W_{[L]},Z_0 \right)   \nonumber\\  
	&  + (3L -N)  \left( H\left( W_{[2L]},Z_0 \right)  +  H\left( W_{[L]}  \right)  \right)   \nonumber \\        
	\overset{(a)}{\ge}    & (2N-3L)    H  \left(  W_{[2L]} ,Z_0 ,  Z_1 \right)   +  (3N-6L)   H\left( W_{[L]},Z_0 \right)    \nonumber\\
	&  + (3L -N)  \left( \frac{N-2L}{N-L}H  \left(  W_{[L]},  Z_0 \right)    +  \frac{NLF}{N-L}   + LF  \right)  \nonumber\\     
    \overset{(b)}{=}        
    &     (2N-3L)    H  \left(  W_{[2L]} ,Z_0 ,  Z_1 \right)   +       \frac{(3L -N)(2N-L)}{N-L}    LF       \nonumber\\   
    & +  \frac{ N(N-2L) }{N-L}   \left(  H\left( W_{[L]},Z_0 \right)   + H\left( W_{[L]},Z_1 \right)   \right) \nonumber\\       
    \overset{(c)}{\ge}    
    &   (2N-3L)  \left(    \frac{N-2L}{N-L}H \left(W_{[L]},Z_0,  Z_1 \right)  +  \frac{NLF}{N-L}    \right)    \nonumber \\  
    & +       \frac{(3L -N)(2N-L)}{N-L}    LF       \nonumber\\    
    & +  
    \frac{ N(N-2L) }{N-L} \left(    H \left(W_{[L]}, Z_0,Z_1  \right)  +     H \left(W_{[L]} \right)    \right)      \nonumber\\   
 =  & (N+3L)LF  +  3(N-2L) H \left( W_{[L]}, Z_0, Z_1\right), 
\end{align*}  
where  $(a)$ and the first term in $(c)$   follow  from Lemma~\ref{Le_conA1} with 
$(\mathcal{V}_1,\mathcal{V}_0,\mathcal{K}) = ([2L],[L],\{0\})$,  and $([2L],[L],\{0,1\})$,  respectively,    $(b)$  follows from Proposition \ref{symmetry}, where the user-index symmetry implies that $2 H \left(W_{[L]}, Z_0\right) = H \left(W_{[L]}, Z_0\right) + H \left(W_{[L]}, Z_1\right) $,  and the last term in $(c)$ follows from the submodularity of the entropy function.  
Thus, we obtain  
 \begin{align*} 
	&	   NH( Z_0 ) +  (3L -N)  H ( Z_1 ) + (2N-3L) H  ( X_{\mathbf D^{(0,1)}} )  \nonumber \\  
	&  +   (6L -2N)  H ( X_{\mathbf D^{(0,0)}}   )     \ge (N+3L)LF.    
\end{align*}         
Applying the cache size and rate constraints  \eqref{cons_C} and \eqref{cons_X}, it follows that $
	3M + 3R \ge N+3L  
$. 
The proof of bound B is thus complete.     

\subsection{Proof of Bound C} 
\label{sec:conv-Nle2L-C} 
For $N<2L$, we have  
\begin{align*}     
	& 3MF+3RF \nonumber \\
	\overset{(a)}{\ge}   	& H \left(Z_0, X_{\mathbf{D}^{(0,0)}} \right) + H \left(Z_1, X_{\mathbf{D}^{(0,0)}} \right) + H \left(Z_0, X_{\mathbf{D}^{(0,1)}} \right) \nonumber \\
	\overset{(b)}{=}  & H \left(W_{[L]},  Z_0, X_{\mathbf{D}^{(0,0)}}   \right) + H \left( W_{[L]}, Z_1, X_{\mathbf{D}^{(0,0)}}   \right) \nonumber \\
	& \qquad + H \left(W_{[L]}, Z_0, X_{\mathbf{D}^{(0,1)}}  \right)  \nonumber \\ 
	\overset{(c)}{\ge}   &   H \left(W_{[L]}, Z_0, Z_1, X_{\mathbf{D}^{(0,0)}}    \right) + H \left(W_{[L]}, Z_0, X_{\mathbf{D}^{(0,1)}}   \right) \nonumber \\
	&  \qquad + H \left(W_{[L]},   X_{\mathbf{D}^{(0,0)}}   \right)     \nonumber \\
	\overset{(d)}{\ge}  &   H \left(W_{[L]}, Z_1,  X_{\mathbf{D}^{(0,1)}}   \right) + H \left(W_{[L]}, Z_0 \right)      + H \left(W_{[L]}, X_{\mathbf{D}^{(1,0)}} \right)   \nonumber \\  
	\overset{(e)}{\ge}  &  H \left(W_{[L]}, Z_1,  X_{\mathbf{D}^{(0,1)}}   \right)   + H \left(W_{[L]}, Z_0, X_{\mathbf{D}^{(1,0)}} \right)      + H \left(W_{[L]} \right)     \nonumber \\
	\overset{(f)}{\ge}  &   2 H \left( W_{[N]}\right)     + H \left(W_{[L]}\right)  
	=     (2N +L) F,     
\end{align*}     
where $(a)$  follows from the cache size  and  rate constraints  \eqref{cons_C} and  \eqref{cons_X},    
$(b)$ and $(f)$ follow from the correctness constraint \eqref{correctness},  
$(c)$, the first two terms of $(d)$, and $(e)$ follow from the submodularity of the entropy function, 
and the last term of $(d)$ follows from Lemma~\ref{Le_conprivacy}.  

The proof of $3M+3R \ge 2N+L$ for $N<2L$ is thus complete.

\subsection{Proof of Bound D}  
\label{sec:conv-Nle2L-D}        
When $N=L$, Bound C established above yields $M+R\ge L$, which is Bound D at $N=L$.
Thus, it suffices to prove Bound D for $L<N<2L$.

We begin by defining a sequence of sets  $\mathcal{A}_i$, $i\in[N]$,  and a sequence $s_i$, $i\in[N]$,  as
\begin{align} \label{def_con}
	\mathcal{A}_i = [(iL+L)\bmod N], \quad
	s_i =  \bigg \lfloor \frac{(i+1)L}{N}  \bigg  \rfloor.   
\end{align}  
Based on these definitions, we first establish the following lemma, which gives a one-step entropy relation between indices $i$ and $i+1$. Summing this relation over $i$ yields the converse bound.  
\begin{lemma}\label{Le_conB1}  
	For $L < N < 2L$ and $i \in [N-1]$, with $\mathbf{D}^{(i,i+1)}$, $\mathcal{A}_i$, and $s_i$ as defined in \eqref{def_Dij} and \eqref{def_con}, the following inequality holds  
    \begin{align*}  
    & 	H\left(W_{\mathcal{A}_i},  X_{\mathbf{D}^{(i,i+1)}}  \right)  + MF +  (s_{i+1} - s_{i}) RF \nonumber \\
    \ge	&   H \left( W_{\mathcal{A}_{i+1}}, X_{\mathbf{D}^{(i+1,i+2)}}  \right) + (s_{i+1} - s_{i}) NF.   
    \end{align*}  
\end{lemma}
\begin{IEEEproof} 
 	Recall that $\mathcal B_i = \{(iL + l)\bmod N : l \in [L]\}$. For  $L < N <  2L$  and $i \in [N-1]$,  
	we have $s_{i+1}-s_i \in \{0,1\}$, and the sets $\mathcal A_i$ satisfy the recursion 
	\begin{align} \label{conB1lemma_1} 
		\mathcal{A}_{i+1}  = 
		\begin{cases} 
			\mathcal{A}_i \cup \mathcal{B}_{i+1}, & \text{if } s_{i+1} - s_i = 0  \\
			\mathcal{A}_i \cap \mathcal{B}_{i+1}, & \text{if } s_{i+1} - s_i = 1 
		\end{cases}. 
	\end{align}   	 
	The proof of Lemma \ref{Le_conB1} is divided into two cases: $s_{i+1}-s_i=0$ and $s_{i+1}-s_i=1$.

	We first consider the case $s_{i+1}-s_i=0$.  	 
	Since $s_{i+1}-s_i=0$ and $L < N<2L$, we have $s_i-s_{i-1}=1$ and $s_{i+2}-s_{i+1}=1$.
	It then follows from \eqref{conB1lemma_1} that 
	$\mathcal A_i=\mathcal A_{i-1}\cap\mathcal B_i$ and
	$\mathcal A_{i+2}=\mathcal A_{i+1}\cap\mathcal B_{i+2}$,
	and thus $ \mathcal A_i\subseteq\mathcal B_i $ and $\mathcal A_{i+2}\subseteq\mathcal B_{i+2} $. 
    Moreover,  under $s_{i+1}-s_i=0$ and $L < N < 2L$,  we have $ (iL+L)\bmod N \le (iL+3L)\bmod N$. Thus, by the definition of $\mathcal A_i$, we have  $	\mathcal A_i \subseteq \mathcal A_{i+2} \subseteq \mathcal B_{i+2}$.    
	Based on the set  relations established above, we have 
	\begin{align}  \label{conB1lemma_2} 
		&  H\left(W_{\mathcal{A}_i},  X_{\mathbf{D}^{(i,i+1)}}  \right) +  MF   
		\overset{(a)}{=} 	 	H\left(W_{\mathcal{A}_i},  X_{\mathbf{D}^{(i,i+2)}}  \right) +  MF   \nonumber \\     
		\overset{(b)}{=} 	&  	H\left(W_{\mathcal{A}_i},  X_{\mathbf{D}^{(i+1,i+2)}}  \right)  +  MF 
		\overset{(c)}{\ge}   	  H \left( W_{\mathcal{A}_{i}}, Z_0,  X_{\mathbf{D}^{(i+1,i+2)}}  \right) \nonumber \\
		\overset{(d)}{\ge} &  H \left( W_{\mathcal{A}_{i+1}},  X_{\mathbf{D}^{(i+1,i+2)}}  \right)   ,  
	\end{align}       
where $(a)$ and $(b)$ follow from Lemma \ref{Le_conprivacy}$,$ using $\mathcal A_i\subseteq\mathcal B_i$ and $\mathcal A_i\subseteq\mathcal B_{i+2}$, respectively,  
$(c)$   follows from  the cache size constraint  \eqref{cons_C}, 
and $(d)$ follows from \eqref{conB1lemma_1}  and correctness constraint  \eqref{correctness}.

We next consider the case $s_{i+1}-s_i=1$. We have 	  
	\begin{align}  \label{conB1lemma_3}  
		&  H \left(W_{\mathcal{A}_i},  X_{\mathbf{D}^{(i,i+1)}}  \right) +  MF + RF  \nonumber \\ 
	\ge	&  H \left(W_{\mathcal{A}_i}, X_{\mathbf{D}^{(i,i+1)}}  \right) +  H\left(Z_1,  X_{\mathbf{D}^{(i,i+1)}}  \right)   \nonumber \\ 
		\overset{(a)}{=} 	    & 	 	H\left(W_{\mathcal{A}_i}, X_{\mathbf{D}^{(i,i+1)}}  \right)+   H\left(W_{\mathcal{B}_{i+1}},  Z_1,  X_{\mathbf{D}^{(i,i+1)}}  \right)  \nonumber \\  
		\overset{(b)}{ \ge}  	& 	H \left( W_{[N]} \right)+  H \left( W_{\mathcal{A}_{i+1}}, X_{\mathbf{D}^{(i,i+1)}}    \right)   \nonumber \\  
		\overset{(c)}{ = } 	& 	 NF +  H \left( W_{\mathcal{A}_{i+1}}, X_{\mathbf{D}^{(i+2, i+1)}}    \right)  \nonumber \\   
		\overset{(d)}{ = } 	& 	 NF +  H \left( W_{\mathcal{A}_{i+1}}, X_{\mathbf{D}^{(i+1, i+2)}}    \right),      
	\end{align}         
where	$(a)$ follows from correctness constraint  \eqref{correctness},    
$(b)$ follows from the submodularity of the entropy function, and the facts that $\mathcal A_{i+1}=\mathcal A_i\cap\mathcal B_{i+1}$ in \eqref{conB1lemma_1} and $\mathcal A_i\cup\mathcal B_{i+1}=[N]$,  
$(c)$   follows from Lemma \ref{Le_conprivacy} using    $\mathcal{A}_{i+1} = \mathcal{A}_{i} \cap \mathcal{B}_{i+1}   \subseteq  \mathcal{B}_{i+1} $,   and    $(d)$ follows from  user-index symmetry guaranteed by  Proposition~\ref{symmetry}.  
   
	 Combining the results in \eqref{conB1lemma_2} and \eqref{conB1lemma_3} completes the proof of Lemma~\ref{Le_conB1}.            
\end{IEEEproof}  

Summing Lemma~\ref{Le_conB1} over $i\in[N-1]$ and using  $\mathcal{A}_{N-1}=\emptyset$, $s_{N-1}-s_0 = L$, we obtain     
\begin{align}  \label{conB1_1}    
	&  H \left(W_{\mathcal{A}_0},  X_{\mathbf{D}^{(0,1)}}\right)  + (N-1)MF + LRF \nonumber \\  
\ge  &     L NF   +    H \left( X_{\mathbf{D}^{(N-1,N)}}   \right) 	\overset{(a)}{=}  L NF   +    H \left( X_{\mathbf{D}^{(0,1)}}   \right),     
\end{align}     
where  $(a)$ follows from Lemma~\ref{Le_conprivacy}, which implies that $ H \left( X_{\mathbf{D}^{(N-1,N)}}   \right) = H \left( X_{\mathbf{D}^{(N-1,1)}}   \right) = H \left( X_{\mathbf{D}^{(0,1)}}   \right) $.   

Moreover, we have    
\begin{align}\label{conB1_2}    
	&   MF  + H \left( X_{\mathbf{D}^{(0,1)}}   \right)   	\ge   H \left( Z_0, X_{\mathbf{D}^{(0,1)}}   \right)           \nonumber \\ 
\overset{(a)}{	=}   &  H \left(W_{\mathcal{B}_0}, Z_0,   X_{\mathbf{D}^{(0,1)}}\right)     
	\overset{(b)}{\ge}    H \left(W_{\mathcal{A}_0},  X_{\mathbf{D}^{(0,1)}}\right)  , 
\end{align}    
where $(a)$ follows from correctness constraint \eqref{correctness}, and $(b)$ follows from  the fact that  $\mathcal A_0 = \mathcal B_0 = [L]$.   

Adding \eqref{conB1_1} and \eqref{conB1_2}, we have     
$
  NMF + LRF  \ge L NF      
$.  
The proof  $NM + LR \ge LN$   for $N < 2L$  is thus complete.

\section{Conclusions} \label{sec_Conclu}      
In this paper, we studied the $(N,K,L)$ demand-private coded caching
problem.
For arbitrary $N$, $K$, and $L$, we established a transformation
from suitable non-private coded caching schemes to demand-private
coded caching schemes with multiple demands.
This transformation yielded a new achievable scheme, which, together
with a new converse bound, improved the characterization of the optimal
memory-rate tradeoff to within a multiplicative factor of $6$ for all
$M\in[0,N]$.
For the case of $K=2$, we established new achievability and converse
results that together completely characterize the optimal memory-rate
tradeoff for all $N$ and $L$.
Determining the optimal memory-rate tradeoff for the case of $K>2$
remains an interesting direction for future work.

\appendices  
\section{Proof of Theorem \ref{th_converse}}  \label{sec_appconverse}

We follow the proof of \cite[Theorem 2]{Qian2019} and   adapt it to the case where each user requests \(L\) files.  

Since  $	\bar N  =  \min\{N,KL\}$,  we have    
$  \lfloor \bar{N}/L \rfloor   =  \min\{ \lfloor N/L \rfloor, K  \}     $.       
For  \(i \in [ 1: \lfloor \bar{N}/L  \rfloor-1]\), and any demand matrix $\mathbf{D}$, we have
\begin{align*} 
	&H \left( X_{\mathbf D} \big| Z_{[i]},W_{\{d_{j,l}:j\in[i-1], l\in[L]\}} \right) 
	\nonumber\\  
	\overset{(a)}{=}    	
	&  H \left( 	W_{\{d_{i-1,l}:l \in[L]\}},  X_{\mathbf D} \big| Z_{[i]},\,W_{\{d_{j,l}:j\in[i-1], l\in[L]\}} \right) 
	\nonumber\\   
	= 	& 
	H \left(
	W_{\{d_{i-1,l}: l \in[L]\}} \big|   
	Z_{[i]},W_{\{d_{j,l}:j\in[i-1],l\in[L]\}} 
	\right)
	\nonumber\\
	&\quad+
	H\!\left(
	X_{\mathbf D} 
	\middle|
	Z_{[i]}, W_{\{d_{j,l}:j\in[i],l \in[L]\}}
	\right) 	\nonumber\\    
	\ge	& 
	H \left(
	W_{\{d_{i-1,l}: l \in[L]\}} \big|   
	Z_{[i]},W_{\{d_{j,l}:j\in[i-1],l\in[L]\}} 
	\right)
	\nonumber\\
	&\quad+
	H \left(
	X_{\mathbf D}
	\middle|
	Z_{[i+1]}, W_{\{d_{j,l}:j\in[i],l \in[L]\}}
	\right) ,  
\end{align*}
where $(a)$ follows from the correctness constraint \eqref{correctness}.   
Summing the above inequalities over   \(i \in [ 1: \lfloor \bar{N}/L  \rfloor-1]  \), we obtain    
\begin{align} 	\label{pfT2_lemma2}    
	&	RF \ge H(X_{\mathbf D}) \ge  H \left( X_{\mathbf D} \big| Z_0   \right)  \nonumber \\  
	\ge 	& \sum_{i = 1}^{\lfloor \bar{N}/L -1 \rfloor  }   	H \left(
	W_{\{d_{i-1,l}: l \in[L]\}} \big|   
	Z_{[i]},W_{\{d_{j,l}:j\in[i-1],l\in[L]\}} 
	\right) \nonumber \\
	& +
	H \left(
	X_{\mathbf D}
	\middle| 
	Z_{[\lfloor \bar{N}/L \rfloor]}, W_{\{d_{j,l}:j\in[\lfloor \bar{N}/L \rfloor-1],l \in[L]\}}
	\right)  \nonumber \\ 
	\overset{(a)}{\ge}    		&	\sum_{i = 1}^{\lfloor \bar{N}/L\rfloor  }   	H \left(
	W_{\{d_{i-1,l}: l \in[L]\}} \big|   
	Z_{[i]},W_{\{d_{j,l}:j\in[i-1],l\in[L]\}} 
	\right),  
\end{align}
where $(a)$ follows from the correctness constraint \eqref{correctness}.     
The bound in \eqref{pfT2_lemma2} extends \cite[Lemma~2]{Qian2019} to the case where each user requests \(L\) files and the exact decoding constraint is imposed. 

Next, we use the homogeneity of the problem to derive a symmetrized version of the bound in \eqref{pfT2_lemma2}. 
For any permutations $\bar{\pi}\in \mathrm{Perm}([K])$
and $\hat{\pi}\in \mathrm{Perm}([N])$, and  any index sets $\mathcal K\subseteq [K]$ and $\mathcal N\subseteq [N]$,  recall that 
$ \bar{\pi}(\mathcal K) =\{\bar{\pi}(k):k\in\mathcal K\} $ and  $ \hat{\pi}(\mathcal N) =\{\hat{\pi}(n):n\in\mathcal N\} $,  and define  
\begin{align*}  
	&	H^*(W_{\mathcal N},Z_{\mathcal K})
	=
	\frac{1}{N!\,K!}
	\sum_{\substack{\hat{\pi}\in \mathrm{Perm}([N])\\ \bar{\pi}\in \mathrm{Perm}([K])}} 
	H\!\left(W_{\hat{\pi}(\mathcal N)},Z_{\bar{\pi}(\mathcal K)}\right).  
\end{align*}   
Similarly, conditional entropy is defined using the same notation.
As shown in \cite{Qian2019}, 
\(H^*(\cdot)\) satisfies all Shannon's inequalities.

For notational convenience, let $\Phi =  ( \phi_0,\phi_1, \ldots,\phi_{K-1})$ denote the collection of caching functions.  
Following the same argument as in \cite[(29)--(32)]{Qian2019}, we consider
the class of demand matrices in which at least
\(\bar N=\min\{N,KL\}\) files are requested, and average the bound in
\eqref{pfT2_lemma2} over all file-index and user-index permutations.
It follows that  
\begin{align}	\label{eq:R-RA} 
	&	R \ge \inf_{F\in\mathbb N_+}\min_{\Phi} R(F,\Phi), \quad  \text{where}   \\ 
	&	R(F,\Phi) = 
	\frac{1}{F}\sum_{i=1}^{\lfloor \bar N/L \rfloor}
	H^*\!\left(
	W_{[(i-1)L:iL-1]}
	\middle| 
	Z_{[i]},\,W_{[(i-1)L]} \nonumber 
	\right).   
\end{align}    

It remains to derive a lower bound on \(R(F,\Phi)\), given in the following lemma.   
\begin{lemma}[Extension of {\cite[Lemma 3]{Qian2019}} ] 
	\label{pfT2_lemma3}       
	For any $s \in [1:\lfloor \bar{N}/L \rfloor]$,  $\lambda \in [0,1]$, and any prefetching scheme \(\Phi\), we have 
	\begin{align*} 
		&	R(F,\Phi )   \ge  (s-1+\lambda ) L - \frac{L\big(2 \lambda s+s(s-1) - t(t-1)\big)}{2 (N-L(t-1))} M,  
	\end{align*} 
	where $t \in [1:s] $ is the minimum value such that   
	\begin{align*}  
		L 	 \left(  s(s-1)-t(t-1)  +   2 \lambda s \right) \le  2 ( N-(t-1)L)t.   
	\end{align*}    
\end{lemma}    
\begin{IEEEproof} 
	We prove the lemma by extending the proof of \cite[Lemma 3]{Qian2019} to the case where each user requests multiple files, following a similar step-by-step argument with suitable adjustments to the coefficients. 
	
	Following from the definition of $R(F,\Phi)$ and the non-negativity of entropy,  we have
	\begin{align}\label{pfT2l3_76}
		R(F,\Phi)F  & 
		\ge 
		\sum_{i = 1}^{s-1} 
		H^* \left(  W_{[(i-1)L:iL-1]} \big| Z_{[i]},   W_{[(i-1)L]}\right) 
		\nonumber \\	&  +
		\lambda  H^* \left(  W_{[(s-1)L:sL-1]} \big| Z_{[s]},   W_{[(s-1)L]}\right). 
	\end{align}
	We then bound each term above in the following $2$ ways:      
	\begin{align}\label{pfT2l3type1}       
		& H^* \left(  W_{[(i-1)L:iL-1]} \big| Z_{[i]},   W_{[(i-1)L]}\right) \nonumber \\
		\ge  &   \frac{L H^* \left(  W_{[(i-1)L:N-1]} \big| Z_{[i]},   W_{[(i-1)L]}\right)  }{N-(i-1)L} \nonumber \\
		\ge  &  LF - \frac{L H^* \left(  Z_{[i]}\big| W_{[(i-1)L]}\right)}{N-(i-1)L},   \\
		\label{pfT2l3type2}        
		& H^* \left(  W_{[(i-1)L:iL-1]} \big| Z_{[i]},   W_{[(i-1)L]}\right) \nonumber \\
		=  &   LF- H^* \left(  Z_{[i]}| W_{[(i-1)L]}\right) +  H^* \left(  Z_{[i]}| W_{[iL]}\right)    \nonumber \\ 
		\ge 	&  LF - H^* \left(  Z_{[i]}\big| W_{[(i-1)L]}\right)   + \frac{i}{i+1}H^* \left( Z_{[i+1]} \big| W_{[iL]} \right),    
	\end{align}    
	where $i \in [1:s]  $ in  \eqref{pfT2l3type1}  and $i \in [1:s-1]  $ in \eqref{pfT2l3type2}.
	
	Next, for $x\in  [1:s-1]$, we define  
	\begin{align*}
		a_x & = \frac{ L (2\lambda s+s(s-1) -x(x+1) )}{2x (N-xL)},   \\
		b_x & = \frac{ L (2\lambda s+s(s-1) -x(x-1) )}{2x (N-L(x-1))},       
	\end{align*}   
	which satisfy 
	\begin{align}
		\frac{(1-a_x)L}{N-(x-1)L}+a_x=b_x,
		\qquad
		\frac{x}{x+1}a_x=b_{x+1}.
		\label{pfT2l3_8182}
	\end{align}    
	Since $t$ is the minimum value such that    \eqref{converse_condition} holds,  we have $a_x\in[0,1]$ for   $x\in [t:s-1]$ and  $b_t\ge \frac{t-1}{t}$.  
	Armed with the above properties,  we now bound $R(F,\Phi)$ as follows.

	For each $x\in [t:s-1]$, by taking 
	$(1-a_x)\times\eqref{pfT2l3type1}+a_x\times\eqref{pfT2l3type2}$ and applying  \eqref{pfT2l3_8182}, we have            
	\begin{align}\label{pfT2l3_83}
		& 	H^* \left(  W_{[(x-1)L:xL-1]} \big| Z_{[x]},   W_{[(x-1)L]}\right)   \nonumber \\ 
		\ge &    LF - b_x H^* \left(  Z_{[x]}\big| W_{[(x-1)L]}\right)  +  b_{x+1}  H^* \left( Z_{[x+1]}\big| W_{[xL]} \right).     
	\end{align}     
	For $x = s$, following from \eqref{pfT2l3type1},  we have   
	\begin{align}\label{pfT2l3_84} 
		&\lambda H^* \left( 
		W_{[(s-1)L:sL-1]} \big|  Z_{[s]}, W_{[(s-1)L]}
		\right)
		\nonumber\\
		\ge	&
		\lambda LF-b_sH^* \left(Z_{[s]}  \big|  W_{[(s-1)L]}\right). 
	\end{align}  
	Summing \eqref{pfT2l3_83} over $x\in [t:s-1]$ and adding \eqref{pfT2l3_84}, we have  
	\begin{align}\label{pfT2l3_85}    
		&   \sum_{i = t}^{ s -1 }  H^*\left( W_{[(i-1)L: iL-1]} \big| Z_{[i]},  W_{[(i-1)L]}   \right) \nonumber \\
		& \quad +    \lambda  H^* \left(  W_{[(s-1)L:sL-1]} \big| Z_{[s]},   W_{[(s-1)L]}\right)  \nonumber \\
		& \ge (s-t+ \lambda )LF - b_t H^* \left(  Z_{[t]}\big| W_{[(t-1)L]}\right).  
	\end{align} 
	On the other hand,  following from \eqref{pfT2l3type2}, we have  
	\begin{align}\label{pfT2l3_86}    
		&   \sum_{i = 1}^{ t -1 }  H^*\left( W_{[(i-1)L: iL-1]} \big| Z_{[i]},  W_{[(i-1)L]}   \right)  \nonumber \\ 
		\ge &    \sum_{i = 1}^{ t -1 } \big(  LF - \frac{1}{i} H^* \left(  Z_{[i]}\big| W_{[(i-1)L]}\right) \big)  \nonumber \\
		& + \frac{t-1}{t}H^* \left( Z_{[t]} \big| W_{[(t-1)L]} \right)  \nonumber \\ 
		\ge   & (t-1)LF - (t-1) MF + \frac{t-1}{t}H^* \left( Z_{[t]} \big| W_{[(t-1)L]} \right).  
	\end{align}         
	
	Using \eqref{pfT2l3_76} and adding \eqref{pfT2l3_85} and \eqref{pfT2l3_86}, we have  
	\begin{align*} 
		&	R(F,  \Phi )F    
		\ge   (s -1 + \lambda )LF  - (t-1) MF  \nonumber \\
		& + \left( \frac{t-1}{t} -   b_t  \right)H^* \left( Z_{[t]} \big| W_{[(t-1)L]} \right)  \nonumber \\ 
		\overset{(a)}{\ge}   	&  (s -1 + \lambda )LF  - (t-1) MF    -  \left( b_t - \frac{t-1}{t}   \right) t M F \nonumber \\
		= & (s-1+\lambda ) LF     - \frac{L\big(2 \lambda s+s(s-1) - t(t-1)\big)}{2 (N-L(t-1))} MF.    
	\end{align*}  
	where $(a)$ follows from the fact that $b_t \ge  \frac{t-1}{t} $. Dividing both sides by $F$ completes the proof.    
\end{IEEEproof}

Combining Lemma~\ref{pfT2_lemma3} with \eqref{eq:R-RA}, we obtain the desired result in Theorem~\ref{th_converse}. The proof of Theorem~\ref{th_converse} is thus complete.

\section{Proof of Theorem \ref{orderoptimal}} \label{sec_apporderoptimal}

We first derive a lower convex envelope of the converse bound in Theorem~\ref{th_converse} by following arguments similar to those in \cite{Qian2019}. We then compare the achievable memory-rate pair with the resulting lower bound and complete the proof of Theorem~\ref{orderoptimal}.    

Recall that $  \lfloor \bar{N}/L \rfloor   =  \min\{ \lfloor N/L \rfloor, K  \}     $. 
For \(s\in[1:\lfloor \bar{N}/L \rfloor]\)  and \(t\in[1:s]\), let 
\begin{align*}
	(M_{s,t},R_{s,t})
	=
	\left(
	\frac{N-L(t-1)}{s},
	\,L\Big(\frac{s-1}{2}+\frac{t(t-1)}{2s}\Big)
	\right). 
\end{align*} 
We first show that the memory-rate pair  $\bigl(M,R^{*\mathrm p}_{N,K,L}(M)\bigr) $  is lower bounded by the lower convex envelope of $ \mathcal{S}_{\mathrm{Lower}} 
\cup
\left\{\bigl(0, L  \lfloor \bar{N}/L \rfloor \bigr)\right\} $, where $\mathcal{S}_{\mathrm{Lower}}  = \bigl\{(M_{s,t},R_{s,t}):  	s\in [1:\lfloor \bar{N}/L \rfloor],  t \in [1:s] \bigr\}$.

To prove the above statement, following arguments similar to those in the proof of inequality~(7) in \cite{Qian2019}, it suffices to show that any linear function, denoted by \(R=L(A+BM)\), that lower bounds all points in 
$\mathcal{S}_{\mathrm{Lower}}\cup\left\{\bigl(0, L\lfloor \bar N/L \rfloor\bigr)\right\}$
also lower bounds \(\bigl(M,R^{*\mathrm p}_{N,K,L}(M)\bigr)\).

If \(A >  0\), since \( (0, L\lfloor \bar N/L\rfloor )\) must be lower bounded by the line, we have  \(A\le  \lfloor \bar{N}/L \rfloor  \).   
Choose  \(s=\lceil A\rceil\), \(\lambda=A-s+1\), and let \(t\) be the minimum value in \([1:s]\) such that \eqref{converse_condition} holds. Such a \(t\) always exists, since \eqref{converse_condition} holds for \(t=s\).  
Since \((M_{s,t},R_{s,t})\in \mathcal{S}_{\mathrm{Lower}}\), we have
\begin{align*}
	L \left(  A+B\frac{N-L(t-1)}{s} \right)
	\le
	L \left(\frac{s-1}{2}+\frac{t(t-1)}{2s}   \right).
\end{align*}  
Then, \(B\) can be upper bounded as 
\begin{align*}
	B \le  & 
	\frac{s(s-1)+t(t-1)-2As}{2\bigl(N-L(t-1)\bigr)} \nonumber \\
	\overset{(a)}{=}  & -\frac{s(s-1)-t(t-1)+2\lambda s}{2(N-L(t-1))}, 
\end{align*}  
where $(a)$ follows from \(A=s-1+\lambda\). 
Thus,  for \(M\ge 0\), we have  
\begin{align*} 
	& 	L(A+BM)   \nonumber \\
	\le &
	(s-1+\lambda)L
	-\frac{L\bigl(2\lambda s+s(s-1)-t(t-1)\bigr)}{2(N-L(t-1))}M  \nonumber \\
	\overset{(a)}{	\le }  & R^{*\mathrm p}_{N,K,L}(M),  
\end{align*} 
where  $(a)$ follows from  Theorem~\ref{th_converse}.

If \(A\le 0\), then \((M_{1,1},R_{1,1})=(N,0)\in\mathcal S_{\mathrm{Lower}}\) implies \(A+BN\le 0\). 
Thus, for  \(M\in[0,N]\),  we have   
\begin{align*}
	A+BM 
	=
	A\frac{N-M}{N}+(A+BN)\frac{M}{N}
	\le 0, 
\end{align*}   
which implies \(L(A+BM)\le 0\le R^{*\mathrm p}_{N,K,L}(M)\).  

From the above two cases, the proof that
$
\bigl(M,R^{*\mathrm p}_{N,K,L}(M)\bigr)
$
is lower bounded by the lower convex envelope of
$ \mathcal{S}_{\mathrm{Lower}}
\cup
\left\{\bigl(0, L  \lfloor \bar{N}/L \rfloor \bigr)\right\} $ 
is complete.

By memory sharing, the achievable memory-rate pair is convex in \(M\).  
Thus, it suffices to compare it with the above lower convex envelope at its corner points, which form a subset of  $\mathcal{S}_{\mathrm{Lower}}
\cup 
\left\{\bigl(0,\,L\lfloor \bar N/L\rfloor\bigr)\right\} $.

We first consider the point $\bigl(0,\,L\lfloor \bar N/L\rfloor\bigr)$.  When \(N \ge KL\), we have  \( R^{\mathrm{p}}_{N,K,L}(0)= KL
=   L\lfloor \bar N/L\rfloor  \).   
When \(N < KL\), we have $ R^{\mathrm{p}}_{N,K,L}(0) = N  \le 2 L\lfloor \bar N /L\rfloor $.    
Thus, it remains to consider the points in \(\mathcal S_{\mathrm{Lower}}\).  
Next, we consider the point $(M_{1,1},R_{1,1})=(N,0) $,  we have $
R^{\mathrm p}_{N,K,L}(M_{1,1})=0= R_{1,1}
$.   
Thus, it remains to show that
\begin{align*}
	& R^{\mathrm p}_{N,K,L}(M_{s,t}) \le 6R_{s,t}, \nonumber \\
	& \qquad  \forall  s\in[1:  \lfloor \bar N/L\rfloor  ],  t\in[1:s],  (s,t)\neq(1,1).
\end{align*}   
We now consider the following two cases.

\subsubsection{For the case $N \le 3sL$}   
From Theorem \ref{theorem_achA}, we note that $(0, \bar{N})$ and  $(N, 0)$ are achievable. Thus, by memory sharing, we have 
\begin{align*} 
	R^{\mathrm p}_{N,K,L}(M_{s,t}) \le 
	\frac{\bar N}{N}(N-M_{s,t}).
\end{align*} 
Dividing both sides  by \(R_{s,t}\), we have  
\begin{align*}
	&	\frac{R^{\mathrm p}_{N,K,L}(M_{s,t})}{R_{s,t}}
	\le 
	\frac{2\bar N\left((s-1)N+L(t-1)\right)}
	{NL\left(s(s-1)+t(t-1)\right)}  \\ 
	\overset{(a)}{\le}   & \frac{2 \left((s-1)N+L(t-1)\right)}
	{ L\left(s(s-1)+t(t-1)\right)} \nonumber \\[0.6ex]
	\overset{(b)}{\le}    & 	\frac{6s(s-1)+2(t-1)}{s(s-1)+t(t-1)}  \overset{(c)}{\le} 6, 
\end{align*}  
where $(a)$ follows from $ \bar N \le N$, 
$(b)$ follows from  $N \le 3sL$,  
and $(c)$ follows from  $  t\ge 1$.

\subsubsection{For the case \(N \ge 3sL\)}
Let $ r  =  \lfloor     \frac{ K\bar{N}(N - L(t-1))}{NLs}    \rfloor  $. Then, we have    
\begin{align*}
	\frac{\binom{K\bar N}{r}-\binom{K\bar N-L}{r}}{\binom{K\bar N}{r}}\,N
	&\le
	\frac{NLr}{K\bar N}\le 
	\frac{N-L(t-1)}{s}
	=   
	M_{s,t}, 
\end{align*}
which shows that the corresponding memory size in Theorem~\ref{theorem_achA} is no larger than \(M_{s,t}\).  
It follows from Theorem~\ref{theorem_achA} that 
\begin{align}
	R^{\mathrm p}_{N,K,L}(M_{s,t}) 	  \le 	&   \frac{\binom{K \bar{N}}{r+1} -  \binom{(K-1) \bar{N}}{r+1}}{\binom{K \bar{N}}{r}  }  
	\le \frac{K\bar N-r}{r+1} \nonumber \\
	\le &   
	\frac{K\bar N}{r+1} \le 
	\frac{NLs}{N-L(t-1)}, 
	\label{order_step2_1} 
\end{align}  
where the last step follows from  $ r + 1  =  \lfloor     \frac{ K\bar{N}(N - L(t-1))}{NLs}    \rfloor  + 1   \ge
\frac{K\bar N (N-L(t-1) )}{NLs}$.   

For the case \((s,t)\neq(1,1)\), dividing \eqref{order_step2_1} by \(R_{s,t}\) yields  
\begin{align*} 
	\frac{ 	R^{\mathrm{p}}_{N,K,L}\left(M_{s,t}  \right) } { R_{s,t}}  
	\le &  \frac{  2N s^2 }  {	  (N - L(t-1))(s(s-1) + t(t-1))  }   \nonumber \\  
	\overset{(a)}{\le} &  \frac{  6s^3 }  {	  (3s -  (t-1))(s(s-1) + t(t-1))  }   \nonumber \\   
	\overset{(b)}{\le} &  \frac{  6s^2 }  {	(2s +1)  (s-1)   }   \overset{(c)}{\le} 4.8,  
\end{align*}     
where $(a)$ follows from \(N\ge 3sL\) and the fact that \(\frac{N}{N-L(t-1)}\) is decreasing in \(N\), $(b)$ follows from \(1\le t\le s\), which implies \(3s-(t-1)\ge 2s+1\) and \(t(t-1)\ge 0\), and $(c)$ follows from \(s\ge 2\).  

Combining the above two cases with the discussion for \((M_{1,1},R_{1,1})\) and \(\bigl(0,\,L\lfloor \bar N/L\rfloor\bigr)\), we complete the proof of Theorem~\ref{orderoptimal}.

 \section{Correctness of \eqref{achB_pc4}}   	\label{appendix_pc3correctness}    
 
In this appendix, we provide a detailed proof of \eqref{achB_pc4}.   
	We begin by defining a sequence of file indices $ (n_{k,0}, n_{k,1}, \dots, n_{k,N-\bar N+L-1})$, such that   
	\begin{align*}  
		P_{n_{k,j}} = \theta_{k,j}, \quad   k \in \{0,1\}, \; j \in   [N-\bar N+L]. 
	\end{align*}   
From the definitions of $\mathbf{H}$ and $\mathbf{G}_k$ in \eqref{achB_C2} and \eqref{achB_pcg1}, respectively, we observe that 
	\begin{align}   \label{apppc3_2}     
		&  g_{k,i,j}=H_{i,n_{k,j}}, \quad i,j \in  [N-\bar N+L].  
	\end{align}   
	Furthermore, following from \eqref{achB_DJ3}, \eqref{achB_DJ4}   and \eqref{achB_pcg2}, we note that $(n_{k,j})_{j\in[N-\bar N+L]}$  is a permutation of the set $  \mathcal{N}_k  \cup   
	\bigl(  [N] \setminus \{ Q_i : i\in[\bar N]\} \bigr) 
	$, and it satisfies   
	\begin{align}   \label{apppc3_3}    
		&	n_{k,l} =  d_{k,l},  \quad   l \in [L]. 
	\end{align}    
	
	Next, for $i \in [N-\bar N+L]$, considering $\widetilde{Z}_{k,i}$ as defined in \eqref{achB_pc3}, we have 
	\begin{align}  \label{apppc3_4}    
		\widetilde{Z}_{k,i} \overset{(a)}{=}   &   
		Z_{k,i} - \sum_{ j \in [\bar{N}] \setminus \{ T_{k,l}: l \in [L] \} }   
		\frac{1}{N+i-P_{Q_j} }  F_{Q_j,k}  \nonumber \\   
		\overset{(b)}{=}  	 & 
		\sum_{j  \in  \mathcal{N}_k   }   H_{i,j} F_{j,k}    
		+  \sum_{j  \in  [N] \setminus \{ Q_r :  r \in [\bar{N}]\}   }   
		H_{i,j} F_{j,k}  \nonumber \\
		\overset{(c)}{=}  &  
		\sum_{j  \in   [N-\bar N+L]    }  g_{k,i,j} F_{n_{k,j},k},
	\end{align}
	where $(a)$ follows from \eqref{achB_DX}, \eqref{achB_DY}, \eqref{achB_DJ4},
	and the fact that
	$Q_j\notin\mathcal N_k$ for
	$j\in[\bar N]\setminus\{T_{k,l}:l\in[L]\}$, 
	$(b)$ follows from  \eqref{achB_C1}, \eqref{achB_C2}, and  the same fact on $Q_j$ as in $(a)$,  
and $(c)$  follows from \eqref{apppc3_2} and the fact that $(n_{k,j})_{j\in[N-\bar N+L]}$  is a permutation of the set $  \mathcal{N}_k  \cup   \bigl(  [N] \setminus \{ Q_r  :  r\in[\bar N]\} \bigr) 
$.

Using the result of \eqref{apppc3_2} and the fact that the indices $(n_{k,j})_{j\in[N-\bar N+L]}$  are distinct, $ \mathbf{G}_k$  is a column-permuted square Cauchy submatrix of $\mathbf{H}$, and is therefore invertible.   
 
Finally, we have  
\begin{align*}   
	& \left( \mathbf G_k^{-1}
	\left[ \big( \widetilde{Z}_{k,i} \big)_{i\in[N-\bar N+L]} \right]^{\mathsf T}
	\right)^{\mathsf T}
	\mathbf e_l^{N-\bar N+L}  \nonumber \\  
	\overset{(a)}{=}     &  \left(  \mathbf{G}_k ^{-1}   \mathbf{G}_k  \Big[ \big(  F_{n_{k,j},k} \big)_{j \in [N -  \bar{N} + L]} \Big] ^{\mathsf{T}}    \right)^{\mathsf T} \mathbf{e}^{N-\bar N + L}_l  \nonumber   \\
	= &  \Big[ \big(  F_{n_{k,j},k} \big)_{j \in [N -  \bar{N} + L]} \Big] \mathbf{e}^{N-\bar N + L}_l  \overset{(b)}{=}  F_{d_{k,l},k},   
\end{align*}    
where   $(a)$ follows from  \eqref{apppc3_4}, and $(b)$  follows from \eqref{apppc3_3}.  

The proof of \eqref{achB_pc4} is thus complete.

  \section{Correctness of \eqref{achB_pc4}}  \label{app_e12}    
We now provide a detailed proof of \eqref{achC_pc_e12}.
First, for \(k\in\{0,1\}\) and \(i\in[3L]\), we have
\begin{align*}  
	\widetilde{X}_{k,i}
	\overset{(a)}{=}
	& \sum_{j \in [2L]} \sum_{m \in \{0,1\}}   v_{i,2j+m} Y_{j,m}  \nonumber \\
	& \qquad -   \sum_{ j \in [2L] \setminus \{ T_{k,l}: l \in [L] \} }    v_{i,2j+k}   Z_{k, P_{Q_j}}  \nonumber \\ 
	\overset{(b)}{=}
	&  \sum_{j \in \{ T_{k,l} : l \in [L] \} }     v_{i,2j+k} Y_{j,k}
	+  \sum_{j \in [2L]}    v_{i,2j+1-k} Y_{j,1-k}  \nonumber \\
	\overset{(c)}{=}
	&  \sum_{  l \in [L] }     v_{i,2T_{k,l}+k} F_{d_{k,l},k+2}
	+  \sum_{j \in [2L]}    v_{i,2j+1-k} Y_{j,1-k}  \nonumber \\
	\overset{(d)}{=}
	&  \sum_{  l \in [L] }     g_{k,i,l} F_{d_{k,l},k+2}
	+  \sum_{j \in [2L]}    g_{k,i,j+L} Y_{j,1-k}, 
\end{align*}
where $(a)$ follows from substituting \eqref{achB_DJ4} and \eqref{achC_DX} into \eqref{achC_pc1}, 
$(b)$ follows from \eqref{achB_DJ3}, \eqref{achC_c1}, and \eqref{achC_refDY}, which together imply that, for any \(j\in[2L]\setminus\{T_{k,l}: l\in[L]\}\), \(Q_j\notin\mathcal{N}_k\) and  \(Z_{k,P_{Q_j}}=F_{Q_j,k}=Y_{j,k}\),
$(c)$  follows from \(Q_{T_{k,l}}=d_{k,l}\)   and \eqref{achC_refDY},  which imply that  $Y_{T_{k,l},k} =   F_{d_{k,l},k+2}$,
and $(d)$  follows from  the definition of $ \mathbf{G}_k$ in \eqref{achC_gk}.

Thus, for the case \(m=k\), we have
\begin{align*}
	&\left(
	\mathbf G_k^{-1}
	\left[ \big( \widetilde X_{k,i} \big)_{i\in[3L]} \right]^{\mathsf T}
	\right)^{\mathsf T}\mathbf e_l^{3L} \\
	=&
	\left(
	\mathbf G_k^{-1}\mathbf G_k
	\left[
	\big(F_{d_{k,r},k+2}\big)_{r\in[L]},
	(Y_{j,1-k})_{j\in[2L]}
	\right]^{\mathsf T}
	\right)^{\mathsf T}
	\mathbf e_l^{3L} \\
	=&
	F_{d_{k,l},k+2}
	=
	\left(\mathbf e_m^3\right)^{\mathsf T}
	\mathbf U_{k,l}\mathbf F_{d_{k,l}}, 
\end{align*}
where the last step  follows from the definition of  \(\mathbf U_{k,l}\)
in \eqref{achB_DU}.

Similarly,   for the case \(m\neq k\), i.e., $m = 1-k$,  we have
\begin{align*}
	&   \left( \mathbf G_k^{-1}
	\left[ \big( \widetilde{X}_{k,i} \big)_{i\in[3L]} \right]^{\mathsf T}
	\right)^{\mathsf T}
	\mathbf e_{L+T_{k,l}}^{3L}  \nonumber \\
	= 
	&   \left[  \big( F_{d_{k,r},k+2}\big)_{r \in[L]} ,  (Y_{j,1-k})_{j \in [2L]}     \right]   \mathbf e_{L+T_{k,l}}^{3L}
	=       Y_{T_{k,l},1-k}  \nonumber \\
	\overset{(a)}{=}
	&   \begin{cases}
		F_{d_{k,l},3-k}, & \text{if } d_{k,l} \in \mathcal{N}_{1-k} \\
		F_{d_{k,l},1-k}, & \text{if } d_{k,l} \notin \mathcal{N}_{1-k}
	\end{cases}
	\overset{(b)}{=}    \left( \mathbf{e}_{m}^{3} \right)^{\mathsf{T}}  \mathbf{U}_{k,l}  \mathbf{F}_{d_{k,l}},
\end{align*}
where  $(a)$  follows from \(Q_{T_{k,l}}=d_{k,l}\) and \eqref{achC_refDY}, and  $(b)$ follows from   the definition of  \(\mathbf U_{k,l}\)
in \eqref{achB_DU}. 
 
The proof of \eqref{achC_pc_e12} is thus complete.  
  
 	
 	\bibliographystyle{IEEEtran}
 	\bibliography{ref}

	\end{document}